\documentclass[twoside,11pt]{report}
\usepackage[utf8]{inputenc}
\usepackage{graphicx}
\graphicspath{{./lq_images/}}
\usepackage[english]{babel}
\usepackage[parfill]{parskip}
\usepackage{amsmath,amssymb,amsthm,amscd}
\usepackage{algorithm}% http://ctan.org/pkg/algorithms
\usepackage{algpseudocode}% http://ctan.org/pkg/algorithmicx
\usepackage{a4wide}% Replace with [a4paper] class?
\usepackage{tabu}
\usepackage{palatino}
\usepackage{pdfpages}
\usepackage[cdot,squaren]{SIunits}
\usepackage{url}
\usepackage{fancyhdr}
\usepackage[toc]{appendix}
\usepackage{tikz}
\usetikzlibrary{shapes}
\usepackage{caption}
\usepackage{color}
\usepackage[ddmmyyyy,nodayofweek]{datetime}
\usepackage[subnum]{cases}
\usepackage{enumerate}
\title{Large-scale Multiscale Particle Models in Inhomogeneous Domains: Modelling and Implementation}
\author{Omar Richardson}

\let\oldtil\tilde
\renewcommand{\tilde}[1]{\oldtil{\mathbf{#1}}}
\renewcommand{\vec}[1]{\mathbf{#1}}
\newcommand{\gvec}[1]{\boldsymbol#1}
\DeclareMathOperator{\diag}{diag}
\newtheorem{newdef}{Definition}
\newtheorem{newthm}{Theorem}
\newtheorem{newlemma}{Lemma}
\DeclareMathOperator*{\argmin}{arg\,min}
\newcommand{\bigo}[1]{\mathcal{O}\left(#1\right)}

\DeclareMathOperator{\D}{D}
\newcommand{\ceil}[1]{\left\lceil#1\right\rceil}
\renewcommand{\div}[1]{\operatorname{div}\left( #1 \right)}
\newcommand{\eps}{\varepsilon}
\newcommand{\worddef}[1]{\emph{#1}}
\newcommand{\boolor}{\mbox{ or }}
\newcommand{\booland}{\mbox{ and }}

\newcommand{\comment}[1]{\emph{\color{red}}}
\begin{document}
\begin{titlepage}
    \centering
    \includegraphics[width=0.5\textwidth]{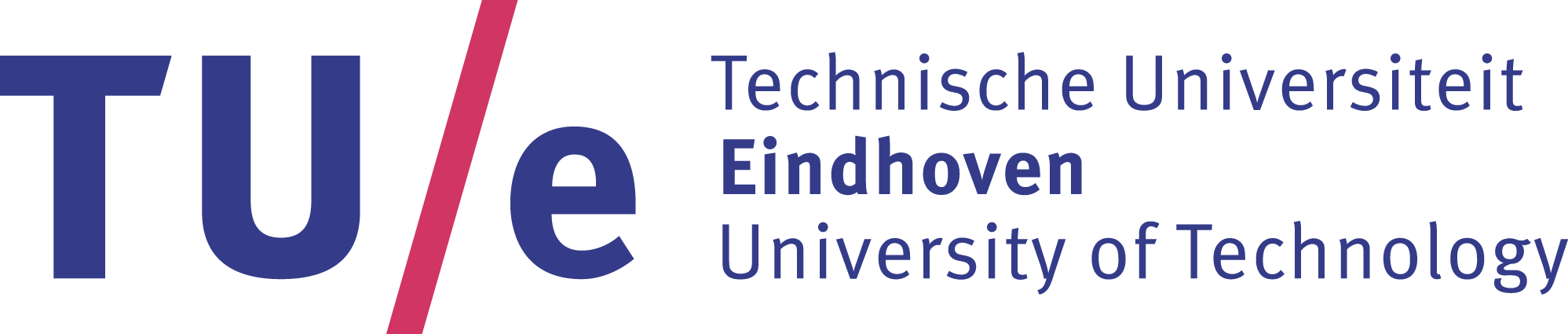}\par\vspace{1cm}
    {\scshape\LARGE University of Technology Eindhoven \par}
    \vspace{1cm}
    {\scshape\Large Master thesis\par}
    \vspace{1.5cm}
    {\huge\bfseries Large-scale Multiscale Particle Models in Inhomogeneous Domains: Modelling and Implementation\par}
    \vspace{2cm}
    {\Large\itshape Omar Richardson\par}
    \vfill
    supervised by\par
    Prof. dr. habil.~Adrian \textsc{Muntean}\\
    Dr.~Andrei \textsc{Jalba}

    \vfill

% Bottom of the page
    {\large \dateenglish\formatdate{11}{04}{2016}\par}
\end{titlepage}
\section*{Abstract}
In this thesis, we develop multiscale models for particle simulations in population dynamics. These models are characterised by prescribing particle motion on two spatial scales: \emph{microscopic} and \emph{macroscopic}.
At the microscopic level, each particle has its own mass, position and velocity, while at the macroscopic level the particles are interpolated to a continuum quantity whose evolution is governed by a system of transport equations.
This way, one can prescribe various types of interactions on a global scale, whilst still maintaining high simulation speed for a large number of particles. 
In addition, the interplay between particle motion and interaction is well tuned in both regions of low and high densities.

We analyse links between models on these two scales and prove that under certain conditions, a system of interacting particles converges to a nonlinear coupled system of transport equations.
We use this as a motivation to derive a model defined on both modelling scales and prescribe the intercommunication between them. 
Simulation takes place in inhomogeneous domains with arbitrary conditions at inflow and outflow boundaries. We realise this by modelling obstacles, sources and sinks.
Integrating these aspects into the simulation requires a route planning algorithm for the particles. Several algorithms are considered and evaluated on accuracy, robustness and efficiency.

All aspects mentioned above are combined in a novel open source prototyping simulation framework called \emph{Mercurial}. 
This computational framework allows the design of geometries and is built for high performance when large numbers of particles are involved. \emph{Mercurial} supports various types of inhomogeneities and global systems of equations.

We apply our framework to simulate scenarios in crowd dynamics.
We compare our results with test cases from literature to assess the quality of the simulations.

\tableofcontents
% Expressing my thanks
\chapter*{Acknowledgements}
By finishing this thesis, my time at the University of Technology in Eindhoven has come to an end.
In this final stage I would like to express my gratitude to all that have aided and guided me along the way.

I have deeply enjoyed my time as a student in this creative, supporting and free environment.
Most notably, I want to thank my supervisor Adrian Muntean for the unlimited encouragement in the process that led to this thesis and for harnessing my stubbornness into something fruitful.

I would like to express my gratitude to Andrei Jalba for his unconditional support. I really appreciate the pleasant discussions we had and the fact that you always took the time for unannounced visits with unformulated questions from your unplanned master student.

I also want to thank LOC7000 and especially Maarten van Lokven for his support in the internship that lead to this work, and for showing me all in crowds which could never be expressed in numbers.

Finally, I would like to thank Marko Boon. Apart from our brainstorming sessions, your help in setting up the simulation environment and your feedback, this would be a good moment to express special thanks to your everlasting patience in every report I was due.
\clearpage

% Introduction of the thesis
\chapter{Introduction}
\begin{quote} 
    \centering 
    \emph{Perception is strong and sight weak. In strategy it is important to see distant things as if they were close and to take a distanced view of close things.}
    
    \raggedleft
    - Miyamoto Musashi, \emph{The Book of Five Rings}.
\end{quote}
\bigskip
Be it in physics, mathematics or biology, many dynamical systems are active on more than one time or length scale. 
If one restricts the scope of a model to only one of those scales, it is nearly impossible to capture all essential phenomena.
This thesis discusses a multiscale approach to particle models for inhomogeneous domains.
Our goal is by exploring particle models at more than one length scale, we are able to devise models better capable of representing complex dynamical systems.

Imagine modelling the circulation of blood in a human body. Blood transport through a single artery can be compared to laminar flow in a pipe.
This has served as a starting point for many models, which have aided greatly in our understanding of the vascular system.
However, if we want to examine the harmful effects of for instance sickle-cell disease we require a different perspective, in which effects like the size and shape of blood cells needs to be included.

Another example is found in predicting the rise and fall of sea levels. A relatively accurate model can be obtained by relating the tides to the position of the Moon.
But without looking at factors on different time scales, like the influence of the Sun and the motion of the Earth, this model is not able to account for spring tides, let alone predict them.

While these phenomena are relatively rare, they are all but irrelevant. If we want to include them, we need to adjust our perspective and look at so-called \worddef{multiscale} models.

Constructing multiscale models is not trivial. One needs to deal with various orders of accuracy and try to obtain relations between the measures and quantities on these scales.

Although this coupling generally requires more knowledge and creativity in the conception of the model, it is often possible to compose the model of simpler components, resulting in a synergetic ensemble.

In this thesis, we explore transport-dominated interactive particle models from a \worddef{microscopic} and a \worddef{macroscopic} perspective.
We look at different modelling approaches, how we can prescribe interaction and transport in inhomogeneous domain, and how we can couple the information from the different scales.
We proceed to build a system in which it is possible to combine both microscopic and macroscopic information in prescribing the evolution of a particle system.
This system is implemented in a simulation framework which is used in simulating crowd dynamics.

After reviewing several relevant contributions we implement two methods related to the various forms of transport discussed.
We discuss the results and validate the simulate by comparing the results to literature.
\section{Background} 
Multiscale modelling in science has only really started ascending from the end of the twentieth century (\cite{horstemeyer90}).
With more and more computational power available to scientists and businesses, it became feasible to build and analyse increasingly complex models.
Both manufacturing industries and scientists became capable of creating reliable simulations that saved a substantial amount of money and time in the design process and analysis of structural systems.

Nowadays, applications of multiscale models are ubiquitous: from the visualisation of flowing lava(\cite{jose16}) or battling armies (\cite{miarmy11}) in movies and video games to the evolution of planetary systems in astrophysics (\cite{zwart09})

Mathematically, simultaneously evaluating systems on multiple scales is often far from trivial.
Multiscale systems are often encountered in perturbation theory, in which a problem is approximated by a series of simpler problems distributed on different time and length scales which can be solved exactly.
A related area of research is homogenisation, where systems of partial differential equations are solved with highly oscillatory coefficients. A popular application is modelling transport through a porous medium.

Our multiscale modelling approach is based on coupling (stochastic) differential equations to transport equations. The differential equations prescribe the motion of point-masses in the plane, while the transport equation governs the evolution of the continuum quantity induced by these point masses.

\section{Own contribution}
We derive formal relations between the limiting behaviour of multispecied interaction-dominated particle systems and macroscopic diffusion-driven transport systems using self-organisation properties and \^Ito's lemma.
These relations and their arguments are found in the paper in Appendix~\ref{chap:paper}.
Furthermore, we show a relation between microscopic and macroscopic measures and show that the translation is mathematically sound using geometrical arguments.

These techniques are implemented in a novel simulation framework called \emph{Mercurial} and applied to various situations in crowd dynamics, the modelling and analysis of the behaviour of pedestrians.
\emph{Mercurial} is built with ease of implementation and high performance in mind. The framework is released as an open source software package to promote the reusability of software in research. More information on \emph{Mercurial} is found in Appendix~\ref{chap:mercurial}.

\section{Outline}
We start by introducing the microscopic and macroscopic modelling scales in Chapter~\ref{chap:analysis}. We show some links that connect models from different scales. In Chapter~\ref{sec:micro_macro} we continue by deriving a consistent translation of microscale information to its macroscale representant and vice versa.
We explore how these concepts come back in the field of modelling crowds in Chapter~\ref{chap:crowds} and elaborate on two models which are extensions of the modelling approaches discussed in Chapter~\ref{chap:analysis}. We implement these models and discuss their results. In Chapter~\ref{sec:comparison} we compare these different simulation techniques and validate them using experimental studies from literature.

% The modelling background and how we use analysis to derive some results. Also, the paper is mentioned in here with some introduction. The actual result is admitted in the back.
\chapter{Multiscale modelling}
\label{chap:analysis}
This chapter introduces the two modelling scales we use throughout this thesis. 
In this framework, we define microscale systems as well as interactions between the microscopic quantities, and then give similar definitions for macroscopic systems.

In addition, this chapter is devoted to providing some mathematical arguments on why it is feasible to couple particle systems with continuum systems. 
We hope to find a connection that justifies our notion of multiscale systems by showing that specific kinds of interactive systems can be expressed on both microscale and macroscale.
We review some commonly used models of particle systems and continuum systems that fit our framework and show links between these models.
Finally, we discuss the notion of self-organisation, a measure of complexity in many of the systems we are interested in.

Before starting on the formulations, we introduce a notation we use throughout this thesis.
\section{Notation}
\label{sec:notation}
We denote a vector ($\vec{a}$) using a boldface script. This is also used for vectors that represent a discretised function.\\
We denote variables ($x$) or functions ($u: \mathbb{R}^n\to \mathbb{R}^m$) with a normal script.\\
Dot products of vectors $\vec{a}$ and $\vec{b}$ are denoted with $\langle\vec{a},\vec{b}\rangle$.

We use $\vec{x}=(x,y)$ to denote a vector in $\mathbb{R}^2$, $\Omega\subset \mathbb{R}^2$ to denote the spatial domain, and $[0,T]$ for some $T>0$ to denote the time domain.
%Say we will show some things on the interaction between micro and macro models of interacting particles and transport systems.\\
\section{Microscopic formulations}
\label{sec:micro}
We start by defining transport systems at the microscale level.
Let domain $\Omega\subset \mathbb{R}^2$ be a connected space, simulated in time interval $[0,T]$. We look at $N$ particles with positions $\vec{x}_{i}(t)$ for $i=1,\dots,N$ on time $0\leq t \leq T\leq \infty$ for $i=1,\dots,N$. Let the mass of each particle be denoted by $m_i$.
We denote the velocity of particle $i$ with $\vec{v}_i(t)$, or $\dot{\vec{x}}_i(t)$ when we wish to emphasise the physical nature of the system.

In classical mechanics, the random motion of particle $i$ can be expressed by means of a stochastic differential equation commonly known as the Langevin equation, proposed for instance in \cite{lemons97}.
\begin{equation}
    \begin{split}
    m_i \ddot{\vec{x}}_i(t) &= -\lambda\dot{\vec{x}}_i(t) + B_i(t),\\
    \vec{x}_i(0) &= \xi_i,
    \end{split}
    \label{eq:langevin_}
\end{equation}
where $\lambda$ denotes the friction coefficient and $B_i(t)$ denotes the Gaussian noise particle $i$ experiences.
For all $i$, $\xi_i$ are independently and identically distributed on domain $\Omega$.
This is often modelled with a normal probability distribution having a correlation function of the form 
\begin{equation}
    \langle B_i(t),B_i(t')\rangle = 2\lambda k_B\delta_{t'}(t).
    \label{eq:correlation}
\end{equation}
Here $k_B$ represents Boltzmann's constant and $\delta_{t'}(t)$ is the Dirac distribution.
Relation \eqref{eq:correlation} implies that no correlation exists between time $t$ and $t'$ if $t\neq t'$.

In practical applications, $\delta$ is approximated by a smooth function to model some correlation for $t'-t$ small. More on how a Dirac distribution can be approximated is found in Section~\ref{sec:interpolants}.
\subsection{Interactive systems}
In interactive systems, a particle is aware of and responds to other particles. 
Since we focus on social systems, we make the assumption \worddef{particle interaction} is determined by inter-particle distance, and the interaction manifests itself in \worddef{attraction} and \worddef{repulsion}.
Other types of particle interaction include maintaining fixed distances (present in leader-follower pairs) and assymmetric interaction (present in predator-prey pairs).
We can include interactions based on particle distances by means of an \worddef{interaction potential} function $V$. 
Including $V$ in \eqref{eq:langevin_} and ignoring the friction component results into the system
\begin{equation}
    m_i \ddot{\vec{x}}_i(t) = \frac{1}{N}\sum_{j=1}^N\nabla V\left(\vec{x}_i(t) - \vec{x}_j(t)\right) + B_i(t).
    \label{eq:interaction}
\end{equation}
By choosing an increasing function for $V$, it is possible to model attraction between particles, while a decreasing $V$ models repulsion.
Combinations of these phenomena (like a preferred interparticle distance) can be modelled by manipulating the slope of $V$.

These systems are general enough to model many interactive particle phenomena. A well known example in crowd dynamics is proposed by \cite{helbing95} and is discussed in Section~\ref{sec:social_force}.
We apply a formulation of this system in \cite{duong16}.
Another example regarding population dynamics is treated by \cite{di13}.

\eqref{eq:interaction} is easily extended to multiple species. 
By prescribing different interaction potentials one is able to model more complex symbiotic phenomena, like predator-prey systems (\cite{ackleh13}) or juvenile-adult models (\cite{dieckmann95}). 

The shape of the interaction potential strongly influences the behaviour of the system. 
In Appendix~\ref{chap:paper} we model a system with particle repulsion using a smooth symmetric potential function with finite support. The potential function is defined as
\begin{equation}
    V_\eps(r) = \frac{1}{c\eps^2\sqrt{2\pi}} e^{-r^2/(c\eps)^2}
    \label{eq:fin_pot}
\end{equation}
The finite support and symmetry enable us to evaluate this system on macroscopic level for $N\to\infty$ and $\eps\to0$.
Because $\nabla V_\eps=0$, the repulsive force between particles weakens when their distances becomes zero. The stochastic component in \eqref{eq:interaction} prevents this anomaly from undermining repulsive behaviour.

\subsection{Philipowski's approach}
\label{sec:philip}
In \cite{philipowski07}, Philipowski shows how a particle system similar to \eqref{eq:interaction} conditionally converges to a macroscopic density that satisfies the porous medium equation.
With some minimal adaptions, this result can be used to show convergence of our particle system to macroscopic transport equations as well.
First, we follow his line of arguments. 

Assume the distance-interaction dominated system defined in \eqref{eq:interaction} in $d$ dimensions, with an interaction potential $V$.
An asymptotic scaling $V^\eps: \mathbb{R}^d\to \mathbb{R}$ is introduced, defined as 
\begin{equation}
    V^\eps(x) := \frac{1}{\eps^d}V(x/\eps),
    \label{}
\end{equation}
where $\eps>0$ scales interaction range.
Introducing a diffusion coefficient $\delta>0$ we obtain the system
\begin{equation}
    \dot{\vec{x}}_i(t) = -\frac{1}{N}\sum_{j=1}^N\nabla V^\eps\left(\vec{x}_i(t) - \vec{x}_j(t) \right) + \delta B_i(t).
    \label{eq:philip}
\end{equation}
Before we state the results of his contribution, we introduce the porous medium equation and the empirical measure:
Let the classical porous medium equation for density $u(x,y,t):\Omega\times[0,T]\to \mathbb{R}$ be defined as 
\begin{equation}
    \begin{split}
        \frac{\partial u}{\partial t} &= \frac{1}{2}\Delta(u^2),\\
        u(\cdot,0) &= u_0.
    \end{split}
    \label{eq:pme}
\end{equation}
By examining the \worddef{empirical measure} $\mu(t):[0,T]\to\Omega$ defined as 
\begin{equation}
    \mu(t) = \frac{1}{N}\sum_{i=1}^N\delta_\vec{x}(t),
    \label{eq:emp_measure}
\end{equation}
Philipowski examines the limit of the system when $N\to\infty$, $\eps\to0$ and $\delta\to0$ such that $N \gg \frac{1}{\eps}$ and $\eps \ll \delta$ under the following assumptions:

Let $W_{n,1}^2\left( \mathbb{R}^d \right)$ be the weighted Sobolev function space. This space consists of all $n$ times weakly differentiable functions $f:\mathbb{R}^d\to \mathbb{R}$ with compact support for $f$ together with the partial derivatives.\\
Take $u_0 \in W_{n,1}^2\left( \mathbb{R}^d \right)$ for all $n\in \mathbb{N}$. The initial condition $u_0$ is defined as
\begin{equation}
    u_0 = \mu(0) = \frac{1}{N}\sum_{i=1}^N\delta_{\vec{x}_i(0)}.
    \label{}
\end{equation}
For $\vec{r}\in \mathbb{R}^d$, assume the following properties hold for $V$.
\begin{itemize}
    \item $V(\vec{r}) = V(\vec{-r})$.
    \item $V\geq0$.
    \item $\int_{ \mathbb{R}^d } V(\vec{r})d\vec{r} = 1$.
    \item $\int_{ \mathbb{R}^d } ||\vec{r}||^nV(\vec{r})d\vec{r} < \infty$.
\end{itemize}
In \cite{philipowski07}, it is shown that under these conditions, both the empirical measure of the particle system as the distribution of the particles converges weakly to a measure that solves \eqref{eq:pme}.

We sketch the steps taken to prove this result.
\begin{enumerate}
    \item Prove that the particle system in \eqref{eq:philip} is well-posed.
    \item Prove that \eqref{eq:non_lin_stoc_1} is obtained from \eqref{eq:philip} if $N\to\infty$ by enforcing $N\gg\frac{1}{\eps}$.
    \item Prove that \eqref{eq:non_lin_stoc_2} is obtained from \eqref{eq:non_lin_stoc_1} if $\eps\to0$ using a fixed point argument and by enforcing $\eps \gg \delta$.
    \item Finally, prove that \eqref{eq:non_lin_stoc_2} converges to a weak solution of \eqref{eq:pme} if $\delta\to0$ using It\^o's lemma.
\end{enumerate}
The intermediate systems are shown below.
\begin{equation}
    \begin{split}
        \dot{\bar{\vec{x}}}_i(t) &= -\left( \nabla V^\eps\ast u^{\eps,\delta} \right)\left( \bar{\vec{x}}_i(t),t \right) + \delta B_i(t),\\
        P\left( \dot{\bar{\vec{x}}}_i(t) \in B(\vec{s},r) \right) &= \int_{B(\vec{s},r)}u^{\eps,\delta}(\vec{s},t)d\vec{x},
    \end{split}
    \label{eq:non_lin_stoc_1}
\end{equation}

\begin{equation}
    \begin{split}
        \dot{\bar{\vec{x}}}_i(t) &= -\nabla u^\delta\left( \bar{\vec{x}}_i(t),t \right) + \delta B_i(t),\\
        P\left( \bar{\vec{x}}_i(t) \in B(\vec{s},r) \right)&= u^\delta(\vec{s},t)|B(\vec{s},r)|,\\
        u^\delta &\in \mathbf{C}^{1,2}_b\left( [0,T]\times \mathbb{R}^d \right)\quad\forall T\geq 0.
    \end{split}
    \label{eq:non_lin_stoc_2}
\end{equation}
Here $B(\vec{s},r)\subset \mathbb{R}^d$ denotes a control volume, a ball with centre $\vec{s}$ and arbitrarily small radius $r>0$.

\subsection{Domain potentials}
\label{sec:domain_pots}
Microscopic systems need not only be defined by interaction. Especially when modelling inhomogeneous geometries, the domain itself plays an significant role in influencing particle motion.
When particle motion is not dominated by distance-based interactions but by other (spatially determined) factors, we require a different formulation of the particle system, and a different interpretation for the potential function. 
Let $\Phi:\Omega\to \mathbb{R}^2$ be a \worddef{domain potential}.
We require $\Phi$ to be differentiable everywhere in $\Omega$. 
Then the (possibly non-linear) system that governs the particle motion is given by
\begin{equation*}
    m \ddot{\vec{x}}_i(t) = g\left(\nabla\Phi(\vec{x})\right) + \eta(t).
\end{equation*}
Here $g$ is a function that converts the potential gradient into a motion direction. A popular choice is to pick $g$ as a normalizing function. The effects of such a system are explored in Section~\ref{sec:potential_planner}.

This potential function was used in \cite{treuille06} to model the reactionary nature of particles (in this case pedestrians in a crowd) to their environment.\\
In \cite{helbing99} it is shown that when such a potential function can be formulated, one can measure the self-organisation in such a system. 
Moreover, when the system satisfies other conditions, like symmetry in interaction, it is shown that this self-organisation leads to optimality in terms of the energy spent in moving.
Since this derivation is provided in a macroscopic framework, we describe it in Section~\ref{sec:macro:domain_pots}.

\section{Macroscopic formulations}
\label{sec:macro}
In this section, we give a definition of a quantity defined on macroscale, and translate the concepts introduced for the microscopic quantities to their macroscopic alternatives.

The macroscopic formulations are typically defined by fluid-dynamic-like representations. More precisely, they are defined in terms of mass and momentum. We introduce a density field $\rho: \Omega\times[0,T] \to \mathbb{R}$ and a velocity field $v = \begin{pmatrix}v_x\\v_y\end{pmatrix}:\Omega\times[0,T] \to \mathbb{R}^2$.
The collection of particle masses is represented in $\rho$, while the velocities are represented in $v$.
The evolution of $\rho$ and $v$ is governed by the conservation law of mass, resulting in the \worddef{continuity equation}:
\begin{equation}
    \frac{\partial \rho}{\partial t} = - \div{\rho v}.
    \label{eq:cont_equation}
\end{equation}
A derivation of \eqref{eq:cont_equation} starting from the conservation of mass can be found in textbooks on fluid dynamics.

Velocity field $v$ can be specified in various ways.
For us, the most interesting choice is to pick functions that approximate particle systems with interacting potentials or domain potentials.

\subsection{Interaction potentials}
Modelling distanced-based interaction on a macroscopic level is possible by coupling the velocity field to the density. In this section we discuss a technique that can be used to incorporate repulsion and attraction.

We model repulsion by imposing Darcy's law on the macroscopic transport.
Darcy's law states a relation between flux $q$ and pressure $p$ in a porous medium with permeability parameter $\kappa$.
This relation is defined as
\begin{equation}
    q = \frac{-\kappa}{\mu}\nabla p,
    \label{eq:darcy_flux}
\end{equation}
where $\mu$ denotes the kinematic viscosity of the fluid.
We assume a relation between pressure and density satisfying
\begin{equation}
    p(\rho) = \left(\frac{\rho}{\rho_{\min}}\right)^\alpha,
    \label{eq:pressure_dens}
\end{equation}
for some $\alpha \in \mathbb{N}^+$ and normalizing constant $\rho_{\min}$.\\

This causes high densities to yield a pressure which reduces those densities and in that way emulates particle repulsion. 
We provide an detailed elaboration on how to model and implement a repulsive interaction potential in simulations in Section~\ref{sec:crowds2}.

Darcy's law also provides us a way to model attraction of particles. Reversing the sign of the flux in \eqref{eq:darcy_flux} we obtain a system where particles are attracted to locations of high densities.
%This gives rise to a system with finite blowup times, where all mass is concentrated in one point after some time.
\subsection{Domain potentials}
\label{sec:macro:domain_pots}
A domain potential can be modelled on macroscale by incorporating it in the flux term.
Assuming the domain potential $\Phi$ from Section~\ref{sec:domain_pots} and then plugging it in the continuity equation, then mass flow is propagated along the steepest descent of the potential function.

The motion of active particles in inhomogeneous systems is often more complex and calls for a more elaborate transport prescription: \worddef{intelligent transport}.
One form of intelligent transport can be induced by ensuring $\Phi = \Phi(x,y,\rho)$. This has been explored in \cite{hughes02} by limiting maximum speed and imposing constraints on the maximum density.
If we reformulate his system of governing equations (expanded on in Section~\ref{sec:macro_literature}) we obtain
\begin{equation}
    \begin{split}
        \frac{\partial \rho}{\partial t} &= - \div{-\rho \nabla\Phi},\\
    \Phi(x,y;\rho) &= g(\rho)f^2(\rho)\varphi(x,y).
    \end{split}
    \label{eq:hughes_pot}
\end{equation}
In \eqref{eq:hughes_pot}, $\varphi$ represents the base domain potential and $f$ and $g$ limit speed and attraction for high densities. 
In Section~\ref{sec:crowds1}, we show how to model geometries with a potential function and examine the performance of such a method.
The domain potential is illustrated in Figure~\ref{fig:narrowdf1}.

\subsection{Self-organisation and the porous medium equation}
\label{sec:self_organisation}
In \cite{helbing99}, it is shown that under certain conditions, it is possible to formulate a Lyapunov functional $F(t)$ for the overall transport system defined as
\begin{equation}
    F(t) = \int_\Omega\rho(\vec{x},t)\Phi(\vec{x},t)d\vec{x}.
    \label{eq:lyapunov}
\end{equation}
When $\Phi$ is non-negative everywhere, $F(t)$ represents a measure of \worddef{self-organisation} of the system at time $t$. The lower $F(t)$, the less energy is spent in transport and the more optimal the system performs. 

Self-organising systems are therefore identified by a decreasing Lyapunov functional.

As an example inspired by \eqref{eq:pressure_dens}, let $\Phi(x,y):=\rho(x,y)$. This models a repulsive system without any directional preference for the particles for a spatially homogeneous system. The resulting system becomes the porous medium equation, equivalent to \eqref{eq:pme}.
\begin{equation*}
    \frac{\partial \rho}{\partial t} = \div{ \rho \nabla \rho }.
\end{equation*}
The resulting Lyapunov functional becomes
\begin{equation*}
    F(t) = \int_\Omega\rho(\vec{x},t)\rho(\vec{x},t)d\vec{x}.
\end{equation*}
We show the self-organising property of this system by proving $F(t)$ is non-increasing with respect to $t$.
\begin{equation*}
    \begin{split}
        F'(t) &= \frac{d}{dt}\int_\Omega\rho^2d\vec{x},\\
        & = 2\int_{\Omega}\rho \frac{\partial\rho}{\partial t}d\vec{x}, \\
        &= 2\int_{\Omega} \rho \div{\rho\nabla\rho} d\vec{x}, \\
        &= 2\int_{\partial\Omega}\rho \left(\rho\nabla\rho \right) d\vec{x} - 2\int_{\Omega} \nabla \rho \left(\rho\nabla\rho  \right)\cdot d\vec{x}, \\
        &= 0-2\int_{\Omega}  \rho\left( \nabla\rho  \right)^{2}d\vec{x}  \leq 0.
    \end{split}
    \label{}
\end{equation*}
In this derivation, we assume $\rho(\vec{x},t)\geq0$ for $\vec{x}\in\Omega$ and $\nabla\rho(\vec{x},t)=0$ for $\vec{x}\in\partial\Omega$, for $t\geq0$.
This assumes the density in the system cannot become negative and the system conserves mass by allowing no transport around the boundaries of the domain.
In addition, because $F(t)\geq0$ for all $t$, the system in \eqref{eq:pme} converges to an optimal equilibrium.
\section{Exploring the limit behaviour of large particle systems}
\label{sec:paper}
In this section, we have seen at least two ways exist of modelling transport and interaction phenomena. On microscale, this can be done by specifying a set of particles moving according to a set of stochastic differential equations. 
On macroscale, this can be done by prescribing the evolution of a quantity with a continuity equation.
What remains to be shown is how the microscale system relates to the macroscale system.

More specifically, if we focus on systems defined by an interaction potential, is it then possible to view the macroscale system as an asymptotic representation of the microscale system?
The derivations from \cite{philipowski07} provided in Section~\ref{sec:philip} show that under certain conditions microscopic models can be reformulated as a porous medium equation.
From the interaction potential formulation in Section~\ref{sec:self_organisation}, we see that the porous medium equation can be viewed as a macroscopic transport equation of two populations driven by repulsive forces.

In \cite{duong16}, we investigate further this relation by modelling an advection-diffusion transport system involving two populations and showing its equivalence to a two-species particle model.
We support these findings by simulating both systems and comparing the results.
This paper is included in Appendix~\ref{chap:paper}.

Our findings in Appendix~\ref{chap:paper} are connected to the crowd dynamics applications described in Chapter~\ref{chap:crowds}. 
In that chapter, we discuss using an interaction potential to model repulsion.
But since we want to avoid evaluating this interaction on a microscopic level, we incorporate this into the macroscopic model using an advection based transport equation based on Darcy's law, much like the system discussed in Appendix~\ref{chap:paper}.

Also, by coupling and implementing particle models on multiple scales, we gain computational efficiency without sacrificing the simplicity of our model definition.

% The modelling assumptions and methods used to transfer information from micro to macro and back, plus some validation.
\chapter{Conversion between micro and macro scales}
\label{sec:micro_macro}
The previous chapter discusses formal arguments on the link between microscopic and macroscopic models.
We have seen that a microscopic system is defined by the mass and velocities of its particles, while a macroscopic system is defined by the density and velocity field of the continuum quantity.
In this chapter we treat implementation aspects involved in describing the interaction between these representations.

First, inspired by the distance-based interaction formulations in Section~\ref{sec:micro} we introduce an interpolation-based method: smooth particle hydrodynamics (SPH). We show how we apply this method to the particle representation to obtain a continuum quantity, translating thus the information from microscale to macroscale.
We use SPH only as an interpolation method. Once we obtain a measure of the state of the system on the macroscopic scale, we proceed to compute the propagation on a grid.
Computing the evolution of continuum quantities with grid-based methods is mathematically much better understood than using mesh-free SPH.

Our approach is loosely based on the 'Particle In Cell'-method as described in \cite{zhu05}.

We show a relation between the minimum distance as a measure on a microscopic level, and the maximum density measured on a macroscopic level.\\
Finally, we elaborate on the bilinear interpolation used to translate macroscopic information back to a microscopic level.
All of these techniques are used in the simulations presented in Chapter~\ref{chap:crowds}.

\section{Smooth particle hydrodynamics}
\label{sec:sph}
\worddef{Smooth particle hydrodynamics}, originally proposed in \cite{gingold77}, is a numerical simulation technique where the motion of fluids (or gases) are modelled by the evolution of a set of discrete particles.
While classical fluid-dynamic models use a fixed grid on which each time step the state variables are approximated, SPH is meshfree and approximates the fluid properties at moving interpolation points using so-called \emph{kernel interpolants}.
These interpolation points represent (a collection of) fluid particles moving with the advection. In particle systems, these interpolation points represent the actual particles.
Extended introductions and formal derivations of the SPH-method are found in \cite{monaghan05} and \cite{violeau12}.

The advantages of the SPH method for particle-like simulations are numerous: mass conservation is easy to achieve, advection based transport can be modelled exact. 
Also, because the SPH interpolation points represent presence of mass, the resolution of the approximation increases in locations with higher densities, and vanishes where no mass is present.
Finally, SPH provides a nice representation of the state variables in both a microscopic sense (as particles with a specified position, velocity and mass) and a macroscopic sense (as interpolated density and velocity fields).

One disadvantage of SPH is the numerical diffusion it introduces. Because of the continuous interpolations, it is difficult to represent and maintain discontinuities which sometimes are expected to occur in the evolution of the system.
Another disadvantage is the difficulty in modelling incompressibility in fluids. Due to accumulating numerical time integration errors, a divergence-free velocity field is difficult to enforce.

\section{Kernel interpolants}
\label{sec:interpolants}
The information from the discrete set of particles is transferred to the continuum level with the use of \emph{kernel interpolants}.
These are functions that approximate the Dirac distributions that account for the positions of the point particles.
The best known example is the Gaussian function, the density function of the normal distribution.

Following the definitions in \cite{monaghan05} and \cite{violeau12}, a kernel interpolant $\psi: \mathbb{R}^2\to\mathbb{R} $ must satisfy the following properties for all $\vec{r}\in \mathbb{R}^2$:
\begin{enumerate}
    \item $\int\psi(\vec{r})d\vec{r} = 1$.
    \item $\psi(\vec{r}) =\psi(-\vec{r})$.
    \item $\nabla \psi(0) = 0$.
\end{enumerate}
We also demand our interpolation kernels to be non-negative everywhere and have finite support (to increase computational efficiency).

Let $\delta_{\vec{x}_0}$ be the Dirac distribution with value $\infty$ in $\vec{x}_0$ and 0 everywhere else. Let $f\ast g$ be the convolution of $f$ and $g$, defined as 
\begin{equation}
    (f\ast g)(\vec{x}) := \int f(\vec{r})g(\vec{x}-\vec{r})d\vec{r}.
\label{eq:convolution}
\end{equation}
Particle $a$ is fully defined by its position $\vec{x}_a$, velocity $\vec{v}_a$ and mass $m_a$.
We model the (discrete) particle density $\tilde{\rho}_a$ with a Dirac distribution:
\begin{equation*}
    \tilde{\rho}_a(\vec{x}) = m_a\delta_{\vec{x}_a}(\vec{x}).
\end{equation*}
When smoothing this particle with the kernel interpolant, we retrieve the continuous density $\rho_a$:
\begin{equation*}
    \rho_a(\vec{x}) = (\tilde{\rho}_a\ast \psi)(\vec{x}) = m_a\int \delta_{\vec{x}_a}(\vec{r})\psi(\vec{x}-\vec{r})d\vec{r}.
\end{equation*}
We need to establish an interaction range. The straightforward way to address this is with a \emph{smoothing length} $h>0$.
Drawing a parallel to the normal distribution, the smoothing length corresponds to the standard deviation of a Gaussian distributed random variable.

The radial symmetry implies we can express the kernel interpolant as a one-dimensional function $f$ satisfying
\begin{equation*}
    \psi(\vec{r}) = \psi(||\vec{r}||) = f(r),
\end{equation*}
for $r = ||\vec{r}|| \in \mathbb{R}$.

Commonly used kernels include:
\begin{enumerate}
    \item The Gaussian function $f_G(r;h) := \frac{1}{\sqrt{2\pi}h^2}e^{-\frac{r^2}{2h^2}}$.
    \item The so-called B-spline with polynomial degree 4:

        $f_4(r;h) := \frac{384}{1199\pi h^2}
        \begin{cases}
            \left( \frac{5}{2} - q \right)^4 - 5\left( \frac{3}{2} - q \right)^4 + 10\left( \frac{1}{2}-q \right)^4 &\mbox{ for } 0 \leq q \leq0.5\\
            \left( \frac{5}{2} - q \right)^4 - 5\left( \frac{3}{2} - q \right)^4 &\mbox{ for } 0.5\leq  q \leq 1.5\\
            \left( \frac{5}{2} - q \right)^4 &\mbox{ for } 1.5\leq  q \leq 2.5\\
            0 &\mbox{ for } 2.5\leq q 
        \end{cases},$\\
        with $q=r/h$.
    \item The Wendland kernel $f_W(r;h) := \begin{cases}
                \frac{7}{4\pi h^2}\left(1-\frac{r}{2h}\right)^4(1+\frac{2r}{h})& 0\leq r \leq 2h\\
                0 & 2h < r 
            \end{cases}.$
\end{enumerate}
In our simulations, we pick the Wendland kernel.
This choice is motivated by the convenient support radius of $2h$ and by the fact it can be reformulated to a single algebraic expression:
\begin{equation}
    f_W(r;h) = \frac{7}{4\pi h^2}\max\left\{1-\frac{r}{2h},0\right\}^4\left(1+\frac{2r}{h}\right).
\end{equation}
This provides a welcome computational benefit over more complex kernel interpolants, since interpolation is a common operation in the simulation.
We plot the Wendland kernel in Figure~\ref{fig:kernel_interpolant}. 
\begin{figure}[h]
    \centering
    \includegraphics[width=0.5\textwidth]{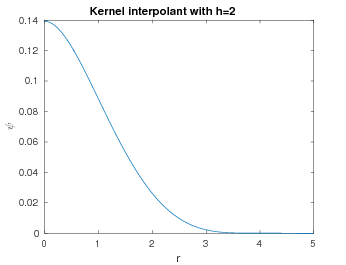}
    \caption{Shape of the Wendland kernel interpolant.}
    \label{fig:kernel_interpolant}
\end{figure}
\section{Smoothing the particles}
Based on the definitions from Section~\ref{sec:interpolants}, we are able to translate the microscopic information.
For convenience, we assume all particles to have equal mass $m$. Let $\Omega \subset \mathbb{R}^2$ be the simulation domain containing $N$ particles $a_1,\dots,a_N$.
Then for all $\vec{x} \in \Omega$ the density field $\rho: \Omega \mapsto \mathbb{R}$ is given by
\begin{equation}
    \rho(\vec{x}) = \sum_{i=1}^N\rho_{a_i}(\vec{x})= \sum_{i=1}^Nm\left( \delta_{\vec{x}_{a_i}}\ast \psi  \right)(\vec{x}),
    \label{eq:dens_sum}
\end{equation}
and the velocity field $v: \Omega \mapsto \mathbb{R}^2$ is given by 
\begin{equation}
    v(\vec{x}) = \frac{\sum_{i=1}^N\vec{v}_{a_i}\left( \delta_{\vec{x}_{a_i}}\ast \psi  \right)(\vec{x})}{\rho(\vec{x})}.
\end{equation}
Note the resemblance between \eqref{eq:dens_sum} and the empirical measure in \eqref{eq:emp_measure}.
Before interpolation, we partition $\Omega$ in square cells with sides of length $\frac{4}{3}h$. For each cell, we collect the particles it contains as well as the particles of each of its eight neighbours and compute their contribution to the density in that cell. Coupling the cell size to the smoothing length ensures that all contributing particles are found in a one-cell radius.
\begin{figure}[h]
	\centering
	\includegraphics[width=0.65\textwidth]{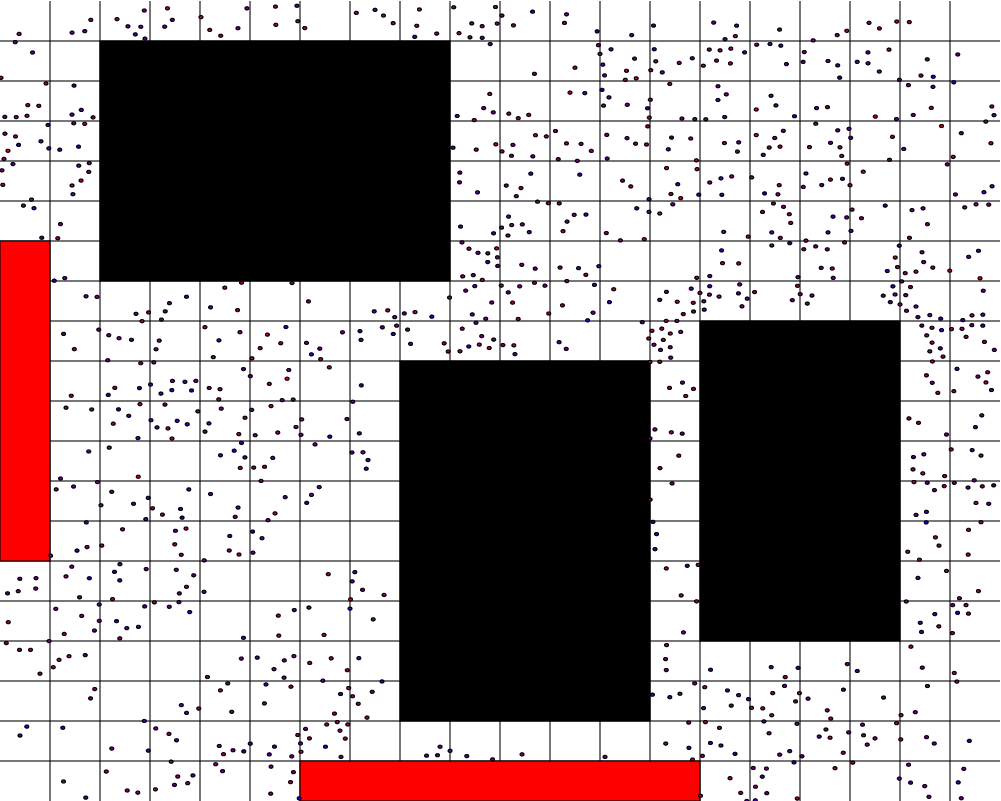}
    \caption{Domain representation with particles as dots, solid boundaries in black and outflow boundaries in red.}
	\label{fig:cont_ex_scene}
\end{figure}
\begin{figure}[h]
	\centering
	\includegraphics[width=0.85\textwidth]{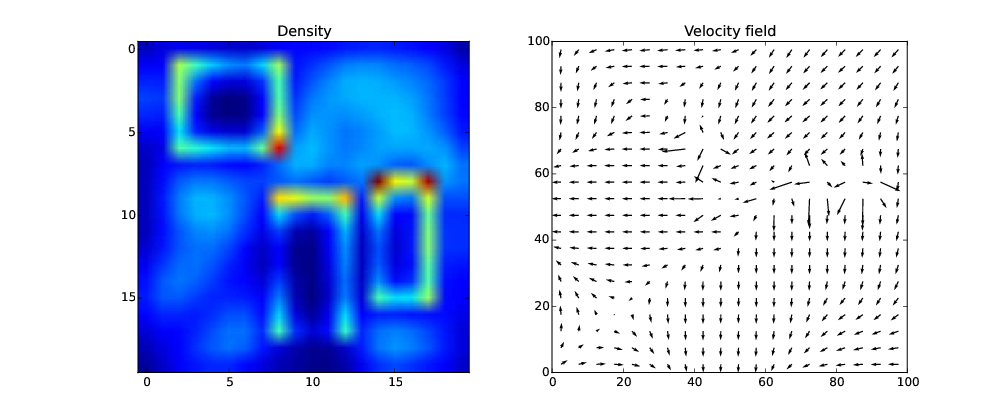}
    \caption{Macroscopic representation of the domain in Figure~\ref{fig:cont_ex_scene}: density and velocity.}
	\label{fig:cont_flow_field}
\end{figure}\\
Figure \ref{fig:cont_ex_scene} shows an example of a domain with one thousand particles and three obstacles. Figure \ref{fig:cont_flow_field} shows the corresponding density and velocity fields.

\section{Discretisation of the state variables}
\label{sec:discretisation}
We require a numerical approximation of the macroscopic density $\rho$ and velocity $v$. To obtain this, we first discretise the domain $\Omega$. From now on, we assume $\Omega$ to be rectangular.
To establish a spatial discretisation, let $N_x,N_y \in \mathbb{N}$ the number of cells in $x$ and $y$-direction on a equidistant grid. 
Let $\Delta x,\Delta y \in \mathbb{R}^+$ be the cell size in $x$ and $y$-direction, such that each cell $c_{(i,j)} = [(i-1)\Delta x,i\Delta x]\times[(j-1)\Delta y,j \Delta y]$ for all $i=1,\dots,N_x$ and $j=1,\dots,N_y$. 
This ensures $\bigcup_{i,j}c_{(i,j)} = \Omega$.\\
The time domain is discretised with a step size $\Delta t>0$.

We discretise the scalar and vector fields corresponding to the grid. We index the cells and the fields from the bottom left, corresponding to to the Cartesian indexing used for the fields. This is illustrated in Figure \ref{fig:indexing}.
\begin{figure}
	\centering
	\begin{tikzpicture}
		\draw[step=2cm,black,thin] (0,0) grid (5,5);
		\draw[thick,->] (0,0) -- (0,5) node[anchor=south east] {$y$};
		\foreach \x in {1,2,3}
			\pgfmathparse{\x*2-1}
			\edef\cx{\pgfmathresult}
			\foreach \y in {1,2,3}
				\pgfmathparse{\y*2-1}
				\edef\cy{\pgfmathresult}
				\pgfmathparse{int(\y-1)}
				\edef\ym{\pgfmathresult}
				\draw (\cx cm,\cy cm) node[anchor=center] {$\begin{matrix}\x+\ym N_x\\ (\x,\y)\end{matrix}$};
	\end{tikzpicture}
	\captionof{figure}[Caption]{Orientation of the coordinates and fields in the scene with corresponding $\left(\begin{matrix}\text{1D} \\ \text{2D} \end{matrix}\right)$ cell indexing of the scene. The 1D notation is used to orient corresponding vectors and matrices, while the 2D notation is used for legibility and element-wise computations.}
	\label{fig:indexing}
\end{figure}
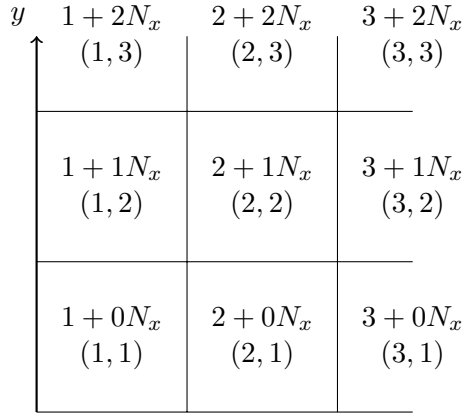\\
This way, we discretise fields to vectors for any time step $k\in\mathbb{N}$. 

For all $i=1,\dots,N_x$ and $j=1,\dots,N_y$, the relation between a field $u$ and its discrete representation $\vec{u}^k_{(i,j)}$ is defined as
$$\vec{u}^k_{(i,j)} := u((i-0.5)\Delta x,(j-0.5)\Delta y,k\Delta t).$$
The entries of vector $\vec{u}$ are approximations of $u$ in the centres of the corresponding cells.

So for any discrete scalar field $\vec{u}$, the value in cell $c_{(i,j)}$ is denoted as $\vec{u}_{i+(j-1)N_x}$. To increase legibility, we introduce a corresponding notation for indexing the scalar field in $(i,j)$:
\begin{equation*}
	\vec{u}_{(i,j)}:=\vec{u}_{i+(j-1)N_x}.
\end{equation*}

Finally, we give an expression for the density approximation in  $\vec{x}\in c_{(i,j)}$ for $i=1,\dots,N_x$ and $j=1,\dots,N_y$.
\begin{equation*}
    \gvec{\rho}(\vec{x}) = \sum_{i=1}^n m\psi\left( \left|\left|\vec{x}_{a_i} - \begin{pmatrix} (i-0.5)\Delta x\\(j-0.5)\Delta y \end{pmatrix}\right|\right| \right).
\end{equation*}
The velocity field can be expressed analogously.
\section{Determining the maximum density}
\label{sec:maxdens}
If we wish to couple a microscopic model with a macroscopic model, then we need a notion of maximum density derived from microscopic quantities. 
If we fix particle mass $m$ and size $r$, then we are able to compute a maximum density by prescribing a minimal distance $d_{\min}$ between particles using geometric arguments.
A derivation is shown below:

\begin{newthm}
	\label{thm:close}
    Let $B(\vec{x},s)\subset \mathbb{R}^2$ denote a circle with centre $\vec{x}$ and radius $s$.\\
    Assume a collection of $N$ circles $B(\vec{x}_1,r+d_{\min}),\dots,B(\vec{x}_N,r+d_{\min})$ such that 
    \begin{equation*}
        \bigcup_{i=1}^N B(\vec{x}_i,r+d_{\min})\subset \Omega\text{ and } B(\vec{x}_i,r+d_{\min}) \cap  B(\vec{x}_j,r+d_{\min}) = \emptyset\mbox{ if }i \neq j.
    \end{equation*}
    In addition, assume $\Omega$ is much larger than $B(\cdot,r+d)$.
	The closest packing a group of particles can attain is a triangular structure.
\end{newthm}
\begin{proof}
	The proof of this theorem can be omitted. It was proven by Lagrange in 1773 and a recent simple proof is found in \cite{fukshansky09}.
\end{proof}
\begin{figure}[h]
	\centering
	\begin{minipage}{.45\textwidth}
		\centering
		\begin{tikzpicture}
		\def\h{1.732}
		\def\x{0.577}
		\def\r{0.1}
		\def\d{0.5}
		\foreach \i in {0,...,5}{
			\draw ({(1+2*\i)*\d},-\h*\d) circle(\r);
			\draw (2*\d*\i,0) circle(\r);
			\draw ({(1+2*\i)*\d},\h*\d) circle(\r);
			\draw (2*\i*\d,2*\h*\d) circle(\r);
			\draw ({(1+2*\i)*\d},3*\h*\d) circle(\r);
		}
		\draw (2*\d,0*\d)--(6*\d,0*\d)--(4*\d,2*\h*\d)--(2*\d,4*0*\d);
		\draw (6*\d,0*\d)--(10*\d,0*\d)--(8*\d,2*\h*\d)--(6*\d,4*0*\d);
		\draw [red] (4*\d,2*\h*\d)--(6*\d,0*\d)--(8*\d,2*\h*\d)--(4*\d,2*\h*\d);
		\draw [red,thin,dashed] (6*\d,0*\d)--node[left]{$h$} (6*\d,\h);
		\node [red,above] at (5*\d,3.5*\d)  {$d$};
		
		\end{tikzpicture}
		\captionof{figure}{Schematic representation of the most dense configuration. Particles (shown as circles) are packed in a triangular lattice with distance $d$ between their centres.}
		\label{fig:dense_conf}
	\end{minipage}%
	\hfill
	\begin{minipage}{.45\textwidth}
		\centering
		\begin{tikzpicture}
		\def\r{1}
		\def\d{4}
		\def\pi{3.1416}
		\def\ss{1.414}
		\draw (2,2) circle(\r);
		\draw (2+\d,2) circle(\r);
		\draw [red] (2,2) -- node[left] {$r$}(2-\ss/2*\r,2+\ss/2*\r);
		\draw (2+\r,2)--node[above] {$d_{\min}$} (2-\r+\d,2);
		\end{tikzpicture}
		\captionof{figure}{Close up view of two particles with radius $r$ and distance $d_{\min}$. Coupled with Figure~\ref{fig:dense_conf}, we observe $d = d_{\min}+2r$.}
		\label{fig:dense_close}
	\end{minipage}
\end{figure}
\begin{newthm}
	Assume a crowd of particles with mass 1 and radius $r$ within a space $\Omega$. Given a minimal distance of $d_{\min}$, the maximum density $\rho_{\max}$ this crowd can attain follows from $d_{\min}$ by
	\begin{equation}
	\rho_{\max}(d_{\min}) = \frac{2}{(d_{\min}+2r)^2\sqrt{3}}.
	\label{eq:max_dens_result}
	\end{equation}
\end{newthm}
\begin{proof}
	To provide an upper bound for the density of a crowd, we examine the closest packing mentioned in Theorem~\ref{thm:close}. An example structure is depicted in Figure~\ref{fig:dense_conf}. 
	In this structure, we identify a triangle (in red) that covers the whole structure when replicated. \\

    We want to obtain an expression for the average density in this structure, expressed in quantity $d$.
	We compute the average density within this triangle, and since replicating the triangle and its contents provides us with the original structure, we conclude the average density of the triangle must equal the average density of the entire structure.
	
	The area $A$ of the indicated triangle in Figure~\ref{fig:dense_conf} can be expressed as $A=dh$, where $d$ is the distance between particle centres and $h$ satisfies $d^2 + h^2 = (2d)^2$, resulting in $A = \sqrt{3}d^2$. 
	Replacing the distance between centres by the distance between particle circumferences, we get $d = d_{\min} + 2r$, as illustrated in Figure~\ref{fig:dense_close}.
	
	In the triangle, 6 particles contribute to its mass. Weighing each particle to their contribution, we obtain a total mass of  $3\cdot \frac{1}{2}+3\cdot\frac{1}{6} = 2$ per triangle. If we divide mass by area, then we obtain the formula in \eqref{eq:max_dens_result}.
\end{proof}

Using this derivation, we come to the same relation assumed in \cite{narain09}. 

\section{Verification of the density conversion}
To illustrate the fact that the relation in the previous section holds, we verify that \eqref{eq:max_dens_result} yields good numerical approximations.
Using the notation from Section~\ref{sec:maxdens}, we let $r=0$ (since the particle radius is irrelevant in this discussion) and $d_{\min}:=d$ arbitrary but fixed.
\begin{figure}
    \centering
    \begin{minipage}{0.45\textwidth}
    \begin{tikzpicture}
    \def\h{1.732}
    \def\x{0.577}
    \def\r{0.1}
    \def\d{0.5}
        \draw ({(1+2*-2)*\d},\h*\d) circle(\r);
    \foreach \i in {-1,...,1}{
        \draw ({(1+2*\i)*\d},-\h*\d) circle(\r);
        \draw (2*\d*\i,0) circle(\r);
        \draw ({(1+2*\i)*\d},\h*\d) circle(\r);
        \draw (2*\i*\d,2*\h*\d) circle(\r);
        \draw ({(1+2*\i)*\d},3*\h*\d) circle(\r);
    }
    \draw (2*\d*2,0) circle(\r);
    \draw (2*\d*2,2*\h*\d) circle(\r);
    \draw (5*\d,\h*\d) circle(\r);
    \draw (1*\d,\h*\d) node[below]{$c_{(i,j)}$} circle (0.01);
    \draw  (1*\d,\h*\d) circle (2*\d);
    %\draw [red] (1*\d+\r,\h*\d)-- node[above]{$d$} (3*\d-\r,\h*\d);
    \draw  (1*\d,\h*\d) circle (2*\h*\d);
    %\draw [green] (1*\d+\r,\h*\d)-- node[above]{$\sqrt{3}d$} (4*\d-0.707*\r,0*\h*\d+0.707*\r);
    \draw[dashed] (1*\d,\h*\d) circle (4*\d);
    %\draw [blue] (1*\d+\r,\h*\d)-- node[below]{$2d$} (3*\d-0.707*\r,3*\h*\d-0.707*\r);
    %\draw(0-\r,-2*\d) --node[below]{$i$} (4*\d+\r,-2*\d) -- (4*\d+\r,2*\d) -- (0-\r,2*\d) --node[left]{$j$} (0-\r,-2*\d);
    \end{tikzpicture}
    \caption{Densest packing, high density observed in cell centre. Circles indicate particles with same distance.}
    \label{fig:micro_vs_macro_1}
\end{minipage}%
\hfill
\begin{minipage}{0.45\textwidth}
    \centering
    \begin{tikzpicture}
    \def\h{1.732}
    \def\x{0.577}
    \def\r{0.1}
    \def\d{1}
    \draw (0,0) node[below]{$c_{(i,j)}$}  circle(0.01);
    \draw[dashed] (0,0) circle (2*\d);
    \draw (0,0) circle(\x*\d);
    \draw (0,0) circle(2*\x*\d);
    \draw (0,0) circle(2.646*\x*\d);

    \draw (-1.5*\d,{(\x-1.5*\h)*\d}) circle(\r);
    \draw (-0.5*\d,{(\x-1.5*\h)*\d}) circle(\r);
    \draw (0.5*\d,{(\x-1.5*\h)*\d}) circle(\r);
    \draw (1.5*\d,{(\x-1.5*\h)*\d}) circle(\r);

    \draw (-2*\d,{(\x-\h)*\d}) circle(\r);
    \draw (-1*\d,{(\x-\h)*\d}) circle(\r);
    \draw (0*\d,{(\x-\h)*\d}) circle(\r);
    \draw (1*\d,{(\x-\h)*\d}) circle(\r);
    \draw (2*\d,{(\x-\h)*\d}) circle(\r);

    \draw (-2.5*\d,{(\x-\h/2)*\d}) circle(\r);
    \draw (-1.5*\d,{(\x-\h/2)*\d}) circle(\r);
    \draw (-0.5*\d,{(\x-\h/2)*\d}) circle(\r);
    \draw (0.5*\d,{(\x-\h/2)*\d}) circle(\r);
    \draw (1.5*\d,{(\x-\h/2)*\d}) circle(\r);
    \draw (2.5*\d,{(\x-\h/2)*\d}) circle(\r);

    \draw (-2*\d,{(\x)*\d}) circle(\r);
    \draw (-1*\d,{(\x)*\d}) circle(\r);
    \draw (0*\d,{(\x)*\d}) circle(\r);
    \draw (1*\d,{(\x)*\d}) circle(\r);
    \draw (2*\d,{(\x)*\d}) circle(\r);

    \draw (-1.5*\d,{(\x+\h/2)*\d}) circle(\r);
    \draw (-0.5*\d,{(\x+\h/2)*\d}) circle(\r);
    \draw (0.5*\d,{(\x+\h/2)*\d}) circle(\r);
    \draw (1.5*\d,{(\x+\h/2)*\d}) circle(\r);

    \draw (-1*\d,{(\x+\h)*\d}) circle(\r);
    \draw (0*\d,{(\x+\h)*\d}) circle(\r);
    \draw (1*\d,{(\x+\h)*\d}) circle(\r);
    \end{tikzpicture}
    \caption{Densest packing, low density observed in cell centre. Circles indicate particles with same distance.}
    \label{fig:micro_vs_macro_2}
\end{minipage}
\end{figure}
We use \eqref{eq:max_dens_result} to compute the maximum density, which due to the closest packing is equal to the observed density, and we obtain
\begin{equation}
    \rho_{ij}(d) = \frac{2}{d^2\sqrt{3}} \approx \frac{1.15}{d^2}.
    \label{eq:sph_bound}
\end{equation}
We use the SPH interpolation with kernel $f_W$ to determine the density in the centre of cell $c_{(i,j)}$. 
This yields an approximation of the density that depends on the offset of the grid with respect to the particle configuration. 
We compute an upper and a lower bound for the density, corresponding to the situations in respectively Figure~\ref{fig:micro_vs_macro_1} and Figure~\ref{fig:micro_vs_macro_2}.
The dashed line represents the support radius of the particles.

We couple the smoothing length to the particle configuration by imposing $h = d$.

\subsubsection{Upper bound on density interpolation}
We obtain an upper bound on the interpolated density by observing the density from the centre of a particle.
By repeatedly using the Pythagorean theorem, we compute the distances between the cell centre and the contributing particles. Table~\ref{tab:upper} shows these distances. 
\begin{table}
	\centering
    \begin{minipage}{0.45\textwidth}
        {\tabulinesep=1.2mm
	\begin{tabu}{|c|c|c|}
		ID &number of particles &Distance  \\
        \hline
		1 & 1 & 0\\
		2 & 6 & $d$\\
		3 & 6 & $\sqrt{3}d$
	\end{tabu}
}
	\caption{Distances between particles and cell centre for the upper bound.}
	\label{tab:upper}
\end{minipage}%
\hfill
\begin{minipage}{0.45\textwidth}
	\centering
        {\tabulinesep=1.0mm
	\begin{tabu}{|c|c|c|}
		ID &number of particles &Distance  \\
        \hline
        1 & 3 & $\frac{1}{\sqrt{3}}d$\\
        2 & 3 & $\frac{2}{\sqrt{3}}d$\\
        3 & 6 & $\sqrt{\frac{7}{3}}d$
	\end{tabu}
}
	\caption{Distances between particles and cell centre for the lower bound.}
	\label{tab:lower}
\end{minipage}
\end{table}
The total density sums up to
\begin{equation}
    \rho_{ij}(d) =\psi(0) + 6\psi(d) + 6\psi\left(\sqrt{3}d\right) \approx \frac{1.19}{d^2}.
    \label{eq:upper_bound}
\end{equation}
\subsubsection{Lower bound on density interpolation}
We obtain the lower bound by maximizing the distance between the interpolation centre and the closest particles. The distances are listed in Table~\ref{tab:lower}. The configuration relative to the particle centre is illustrated in Figure~\ref{fig:micro_vs_macro_1}. In this case, the total density sums up to 
\begin{equation}
    \rho_{ij}(d) = 3\psi\left( \frac{1}{\sqrt{3}}d \right) + 3\psi\left( \frac{2}{\sqrt{3}} d\right) + 6\psi\left( \sqrt{\frac{7}{3}} d\right) \approx \frac{1.14}{d^2}.
    \label{eq:lower_bound}
\end{equation}
Comparing \eqref{eq:upper_bound} and \eqref{eq:lower_bound} with \eqref{eq:sph_bound} shows that independent of the distance $d$ between the particles, both the density measure on microscale and on macroscale scale with $\frac{1}{d^2}$. In addition, they scale with virtually the same proportionality constant.
\section{Bilinear interpolation}
To translate the grid-based information back to individual particle positions, we use \worddef{bilinear interpolation}, a technique often used in image processing. 
Let $f$ be a discrete function only defined on the cell centres of the discretised $\Omega$. We can approximate the value of $f$ in any point $\vec{x}=(x,y) \in \Omega$ by first applying linear interpolation in $x$-direction and subsequently in $y$-direction.

Let $\bar{f}$ be the continuous bilinear interpolation function of $f$ defined on the entire space $\Omega$. Let $d_{(i,j)} := c_{(i,j)} + \left( \frac{1}{2}\Delta x,\frac{1}{2}\Delta y \right)$.
Then for each set $d_{(i,j)} \subset \Omega$ we define function $\bar{f}_{(i,j)}: d_{(i,j)} \mapsto \mathbb{R}$ as 
\begin{equation}
    \bar{f}_{(i,j)}(x,y) = a_0 + a_1x + a_2y + a_3xy.
    \label{}
\end{equation}
The coefficients $a_0,\dots,a_3$ are defined by the solution of the system
\begin{equation}
    \begin{pmatrix}
        1 & x_i & y_j & x_iy_j\\
        1 & x_i & y_{j+1} & x_iy_{j+1}\\
        1 & x_{i+1} & y_j & x_{i+1}y_j\\
        1 & x_{i+1} & y_{j+1} & x_{i+1}y_{j+1}\\
    \end{pmatrix}
    \begin{pmatrix}
        a_0 \\ a_1 \\ a_2 \\ a_3
    \end{pmatrix}
     = 
     \begin{pmatrix}
         f(x_i,y_{j})\\
         f(x_i,y_{j+1})\\
         f(x_{i+1},y_j)\\
         f(x_{i+1},y_{j+1})\\
     \end{pmatrix}.
\end{equation}
Finally, for $(x,y) \in d_{(i,j)}$ for some $i$ and $j$, the interpolation function is defined as $\bar{f}(x,y) = \bar{f}_{(i,j)}(x,y)$.

In spite of the name, bilinear interpolation yields not a linear, but a quadratic function. 
Figure~\ref{fig:bilinear} shows an image where the velocity field is interpolated for a particle with arbitrary location $\vec{x}_a$.
\begin{figure}
   \centering
    \begin{tikzpicture}
    \def\dx{2.5}
    \draw (\dx/2,\dx/2) circle  (0.01);
    \draw[step=\dx,black,thin,shift={(-\dx/2,-\dx/2)}] (0,0) grid (\dx*2,\dx*2);
    \draw[->,dashed]  (0,0) node[below]{$c_{(i,j)}$}-- (0.8,-1.1);
    \draw[->,dashed]  (\dx,0) node[below]{$c_{(i+1,j)}$}-- (\dx + 0.6,0-0.3);
    \draw[->,dashed]  (0,\dx) node[below]{$c_{(i,j+1)}$}-- (0+1.2,\dx + 1.1);
    \draw[->,dashed]  (\dx,\dx) node[below]{$c_{(i+1,j+1)}$}-- (0+\dx + 1,\dx + 0.3);
    \draw[->,dashed,red] (\dx*0.3,\dx*0.6) node[below]{$\vec{x}_{a}$} -- (\dx*0.3 + 0.98,\dx*0.6 + 0.46);
    \draw[dotted,red] (\dx*0.3,-\dx/2) node[below]{$x_a$} -- (\dx *0.3,\dx*3/2);
    \draw[dotted,red] (-\dx/2,\dx*0.6) node[left]{$y_a$} -- (\dx*3/2,\dx*0.6);

    \draw[dotted] (0,-\dx/2) node[below]{$x_i$} -- (0,\dx*3/2);
    \draw[dotted] (-\dx/2,0) node[left]{$y_j$} -- (\dx*3/2,0);

    \draw[dotted] (\dx,-\dx/2) node[below]{$x_{i+1}$} -- (\dx ,\dx*3/2);
    \draw[dotted] (-\dx/2,\dx) node[left]{$y_{j+1}$} -- (\dx*3/2,\dx);
    \end{tikzpicture}
    \caption{Example of bilinear interpolation of a velocity field.}
    \label{fig:bilinear}
\end{figure}
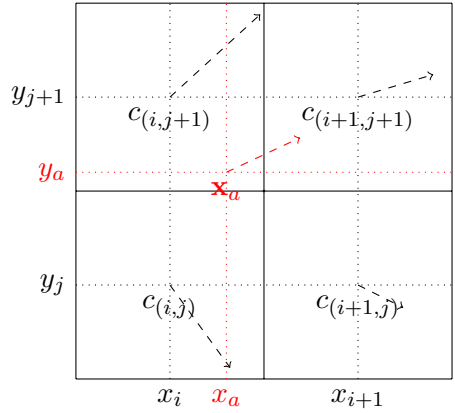
\section{Discussion}
In this chapter we discussed how microscopic measures like mass, position and velocity can be converted to macroscopic measures like density and velocity fields. 
In the next chapter, we use these techniques to simulate pedestrian dynamics on multiple scales. 
We introduce an implementation of a particle model as defined in Section~\ref{sec:micro} and model the interaction on a global level by obtaining a macroscopic measures of the crowd. 
The effects of this global interaction are evaluated for each of the pedestrians using bilinear interpolation.

\chapter{Application: Crowd dynamics in inhomogeneous domains}
\label{chap:crowds}
% Introduction and literature review Crowd Dynamics
\begin{quote} 
    \centering 
    \emph{Thus the scientist hopes that from an objective study of the actual ways that we human beings do in fact behave, he may disclose the nature of the underlying principles that govern our conduct.}

    \raggedleft
    - George Kingsley Zipf, \emph{Human Behaviour and the Principle of Least Effort}.
\end{quote}

\section{Introduction}
% General definition of crowd dynamics.
Crowd dynamics is the field of research dedicated to studying the behaviour and interaction of groups of people in motion. 
Highly multifaceted, it is of interest to researchers from many different areas.
The ubiquitous nature of human crowds creates a universal need for their understanding not only in science, but also in applied fields.
Architects designing new urban environments benefit greatly from insight in crowds, which they can use to assess the level of safety and improve evacuation procedures. 
The same holds for traffic engineers working on infrastructures in large cities, calling for the avoidance congestions and hazardous situations. 
Another example presents itself in the entertainment industry, where film and game developers require a convincing virtual environment, including realistic crowd animations.\\
Many more applications exist. Yet we hope these suffice to get the point across: the world is packed with crowds.

Current understanding of the quantitative aspects in crowds is quite limited. For the reasons mentioned above, modelling crowds is receiving an increasing amount of attention in the scientific community. 

We provide a literature review discussing the most popular crowd modelling techniques used in the last decades. 
In doing so we create a context for our multiscale model in order to compare it to existing techniques. 
It is shown that employing a multiscale model has several recognised benefits in analysing crowd dynamics.
We tailor our framework with features specific to human crowds and use it on several test cases. 
Afterwards, the results are compared to observed phenomena found in experimental studies and we evaluate our implementation.

\section{Crowd modelling approaches}
Correctly predicting the behaviour of an individual is not a trivial task. By extension, correctly predicting the motions and interactions of many individuals is a challenge, especially since the nature of crowds is as diverse as the individuals they are composed of. A group of people might show different interactions depending on their location, time of day, state of mind, cultural habits, etcetera.
As a consequence, many different models have been developed to analyse and predict different types of crowds. At the time of writing, to our knowledge no unified and validated crowd dynamic model exists.
Nevertheless, a lot of progress has been made in developing new models and improving the accuracy of existing ones. 
To model the behaviour of a crowd, one needs at least four components.
\begin{itemize}
    \item \textbf{Representation of people}\\
        The model should capture the dynamic effects of a moving crowd, be it on an individual or a global level.
    \item \textbf{Scene representation}\\
        Crowds interact with their environment, so in modelling crowd it is often necessary to model their environments as well.
    \item \textbf{Route planning}\\
        Very often, a crowd has a (or more than one) destination. Models should incorporate means for the people in the crowd to reach this destination.
    \item \textbf{Interaction prescription}\\
        The limited free space in a crowd makes interaction inevitable. While is probably the most complex aspect of crowd dynamics, any model should incorporate it.
\end{itemize}

When reviewing scientific contributions in crowd dynamics (like \cite{pettre14}, \cite{hoogendoorn15}, \cite{pietschmann13}), there seems to be a widely accepted classification of models into two categories microscopic and macroscopic models.
A recent review in \cite{zheng09} provides a subdivision into seven commonly used models. We limit ourselves to the three more dominant and mathematically-based ones:
\begin{enumerate}
	\item Cellular automata
	\item Particle models
	\item PDE-like models
\end{enumerate}
\section{Cellular automata}
The \worddef{cellular automaton model} (CA),  also called lattice model, has been proposed by Von Neumann around 1940. 
According to the definitions in \cite{sayama15}, a cellular automaton is defined as a theoretical machine, usually a cell in a rectangular grid, that changes its internal state depending on input and its previous state. 
A collection of cellular automata creates a spatially distributed dynamical system, discrete in both time and space.
A CA model has several beneficial aspects. For instance, emergent behaviour can already be observed with a simple set of rules. 
Also, some rule sets have the property to show consistent and converging behaviour when letting the grid size in time and space go to zero.
This allows for mathematical analysis of the CA, the prediction of limiting behaviour and sometimes an inferred system of equations representing macroscopic behaviour.
Lastly, a CA model is quite easy to implement and has a low computational cost in comparison to other models.

Probably the best known cellular automata is John Conway's 'Game of Life' (proposed in \cite{conway70}), a deterministic simulation game where cells live or die depending on their neighbour cells. The Game of Life somewhat resembles the rise and fall of a population. From there, it is a small leap to invent a set of rules that simulate population motion and interaction.

\begin{figure}[h!]
	\centering
	\includegraphics[width=0.3\textwidth]{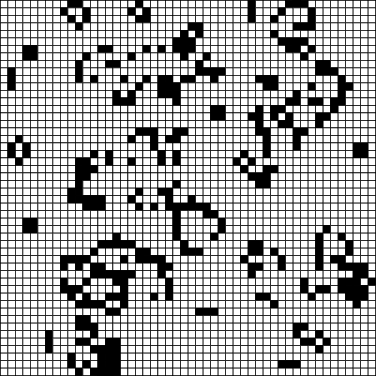}
	\caption{Example of a generation in the CA simulation Game of Life.}
	\label{fig:gameoflife}
\end{figure}
One of the first applications of a CA model for pedestrian dynamics was proposed in \cite{blue00}.
In this model, pedestrians are represented by occupied cells on the lattice. 
Whenever a pedestrian moves to a new cell, the cell he departs becomes unoccupied. Every cell holds at most one pedestrian.\\
The main focus of this simulation was the observation of bidirectional flow and the formation of lanes. 
They were able to differentiate in pedestrian speed by varying the number of cells a pedestrian moves each time step. 
While congestion and acceleration/deceleration were not that realistic,
the implementation of this model shows some emergent behaviour culminating in sidestepping and lane formation. 
However, it should be noted lane formation was the phenomenon sought after.

Cellular automata have remained a popular way for modelling crowds, because of their extendibility and relative ease of implementation.
For instance, in \cite{bandini11} a CA is proposed in which pedestrians are assigned certain goals. Pedestrians are equipped with an observation fan, a means of perceiving the surroundings in front of them, and are able to avoid potential obstacles and other pedestrians. 
They also extended the stochastic nature of the CA presented in \cite{blue00}, accounting for group cohesion and geometric and proxemic repulsion.
\cite{sarmady10} describes an extension in which pedestrians take up several cells, thereby providing a way to model various crowd densities and different speeds.
It must be noted that for detailed microscopic analysis, most CA models are not capable of accurately describing real life crowd dynamics.  
Most CA's lack directional flexibility because of the spatial discretisation on a lattice. This also results in a fixed speed for all pedestrians and rigid maximum densities.

Despite these facts, some CA models lend themselves particularly well for mathematical analysis.
\cite{pietschmann13} describes a CA inhabited by two species with rules modelling pedestrian dynamics and social cohesion, where state transitions are based on probabilities.
In \cite{cirillo13}, a model is implemented and analysed where pedestrians experience limited visibility while being evacuated. By including social factors like cooperation, they discovered altruism can very well induce disasters in evacuation.
\cite{burger10} describes a similar model and proceeds to develop a corresponding system of limit equations, thereby analysing the microscopic model on a macroscopic level. 
This allows for a derivation of stability conditions and, up to a certain level, validate the results of the corresponding Monte Carlo simulations.
We provide a short summary of their derivations, since the model examined is representative for many crowd dynamics lattice simulation. More important, it shows a (different) way to couple microscopic and macroscopic models together with the transition from a discrete to a continuous system.
%Define size exclusion
\subsection{Modelling interaction and social dynamics}
We define a CA on a two-dimensional grid of size $m_x\times m_y$ with square cells of size $h\times h$. We assume time step $\Delta t$, at which all cells simultaneously attain their new state.
We introduce particles (pedestrians) which occupy a subset of the grid cells, according to a certain initial condition. Each grid cell holds at most one pedestrian. 
Each time step, pedestrians are able to move to an adjacent cell in horizontal or vertical direction.

To model route planning, floor fields are introduced, driving forces which determine the probabilities of moving to other cells. Each pedestrian has a static floor field $S\in\mathbb{R}^{m_x\times m_y}$ which does not change throughout the simulation, and a dynamic floor field $D\in\mathbb{R}^{m_x\times m_y}$ which is updated on every time step. 
The static field $S$ represents the scene and accounts for the environment, like obstacles, entrances, exits, etcetera. Field $S$ has lower values for cells closer to the exits and high values for cells surrounding obstacles.
It is a discrete representative of a domain potential function as discussed in Section~\ref{sec:domain_pots}.
%It is the discrete representative of the potential field planner discussed in Section~\ref{sec:potential_planner}.

The dynamic field $D$ represents the pedestrian interaction and accounts for group behaviour effects like herding and following of other group members, modelling forms of attraction and repulsion. Field $D$ is initialized with zero values.

This method is extensible to multiple types of pedestrians. 
Each group of pedestrians is assigned one goal and has a social cohesion with other group members. In particular, each group has its own static and dynamic fields. In this implementation two groups are considered, labelled red (${r}$) and blue (${b}$), each with their own static field $S_r$/$S_b$ and dynamic field $D_r$/$D_b$. 
Whenever a cell becomes unoccupied, its dynamic field value increases to account for the popularity of that cell. Also, each time step, all positive dynamic floor field values decrease with a certain rate $\delta$. This approach is also found in models of swarm intelligence, as a way to model ant trails (\cite{watmough95}).

These fields are used to determine the probability of a pedestrian to move into a neighbouring grid cell $(i,j)$ and is expressed as follows:
\begin{equation}
	(P_r)_{i,j}=(N_r)_{i,j}e^{(\kappa_D(D_r)_{i,j})}e^{(\kappa_S(S_r)_{i,j})}(1-c_{i,j}).
	\label{eq:ca_prob}
\end{equation}
In this expression, $\kappa_D$ and $\kappa_S$ are weight factors for the floor fields. The term $c_{i,j}$ has value 1 if cell $(i,j)$ is occupied, and zero otherwise, to ensure each cell holds at most one pedestrian per time step. 
Finally, $N_r$ is a normalizing factor such that all jump-probabilities sum to 1.

\subsection{Deriving macroscopic equations}
Macroscopic equations are derived for this model in one dimension, along the lines of \cite{burger10}.
We combine the static and dynamic fields into one potential function per species: $\Phi_r$ and $\Phi_b$.
Computation takes place in spatial domain $L = [0,1]$. Domain $L$ is partitioned into $n$ cells with cell size $h$ such that $nh=1$. The potential fields and probabilities are defined on $L\times [0,T]$ with $T$ the maximum computation time. Let $\Delta t$ be the time step.
We define diffusion parameters $\alpha_l$ and $\alpha_r$ and flux parameters $\beta_l$ and $\beta_r$.
Below follows a derivation for red particles (in which we omit subscript $r$). Blue particles follow analogously.
With these parameters, right and left flux rates $\gamma^+$ and $\gamma^-$ are defined as
\begin{equation}
    \begin{split}
    \gamma^+ = \alpha - h\beta\partial_x \Phi(x+h/2,t),\\
    \gamma^- = \alpha + h\beta\partial_x \Phi(x+h/2,t).
    \end{split}
\end{equation}
Potential function $\Phi$ is scaled such that $0\leq\gamma^-,\gamma^+\leq 1$.
High $\alpha$ increases the randomness of pedestrian motion, thereby increasing diffusion. High $\beta$ increases the weight of the potential function, prioritizing the flux to the exit.
It is assumed right and left hopping is only possible when the destination cell is unoccupied. 
Let $r(x,t)$, $b(x,t)$ denote the probability that cell $x$ holds a red or blue particle pedestrian at time $t$.
The hopping probabilities are expressed as:
\begin{align}
    \begin{split}
    \Pi^+(x,t) = \gamma^+P(\text{cell $x+h$ free}) = \gamma^+(1-r(x+h,t)-b(x+h,t)),\\
    \Pi^-(x,t) = \gamma^-P(\text{cell $x-h$ free}) = \gamma^-(1-r(x-h,t)-b(x-h,t)).
    \end{split}
    \label{eq:hop_prob_def}
\end{align}
Considering that a pedestrian moves at most one cell, $r(x,t+\Delta t)$ can be computed from probabilities in $t$ with
\begin{align*}
    r(x,t+\Delta t) = r(x,t) \left( 1- \Pi^+(x,t) - \Pi^-(x,t) \right),\\
    +r(x+h,t)\Pi^-(x+h,t) + r(x-h,t)\Pi^+(x-h,t).
\end{align*}
This expression is rewritten to facilitate the substitution of Taylor expressions 
\begin{align*}
    r(x,t+\Delta t) - r(x,t) = r(x,t)\left(\Pi^+(x-h,t) + \Pi(x+h,t)-\Pi^+(x,t)-\Pi^-(x,t)\right),\\
    + (r(x+h,t)-r(x,t))\Pi^-(x+h,t)+(r(x-h,t)-r(x,t))\Pi^+(x-h,t),
\end{align*}
and after substitution, this becomes
\begin{align}
    \begin{split}
        r(x,t+\Delta t) - r(x,t)= \\
        r(x,t)\left( h(\partial_x\Pi^-(x,t)-\partial_x\Pi^+(x,t)) + \frac{h^2}{2}\left( \partial_{xx}\Pi^+(x,t) + \partial_{xx}\Pi^-(x,t) + \bigo{h^3}\right)\right).
    \end{split}
    \label{eq:ca_almost_done}
\end{align}
Using the probability definitions from \eqref{eq:hop_prob_def} we obtain
\begin{align*}
    \partial_x\Pi^-(x,t) - \partial_x\Pi^+(x,t) = 2h\beta(\partial_x\Phi(1-r-b)) + 2h\alpha(\partial_{xx}r+\partial_{xx}b)+\bigo{h^2},\\
    \partial_{xx}\Pi^+(x,t) + \partial_{xx}\Pi^-(x,t) = -2\alpha(\partial_{xx}r+\partial_{xx}b+\bigo{h},\\
    \Pi^-(x+h,t) - \Pi^+(x-h,t) = 2h\beta\partial_{xx}\Phi(1-r-b)+\bigo{h^2},\\
    \Pi^-(x+h,t) - \Pi^+(x-h,t) = 2\alpha(1-r-b) + \bigo{h},
\end{align*}
and substituting this in ~\eqref{eq:ca_almost_done} yields
\begin{equation*}
    r(x,t+\Delta t)-r(x,t)=2h^2\beta\partial_x(r(1-r-b)\partial_x\Phi(x,t))+h^2\alpha\partial_x\left((1-r-b)\partial_xr + r(\partial_xr+\partial_xb) \right).
\end{equation*}
Applying the scalings $\frac{2h^2}{\Delta t} = 1$, $\frac{\alpha}{2}=D$, $\frac{2\beta}{\alpha}=\mu$ and letting $\Delta t \to 0$ the system results in:
\begin{equation}
    \begin{split}
    \partial_t r = D_r\partial_x \left( (1-b)\partial_x r + r\partial_xb+\mu_1r(1-r-b)\partial_x\Phi_r \right),\\
    \partial_t b = D_b\partial_x \left( (1-r)\partial_x b + b\partial_xr+\mu_2b(1-r-b)\partial_x\Phi_b \right).
    \end{split}
    \label{eq:macro_equations}
\end{equation}

\eqref{eq:macro_equations} represents a coupled system of advection-diffusion equations, including cross-diffusion.
Further examining these equations, we see the first term represents a cross-diffusion coefficient, the second term represents the flux due to the concentration of blue particles, while the third term represent the flux due to the domain potential.
\section{Particle models}
Perhaps the most popular way of modelling crowds is by prescribing the behaviour for each of the pedestrians individually by viewing them as particles, and let the global behaviour emerge from motion and interaction imposed at the individual level.
Both particle models and CA's are defined at a microscopic level, but particle models differ from CA's in their discretisation. 
Where CA's are defined on a lattice, particle models are completely continuous in space. 
Often, pedestrians are modelled as fixed-radius particles in a two-dimensional space representing the scene.
Each time step, new pedestrian positions are computed from either a (net) force field or a velocity field. These fields follow from a set of rules and/or a governing equation modelling pedestrian interaction and route planning. An illustration is provided in Figure~\ref{fig:particle_model}.
\begin{figure}[h!]
    \centering
    \includegraphics[width=0.3\textwidth]{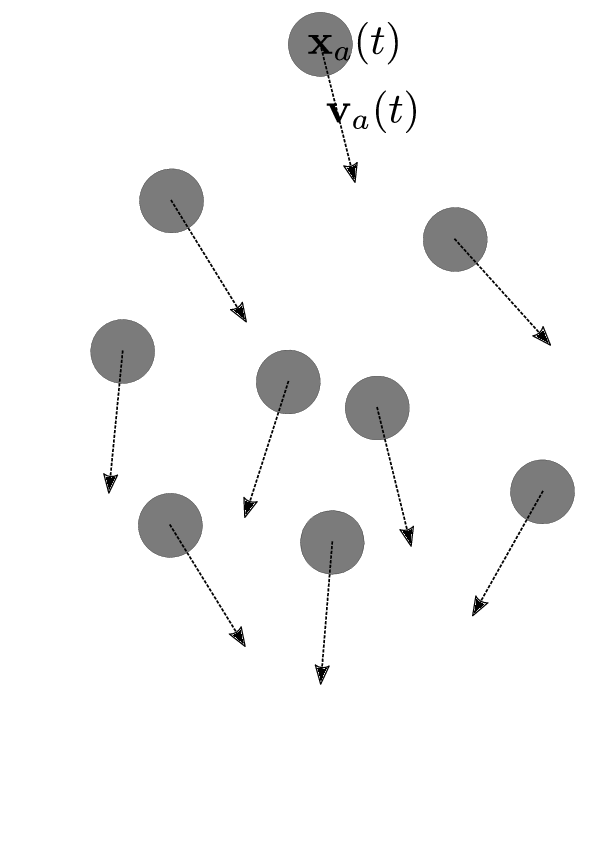}
    \caption{Example of generic particle model. In this instance, positions are updated according to a time- and pedestrian-dependent velocity field $\vec{v}(t)$.}
    \label{fig:particle_model}
\end{figure}

While the idea sounds fairly intuitive, researchers have struggled with reproducing realistic and crowd-like motions for a long time.  One of the first successful attempts can be attributed to Reynolds in \cite{reynolds87}, published in the late eighties. 
In his attempt to model flocks of animals (specifically birds), he establishes rules which apply for many coherent groups of agents, including pedestrians. He proposes the Boid model (acronym for bird-oids, agents with birdlike flock behaviour) based on three principles: 
\begin{enumerate}
    \item \textbf{Collision avoidance}\\
		Any bird must attempt to keep a safe distance from its fellow birds.
    \item \textbf{Velocity matching}\\
		Any bird will adapt to the velocity of its direct neighbours.
    \item \textbf{Flock centring}\\
		Any bird will attempt to stay as close to its neighbours without causing collisions.
\end{enumerate}
Reynolds main contribution is modelling a flock as a group of conforming agents, as individuals aware of their neighbours. 
Both researchers and designers successfully used these principles as a basis for more realistic crowd motion.
However, in these models it is difficult to express individual characteristics.
For this reason, many models since have taken a more mechanical approach, modelling specific preferences with different velocities or potential functions.
\subsection{Social force model}
\label{sec:social_force}
An influential example is the social force model, described in \cite{helbing95}: a particle model that expresses social features in a physical manner.
These pedestrians have a pre-set goal, and try to reach that goal while maintaining their desired speed. While approaching the goal, they experience attraction and repulsion forces from other pedestrians and objects, depending on their personal preference. This way, both social behaviour is simulated as well as collision avoidance. 

This model can be expressed as a coupled system of adapted Langevin equations:
Let $\vec{x}_\alpha(t)\in \mathbb{R}^2$ be the position of pedestrian $\alpha$ at time $t$. Then the system governing his motion is given by

\begin{align*}
    \frac{d\vec{v}_\alpha(t)}{dt} = \vec{F}_\alpha(t) + B_\alpha(t),\\
    \frac{d\vec{x}_\alpha(t)}{dt} = g(\vec{v}_\alpha(t)).
\end{align*}
Here $B_\alpha(t)$ corresponds to a stochastic fluctuation, for instance a Wiener process. 
$\vec{F}_\alpha$ corresponds to the sum of all attraction and repulsion forces for pedestrian $\alpha$ and $g$ limits the velocity such that
\begin{equation*}
    g(\vec{x}) = \frac{\vec{x}}{||\vec{x}||}\min(||\vec{x}||,v_{\max}).
\end{equation*}
Implementations of this model show typical crowd behaviour, like congestions around bottlenecks and lane formation. 
Compared to the CA model in \cite{blue00}, these phenomena emerge more naturally from the model, considering one only prescribes a set of attractors and repellers.

The social force model and many of its derivatives tend to be computationally expensive. While moving the pedestrians can be done in linear time, computing the influence of (long-range) attraction and repulsion forces from other pedestrians takes quadratic time. This shows an important limitation in simulation of large crowds.

Recently, a new type of model has emerged; the so-called velocity-based model.
Velocity-based models prescribe interactions between pedestrians not only based on positions, but on velocities as well. 
In theory, including velocity in pedestrian interaction yields smoother dodging manoeuvres and allows a sharper definition of the neighbourhood of interaction.
An example is the Paris model described in \cite{pettre14}, in which in each time step a pedestrian's admissible and inadmissible velocity domain is computed. 
Inadmissible velocities might lead to collisions with other pedestrians or obstacles, or exceed a maximum speed. 
The pedestrians picks the admissible velocity closest to its desired velocity and advances to its next position. Because of the extra degrees of freedom, velocity-based models are in general even more computationally expensive than social force models.
\section{PDE-like models}
\label{sec:macro_literature}
When crowds become large, one often prioritises the motion and behaviour of the collective over the motion of any individual. This gives rise to a macroscopic representation of a crowd; not as a collection of individuals, but rather as a body in motion.
Crowd behaviour in this case is not expressed by individual trajectories, but by changes in densities and velocity fields.

Hughes describes a rather complete macroscopic theory of pedestrian flows in \cite{hughes02}.
His theory is based on the continuity equation complemented with a potential function for modelling pedestrian discomfort at high crowd densities.
We elaborate on his contribution. It compares nicely to the advection-diffusion system obtained in \cite{burger10} and is a straightforward way to derive a macroscopic equation describing crowd flow from a balance principle.

\subsection{Derivation of pedestrian transport equations}
Pedestrian flows are characterised fully by density $\rho(x,y,t)$ and velocity $v(x,y,t) = \begin{pmatrix} v_x\\v_y \end{pmatrix}$ in location $(x,y)$ at time $t$. Using these quantities, we express the continuity equation:
\begin{equation}
    \frac{\partial\rho}{\partial t} + \div{\rho v}=0.
	\label{eq:cont_eq}
\end{equation}
More assumptions are made specific to pedestrians flow:
\begin{enumerate}
	\item Pedestrian speed depends only on density.
	\item Pedestrians have a common and complete sense of their surroundings and locations in the form of a potential.
	\item Pedestrians try to minimise their travel time, but subvert this behaviour to avoid high density regions.
\end{enumerate}
Assumption 1 is formulated by introducing a speed function $f(\rho)$. Using direction cosines $\hat{\phi}_x$ and $\hat{\phi}_y$, the velocity is expressed as
\begin{equation}
	v_x = f(\rho)\hat{\phi}_x\quad	v_y = f(\rho)\hat{\phi}_y.
	\label{eq:hyp_1}
\end{equation}
Assumption 2 is included by defining a domain potential $\phi$ such that the motion of the pedestrians is opposite to the gradient of this potential. Therefore, the direction cosines are formulated as
\begin{equation}
	\hat{\phi}_x = \frac{-\frac{\partial\phi}{\partial x}}{\sqrt{\left( \frac{\partial\phi}{\partial x} \right)^2+\left( \frac{\partial\phi}{\partial y} \right)^2}},
	\quad
	\hat{\phi}_y = \frac{-\frac{\partial\phi}{\partial y}}{\sqrt{\left( \frac{\partial\phi}{\partial x} \right)^2+\left( \frac{\partial\phi}{\partial y} \right)^2}}.
	\label{eq:hyp_2}
\end{equation}
Assumption 3 is modelled by assuming the avoidance behaviour as a function $g(\rho)$. This function is unity for all $\rho$ below some threshold value $\bar{\rho}$ but becomes large for $\rho > \bar{\rho}$. The following relation is postulated:
\begin{equation}
	\frac{1}{\sqrt{\left( \frac{\partial\phi}{\partial x} \right)^2+\left( \frac{\partial\phi}{\partial y} \right)^2}}
	= g(\rho)\sqrt{v_x^2+v_y^2}.
	\label{eq:hyp_3}
\end{equation}
Substituting expressions \eqref{eq:hyp_1}-\eqref{eq:hyp_3}, the following governing equations are obtained:
\begin{equation}
	\frac{\partial\rho}{\partial t}=\frac{\partial}{\partial x}\left( \rho g(\rho)f^2(\rho)\frac{\partial\phi}{\partial x} \right)+
	\frac{\partial}{\partial y}\left( \rho g(\rho)f^2(\rho)\frac{\partial\phi}{\partial y} \right),
	\label{eq:gov_pde a}
\end{equation}
with
\begin{equation}
	g(\rho)f(\rho) = \frac{1}{\sqrt{\left( \frac{\partial\phi}{\partial x} \right)^2+\left( \frac{\partial\phi}{\partial x} \right)^2}}.
	\label{eq:gov:pde_b}
\end{equation}
It is stated that boundary conditions of this system depend on the geometry. Walls, entrances and exits are modelled by prescribing at those locations respectively fluxes, source terms, or sink terms.

Furthermore, Hughes analyses properties of solutions to this system and performs a perturbation analysis. 
He introduces multiple pedestrian types, characterised by their goal (and therefore their potential field). 
This model is applied to the Jamarat Bridge, a religious attraction near Mecca and compared to empirical data.

\section{A macro-micro discussion}
From the discussion above, we would like to emphasise the following observations:
\begin{itemize}
    \item A lattice model is a convenient way to roughly describe the behaviour of a crowd. 
    While its inherent discrete nature often results in rigid and coarse simulations, implementation and mathematical evaluation is relatively painless.
    \item More detailed crowd behaviour is reproduced with particle models. In principal, modelling possibilities are endless. On the other hand, with complexity comes computation time, increasing the difficulty of simulating large crowds at interactive rates.
    \item Performance of macroscopic models does not decline for large numbers of pedestrians, although it is no longer possible to input or extract individual pedestrian characteristics. 
    On top of that, it is difficult to correctly represent regions of low density.
\end{itemize}

Several of these drawbacks are alleviated using a combination of micro- and macroscale modelling. 
As we noted before, it is not the propagation of the particles that is computationally expensive, but the interaction and route planning.
The approach covered in Chapter~\ref{sec:micro_macro} allows for various microscale modelling benefits, like implementing different pedestrian traits in one crowd. But contrary to microscale models, pedestrian interaction is evaluated on a global scale, thereby bypassing the quadratic time evaluation constraint.
This modelling approach has been applied in crowd dynamics as well and we shall base our approach on some recent contributions.

% Derivation and results of the domain potential part
\section{Implementation of a multiscale model: Part I}
\label{sec:crowds1}
We pose a general system of governing equation based on the discussion on domain potentials in Section~\ref{sec:macro:domain_pots}.

\begin{equation*}
    \frac{\partial\rho}{\partial t} = \div{ \rho \nabla \Phi },
\end{equation*}
where we let $\varphi$ be defined by the solution of the equation
\begin{equation}
    ||\nabla\varphi(\vec{x})||:= u(\vec{x}),
    \label{eq:potential_from_unit_cost}
\end{equation}
for some $u(\vec{x})$ representing the marginal walking cost at location $\vec{x}$.

$l(\vec{x})$ is a limiting function to represent the finite pedestrian speed $v_{\max}$:
\begin{equation*}
    l(\vec{x}) = \frac{\vec{x}}{||\vec{x}||}\max\left( v_{\max},||\vec{x}|| \right).
\end{equation*}
%The paper \cite{treuille06} describes a particle model in a heterogeneous domains with exits as goals. 
%By means of a potential function, the particles are steered towards the goal while dodging the obstacles and regions of high density.
\subsection{Domain potential-based transport}
\label{sec:potential_planner}
In \cite{treuille06}, a potential function is proposed modelling both domain aspects as well as interaction aspects.
It is based on the principle of least effort: each time step it computes the direction minimizing the walking cost, found by computing the steepest descent of a potential function.

Let $\Omega \subset \mathbb{R}^2$ be the simulation domain and $G \subset \Omega$ be a pedestrian goal. We assume position dependent maximum speed $f(x,\theta)$ with position $x\in \Omega$ and angle $\theta \in [0,2\pi]$. 
Without any impediments, pedestrians attain maximum walking speed.
To model spatial preferences a discomfort field is introduced. This field represents the \worddef{walkability} of the area. For instance, the preference of walking on side-walk instead of the main road could be modelled by increasing the discomfort field values on the main road.
This field can be coupled to the density field to model aversion to locations of high density.

To account for the delay experienced when a pedestrian walks through a group of people with another direction, the maximum speed field depends on the observed density.
Let $\vec{n}_\theta$ be the unit vector in direction $\theta$. 
Let $\rho(x)$ be the observed density at location $x$, in interval $[\rho_{\min},\rho_{\max}]$, and let $\bar{f}^+$ and $\bar{f}^-$ be respectively the maximum and minimum pedestrian speed for pedestrian.\\

To scale and cut off the density values between their threshold values, we define piece-wise linear function $L$ as
\begin{equation}
    L_{a}^{b}(t) = \begin{cases}
        0&\mbox{ if }t<a\\
        \frac{t-a}{b-a}&\mbox{ if }a\leq t \leq b\\
        1&\mbox{ if }t>b
    \end{cases}.
    \label{eq:cutoff}
\end{equation}
Then the anisotropic (directionally dependent) maximum speed field $f(\vec{x},\theta)$ is defined as
\begin{equation*}
    f(\vec{x},\theta) = \bar{f}^+ + L_{\rho_{\min}}^{\rho^{\max}}\left( \rho(\vec{x}+r\vec{n}_\theta) \right)(\bar{f}^- - \bar{f}^+).
\end{equation*}
The density is not observed at pedestrian location $\vec{x}$ but rather slightly $r$ ahead, at $\vec{x}+r\vec{n}_\theta$ to account for pedestrian vision and anticipation.
\subsection{Computing the walking cost}
The walking cost depends on the maximum speed $f$ and discomfort $g$. 
We use $g$ to account for pedestrian avoidance of high density regions and locations close to obstacles.
In our simulations, we choose
\begin{equation}
    g(\vec{x}) = g_{\textrm{obs}}(\vec{x}) + g_{\textrm{dens}}(\vec{x}).
    \label{}
\end{equation}
$g_{\textrm{obs}}(\vec{x})$ represents the discomfort from moving too close to any of the $n$ obstacles $M_1,\dots,M_n$. It is defined such that 
\begin{equation}
    g_{\textrm{obs}}(\vec{x}) = 
    \begin{cases}
        1&\mbox{ if } d(\vec{x},M_i) < r \mbox{ for } i=1,\dots,n\\
        0&\mbox{ else }
    \end{cases}.
    \label{}
\end{equation}
Here $r$ can be chosen arbitrarily small, depending on the nature of the pedestrian simulations. In practice, we choose $r$ the size of one cell.

$g_{\textrm{dens}}(\vec{x})$ represents the discomfort a pedestrian experiences when moving through high densities. It is defined as 
\begin{equation}
    g_{\textrm{dens}}(\vec{x}) = L_{\rho_{\min}}^{\rho^{\max}}\left( \rho(\vec{x}+r\vec{n}_\theta) \right).
    \label{}
\end{equation}

Under these parameters, a pedestrian at location $\vec{x}$ picks the path $P$ (a curve in $\Omega$) that minimizes a total walking cost:
\begin{equation}
    P = \argmin_P \underbrace{\left(\alpha\int_P1d\vec{s}\right.}_{\text{path length}} + \underbrace{\beta\int_P \frac{1}{f}d\vec{s}}_{\text{time}} + \underbrace{\left.\gamma\int_P\frac{g}{f}d\vec{s}\right)}_{\text{discomfort}},
    \label{eq:walking_cost}
\end{equation}
where $\alpha,\beta,\gamma \in \mathbb{R}$ are weight constants.
In evaluating discomfort, the discomfort field $g$ is divided by the speed field $f$ to account for the time of discomfort spent, instead of the distance.
This is a modelling choice and is based on the idea that pedestrians prefer to minimise the time spent in discomfort, rather than the distance.

Defining $u$ as 
\begin{equation}
    u(\vec{x},\theta) := 
    \begin{cases}
        \frac{\alpha f(\vec{x},\theta) + \beta + \gamma g(\vec{x})}{f(\vec{x},\theta)}&\mbox{ if } \vec{x} \notin M_i \mbox{ for }i=1,..,n\\
        \infty&\mbox{ if } \vec{x} \in M_i \mbox{ for }i=1,..,n
    \end{cases},
    \label{eq:unit_cost}
\end{equation}
an expression for the walking cost field is obtained. From this, a domain potential $\varphi$ is computed satisfying
\begin{equation}
    || \nabla \varphi(\vec{x}) || = \min_\theta u(\vec{x},\theta).
    \label{eq:Eikonal_equation}
\end{equation}
This equation is called the Eikonal equation. It has no known analytical solution for general $u(x)$. \\
Fast algorithms exist for grid based approximations, most notably the fast marching algorithm described in \cite{tsitsiklis95}, commonly used for computing distance transforms in images. 
\subsection{Discretisation schemes}
We assume the discretisation discussed in Section~\ref{sec:discretisation}.\\
The discrete representative of discomfort field $g(x)$ and potential field $\varphi(x)$ is given by
\begin{align*}
    \vec{g}_{i,j} &=\frac{1}{\Delta x\Delta y}\int_{c_{i,j}} g(x)dx,\\
    \gvec{\varphi}_{i,j} &=\frac{1}{\Delta x\Delta y}\int_{c_{i,j}} \varphi(x)dx.
\end{align*}
We reduce the set of possible angles to $\left\{0,\frac{1}{2}\pi,\pi,1\frac{1}{2}\pi\right\}$, corresponding with the normal vectors of the cell edges.\\
We use fractional indexing to indicate cell edges.
Let $(\bar{i},\bar{j})_\theta := (i+\frac{\vec{n}_{\theta,x}}{2},j+\frac{\vec{n}_{\theta,y}}{2})$ where $\vec{n}_\theta$ corresponds with the unit vector in direction $\theta$ (normal to the cell edge).\\
Speed field $f(\vec{x},\theta)$ and unit cost field $u(\vec{x},\theta)$ are discretised as follows.
\begin{align*}
    % F_{i+\frac{n_{\theta,x}}{2},j+\frac{n_{\theta,y}}{2},\theta} &= f((i-n_{\theta,x}/2)\Delta x,(j-n_{\theta,y}/2)\Delta y,\theta)\\
    % U_{i+\frac{n_{\theta,x}}{2},j+\frac{n_{\theta,y}}{2},\theta} &= u((i-n_{\theta,x}/2)\Delta x,(j-n_{\theta,y}/2)\Delta y,\theta)
    \vec{f}_{(\bar{i},\bar{j})_\theta} &= f((i-\vec{n}_{\theta,x}/2)\Delta x,(j-\vec{n}_{\theta,y}/2)\Delta y,\theta),\\
    \vec{u}_{(\bar{i},\bar{j})_\theta} &= u((i-\vec{n}_{\theta,x}/2)\Delta x,(j-\vec{n}_{\theta,y}/2)\Delta y,\theta).
\end{align*}
The relation between cell edges and index fractions is illustrated in Figure~\ref{fig:fractional_indexing}.
\begin{figure}
    \centering
	\begin{tikzpicture}
		\draw[step=2cm,black,thin] (2,0) grid (4,6);
		\draw[step=2cm,black,thin] (0,2) grid (6,4);
        \filldraw (3,1) circle (1pt) node[anchor=north] {$c_{i,j-1}$};
        \filldraw (1,3) circle (1pt) node[anchor=north] {$c_{i-1,j}$};
        \filldraw (3,3) circle (1pt) node[anchor=north] {$c_{i,j}$};
        \filldraw (2,3) circle (1pt) node[anchor=south] {$c_{i-\frac{1}{2},j}$};
        \filldraw (3,2) circle (1pt) node[anchor=north] {$c_{i,j-\frac{1}{2}}$};
	\end{tikzpicture}
    \captionof{figure}[Caption]{Staggered grid discretisation with fractional indices.}
   \label{fig:fractional_indexing}
\end{figure}
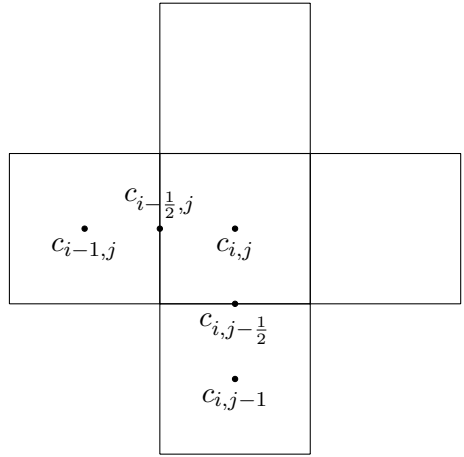
\subsection{Solving the Eikonal equation}
This algorithm uses an adapted first order upwind discretisation scheme to approximate potential $\varphi$ from \eqref{eq:potential_from_unit_cost} in each grid cell:
\begin{equation}
    \left( \frac{\gvec{\varphi}_{i,j}-\gvec{\varphi}_{i+m_x,j}}{\vec{u}_{\Theta}(i-\frac{m_x}{2},j)}\right)^2 + \left( \frac{\gvec{\varphi}_{i,j}-\gvec{\varphi}_{i,j+m_y}}{\vec{u}_{\Theta}(i,j-\frac{m_y}{2})}\right)^2=1,
    \label{eq:upwind_pot}
\end{equation}
where $(m_x,m_y)$ indicates the direction of the upwind discretisation.
This scheme is described in Algorithm~\ref{alg:pot_upwind_scheme}.
\begin{algorithm}
    \caption{Potential computation for one grid cell.}
    \label{alg:pot_upwind_scheme}
    \begin{algorithmic}[1] % The number tells where the line numbering should start
        \Procedure{computeNewPotential}{$i,j,\gvec{\varphi},\vec{u}$}\Comment{Row $i$, column $j$, unit cost field $\vec{u}$ and (intermediate) potential field $\gvec{\varphi}$.}
        \State $\mathrm{hor\_pot},\mathrm{ver\_pot} \gets 0$
        \State $\mathrm{hor\_cost},\mathrm{ver\_cost} \gets \infty$
        \ForAll{$\vec{n} \in \left\{ \vec{n}_x,-\vec{n}_x\right\}$}\Comment{Find horizontal upwind direction}
        \State $\theta = \mathrm{angle}(n)$
        \State nb\_edge $= (i,j)+\vec{n}/2$ \Comment{The unit cost is defined on cell edges}
        \State nb\_pot $=\gvec{\varphi}[(i,j)+\vec{n}]$ \Comment{Potential in neighbour cell}
            \State nb\_cost $=\vec{u}_{\theta}[\mathrm{-nb\_edge}]$\Comment{Cost of moving from $(i,j)$ to $(i,j)+n$}
            \If{$\mathrm{nb\_pot}+\mathrm{nb\_cost} < \mathrm{hor\_pot}$} \Comment{Find the lowest potential}
            \State $\mathrm{hor\_pot} = \mathrm{nb\_pot}$
            \State $\mathrm{hor\_cost} = \mathrm{nb\_cost}$
        \EndIf
        \EndFor
        \ForAll{$\vec{n} \in \left\{ \vec{n}_y,-\vec{n}_y\right\}$}\Comment{Find vertical upwind direction}
        \State \emph{Repeat last loop for $y$-direction}
        \EndFor
        \State $a\gets 1/\mathrm{hor\_cost}^2+1/\mathrm{ver\_cost}^2$ \Comment{Solve equation \eqref{eq:upwind_pot}}
        \State $b\gets -2*(\mathrm{hor\_pot}/\mathrm{hor\_cost}^2+\mathrm{ver\_pot}/\mathrm{ver\_cost}^2)$
        \State $c\gets (\mathrm{hor\_pot}/\mathrm{hor\_cost})^2 +  (\mathrm{ver\_pot}/\mathrm{ver\_cost})^2$ - 1
        \State $r \gets (2*c)/(-b-\sqrt{b^2-4ac})$ \Comment{Remains numerically stable for $a<<c$}
        \State \Return $r$\Comment{Return the largest root}
        \EndProcedure
    \end{algorithmic}
\end{algorithm}
The fast marching algorithm propagates using three data structures: known, unknown and candidate cells. At each step, all cells are assigned to one of these structures.
In the initial stage, the known cells consist of cells in the exit, and cells representing obstacles. All the other cells are stored as unknown. The fast marching algorithm assigns zero potential to the known cells representing the exit.
The standard fast marching method as described in \cite{tsitsiklis95} can be summarised as follows.

Each step, we compute potential values from unknown cells adjacent to known cells, and add these cells to the candidate structure. We find the cell in the candidate structure with the smallest potential and add it to the known cells.
These operations are repeated until the unknown and candidate structures are empty. 
The anisotropic nature of the unit cost field requires a different solution algorithm than a standard fast marching method. The implementation for a $N_x\times N_y$ grid is described in Algorithm~\ref{alg:pot_computation}.
Using the right data structure for the candidate cells (a heap), this algorithm runs in $\bigo{n\log n}$ time, with $n=N_xN_y$ the number of grid cells.
\begin{algorithm}
    \caption{Altered fast marching method.}
    \label{alg:pot_computation}
    \begin{algorithmic}[1] % The number tells where the line numbering should start
        \Procedure{computePotentialField}{$N_x,N_y$}
        \State $\vec{\varphi} = \mathrm{array}(N_x,N_y)$
        \State $D = \mathrm{flags}(N_x,N_y)$\Comment{Flags: KNOWN,UNKNOWN,CANDIDATE}
        \State neighbours = $\left\{ \vec{n_x},\vec{n_y},-\vec{n_x},-\vec{n_y} \right\}$\Comment{Vectors into neigbouring directions}
        \ForAll{$(i,j) \in \left\{ 1,\dots,N_x \right\}\times\left\{ 1,\dots,N_y \right\}$}
            \If{$c_{i,j} \subset G$}
            \State $\gvec{\varphi}_{ij}\gets 0$
            \State $D_{i,j} \gets \mathrm{KNOWN}$
            \ElsIf{$c_{i,j} \subset M_n$}
            \State $\gvec{\varphi}_{ij}\gets \infty$
            \State $D_{i,j} \gets \mathrm{KNOWN}$
            \Else
            \State $\gvec{\varphi}_{ij}\gets \infty$
            \State $D_{i,j} \gets \mathrm{UNKNOWN}$
            \EndIf
        \EndFor
        \State list new\_cells = $\left\{(i,j) | (D_{i,j} = \mathrm{KNOWN}\right\}$
        \State heap candidates = \Call{get\_new\_candidate\_cells}{new\_cells}\Comment{Only (re)compute potential field for cells which need the update}

        \While{$\mathrm{size}(\mathrm{candidates})>0$}
        \ForAll{$(i,j)\in \mathrm{candidates}$}
        \State $\gvec{\varphi}_{i,j} = \min\{\gvec{\varphi}_{ij},\Call{compute\_new\_potential}{i,j,\gvec{\varphi},\vec{u}}$
        \State $D_{i,j} = \mathrm{CANDIDATE}$
        \EndFor
        \State $\mathrm{new\_cell} = \min\left\{ (i,j) | D_{i,j}=\mathrm{CANDIDATE}\right\}$
        \State $D_{i,j} = \mathrm{KNOWN}$
        \State candidates += \Call{get\_new\_candidate\_cells}{new\_cells}
        \EndWhile
        \State \Return $\gvec{\varphi}$
        \EndProcedure

        \Procedure{get\_new\_candidate\_cells}{new\_known\_cells}
        \State list new\_candidate\_cells = \\
        $\left\{ (i,j) + n | (i,j) \in \mathrm{new\_known\_cells}, n \in \mathrm{neighbours},D_{(i,j)+n}=\mathrm{UNKNOWN}\right\}$
        \State \Return new\_candidate\_cells
        \EndProcedure

    \end{algorithmic}
\end{algorithm}

The velocity of a pedestrian at location $\vec{x}$ with direction $\theta$ is equal to
\begin{equation}
    \dot{\vec{x}} = -f(\vec{x},\theta)\frac{\nabla\varphi(\vec{x})}{||\nabla\varphi(\vec{x})||}.
    \label{eq:dyn_velo}
\end{equation}
After obtaining the potential field approximation, the potential field value at the pedestrians location is estimated using bilinear interpolation in the same manner as described in Chapter \ref{sec:micro_macro}.

We equip this model with a mathematically sound discretisation scheme and the ability to model obstacles. While theoretically it is possible to model inhomogeneous domains with a discomfort field, using it to make certain areas of the domain inaccessible does not always yield satisfying results.
\subsection{Approximation of state variables}
In\cite{treuille06}, a density approximation is obtained by computing a density contribution of each pedestrian to his four surrounding grid cells. The density contribution depends on the distance to the grid cell centre.

Pedestrian $a_k$ with position $\vec{x}_{a_k}$ contributes to cell $c_{(i,j)}$ with cell centre $\vec{c}_{(i,j)}$ if 
\begin{equation*}
    \left(  \vec{c}_{(i,j)}  - \vec{x}_{a_k}\right)_x < \Delta x \mbox{ or } \left(  \vec{c}_{(i,j)}  - \vec{x}_{a_k}\right)_y < \Delta y.
\end{equation*}
Figure~\ref{fig:disc_treuille} provides an illustration.

The density contribution for cell $c_{(i,j)}$ becomes
\begin{equation}
    \rho_{a_k}(i,j) = \max \left( \min\left(1-\frac{(\vec{c}_{i,j}  - \vec{x}_{a_k})_x}{\Delta x},1- \frac{(\vec{c}_{i,j}  - \vec{x}_{a_k})_y}{\Delta y}\right)^\lambda,0 \right),
    \label{eq:disc_treuille}
\end{equation}
where $\lambda$ is a decay parameter, coupled to the density threshold $\rho_{\min}$. It is chosen in such a way that $\left(\frac{1}{2}\right)^\lambda < \rho_{\min}$, to impose a maximum density contributions for neighbouring grid cells.

This discretisation has three disadvantages.
\begin{itemize}
    \item The resulting density is grid size dependent. Increasing the grid resolution changes the behaviour of the system inherently.
    \item The density contribution of a pedestrian does not depend on the Euclidean distance to the cell centre, but on an approximation of that distance using the maximum norm.
    \item The density interpolation only satisfies the density threshold requirements and lacks other useful properties like mass conservation.
\end{itemize}

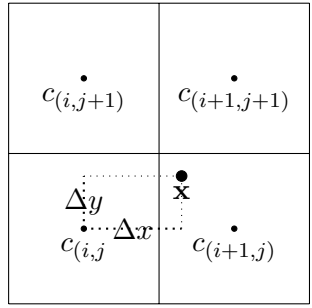
\begin{figure}
    \centering
	\begin{tikzpicture}
        \edef\dx{2.3}
        \edef\dy{1.7}
		\draw[step=2cm,black,thin] (0,0) grid (4,4);
        \filldraw (1,1) circle (1pt) node[anchor=north] {$c_{(i,j}$};
        \filldraw (3,1) circle (1pt) node[anchor=north] {$c_{(i+1,j)}$};
        \filldraw (1,3) circle (1pt) node[anchor=north] {$c_{(i,j+1)}$};
        \filldraw (3,3) circle (1pt) node[anchor=north] {$c_{(i+1,j+1)}$};
        \filldraw (\dx,\dy) circle (2pt) node[anchor=north] {$\vec{x}$};
        \draw[dotted,thick] (1,1) -- node {$\Delta x$}(\dx,1);
        \draw[dotted,thick] (1,1) -- node {$\Delta y$} (1,\dy);
        \draw[dotted] (\dx,1) -- (\dx,\dy);
        \draw[dotted] (1,\dy) -- (\dx,\dy);
	\end{tikzpicture}
    \captionof{figure}[Caption]{Schematic representation of pedestrian density contribution corresponding to \eqref{eq:disc_treuille} in \cite{treuille06}.}
   \label{fig:disc_treuille}
\end{figure}
We use the density interpolation from Chapter \ref{sec:micro_macro} to alleviate these issues.
By adjusting the smoothing length $h$, we satisfy the density requirement while imposing no relation between the grid size and the density contribution radius per pedestrian. For cell size $\Delta x$ we would pick $h$ satisfying
\begin{equation*}
    \psi\left( \frac{1}{2}\sqrt{2}\Delta x,h \right) = \frac{1}{2\pi h^2}e^{-\frac{\Delta x^2}{4h^2}}< \rho_{\min},
\end{equation*}
where $\psi$ is our interpolation kernel.
% \subsection{Enforcing minimal distance}
% Each time step we perform a pairwise collision resolution by enforcing a strict minimal distance between pedestrians. This is done by iterating over all pairs of pedestrians with a distance smaller than have the minimal distance corresponding to the maximum distance and increasing their distance to respect this latter minimal distance. 
% Note that in case of highly dense crowds, this does not ensure all minimal distances are respected. However, since the planner ensures this does not happen, this solution works well.
\subsection{Modelling the obstacles}
Mathematically, it is not so clear how obstacles can be modelled simultaneously for particles and densities. 
We choose an algorithmic approach, which is easy to implement and can be justified for discrete objects.

We model obstacles in the preprocessing step by modifying the domain discretisation. After partitioning the domain into a grid with cell size $\Delta x$ we snap the obstacle boundaries to cell edges, to ensure each cell is either fully accessible or fully covered by an obstacle.
In the fully covered cells, the potential field is set to infinity. As a result, these cells will be ignored in the fast marching algorithm.\\
At runtime, we apply a homogeneous Neumann boundary condition to the potential gradient. The boundary is identified by the transition from finite to infinite potential field values.
Finally, we increase the discomfort field in the cells around the obstacles to account for the obstacle clearance.
\section{Simulation results: Part I}
\label{sec:results}
We test the model discussed in Section~\ref{sec:crowds1} and our implementation by creating two different scenarios and observing pedestrian time spent in the scene, points of congestion and paths taken.
We present two cases. 
\begin{itemize}
    \item Case A: we handle the evacuation of a large crowd in an open space.
    \item Case B: we handle a building evacuation.
\end{itemize}
\subsection{Case A: Evacuation through a narrowing corridor}
The first scenario is depicted in Figure~\ref{fig:narrowd0}. It models an outdoor situation where a large crowd has to be evacuated via one escape route. Three more snapshots of the simulation are provides in Figure~\ref{fig:narrowd1} to Figure~\ref{fig:narrowd3}.
The scenario represents an evacuation of a dense crowd through an increasingly narrow corridor. We pick this case to investigate how the geometry of a scene impacts the density and how this affects the pedestrian interaction.
In addition, this scenario tests the models capabilities for dealing with high densities and large numbers of people.
\subsection{Choice of parameters}
To show our implementation is capable of handling dense and large crowds, we simulate 10000 pedestrians in an area of $320\times320$\meter\squared.
Assume each pedestrian has mass $m=1$.
The obstacles occupy a fraction of 0.18 of the scene, resulting in an initial average density of 0.119\meter\rpsquared.
% The simulation runs at interactive rates: \dt is as large as the real time computation step: 0.1 seconds (if we were to set the time step that small)
The scene is partitioned into 6400 cells of $4\times4$\meter\squared. Each time step ($\Delta t$=0.05\second) the scene is updated.

The maximum speed of the pedestrians follows a normal distribution with mean $1.44\meter\per\second$ and standard deviation 0.15\meter\per\second. These values are representative for measured European pedestrian walking speed according to \cite{daamen07}. 
The exit has a width of 200 \meter. We impose no outflow condition on the exit, given that the pedestrians respect both a maximum speed and a minimum distance (of 0.25\meter) to other pedestrians.

These parameters are chosen to observe critical behaviour in the pedestrian interaction; ensuring the densities do not cause insurmountable bottlenecks, but do cause visible interaction.
\begin{figure}[h]
\centering
\begin{minipage}{.45\textwidth}
	\centering
	\includegraphics[width=0.5\textwidth]{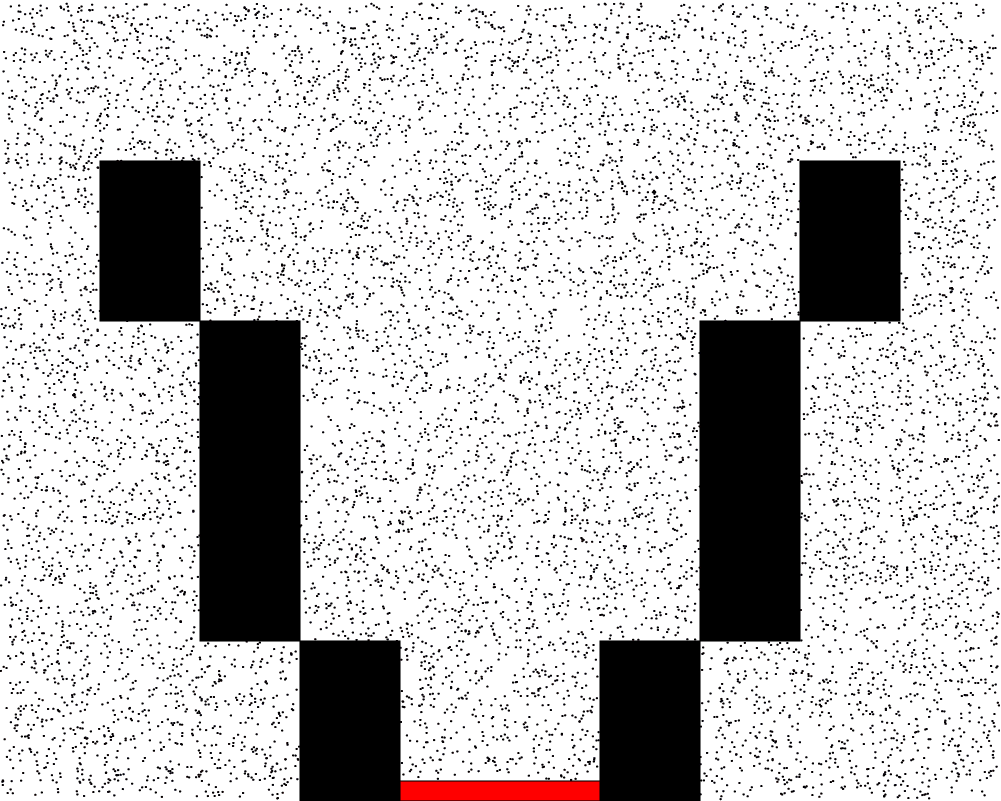}
	\caption{Initial state of Case A. The evacuation exit is at the bottom of the domain.}
	\label{fig:narrowd0}
\end{minipage}%
\hfill
\begin{minipage}{.45\textwidth}
	\centering
	\includegraphics[width=0.5\textwidth]{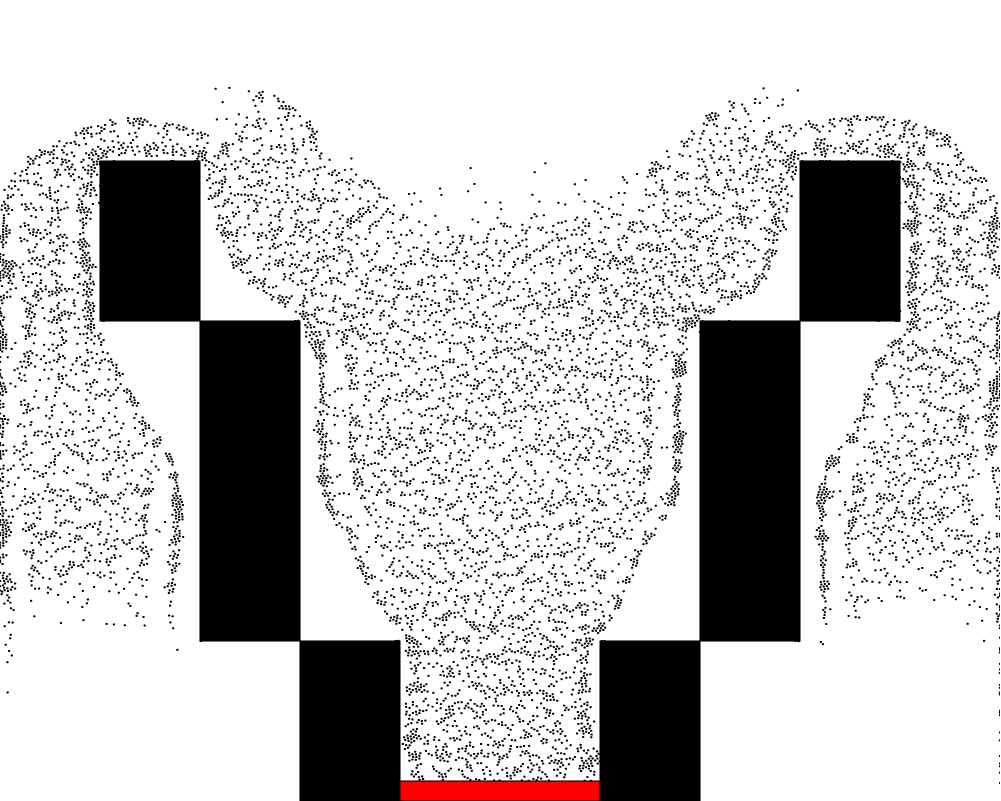}
	\caption{Scenario of Case A after 58 seconds.}
	\label{fig:narrowd1}
\end{minipage}
\begin{minipage}{.45\textwidth}
	\centering
	\includegraphics[width=0.5\textwidth]{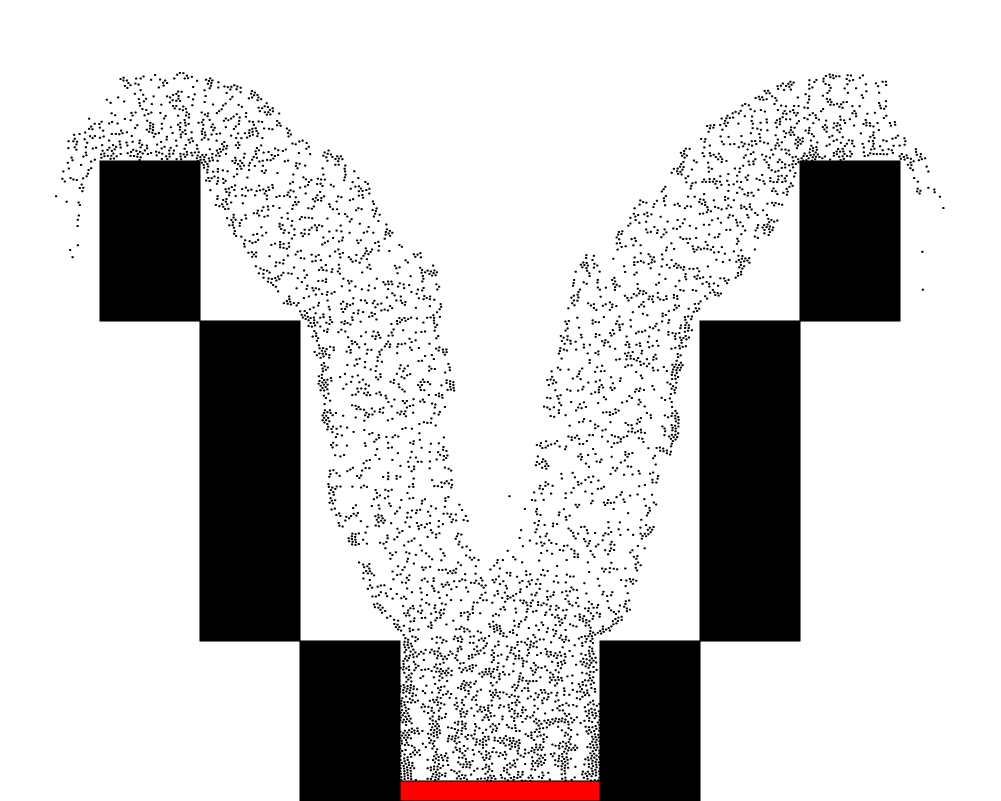}
	\caption{Scenario of Case A after 257 seconds.}
	\label{fig:narrowd2}
\end{minipage}%
\hfill
\begin{minipage}{.45\textwidth}
	\centering
	\includegraphics[width=0.5\textwidth]{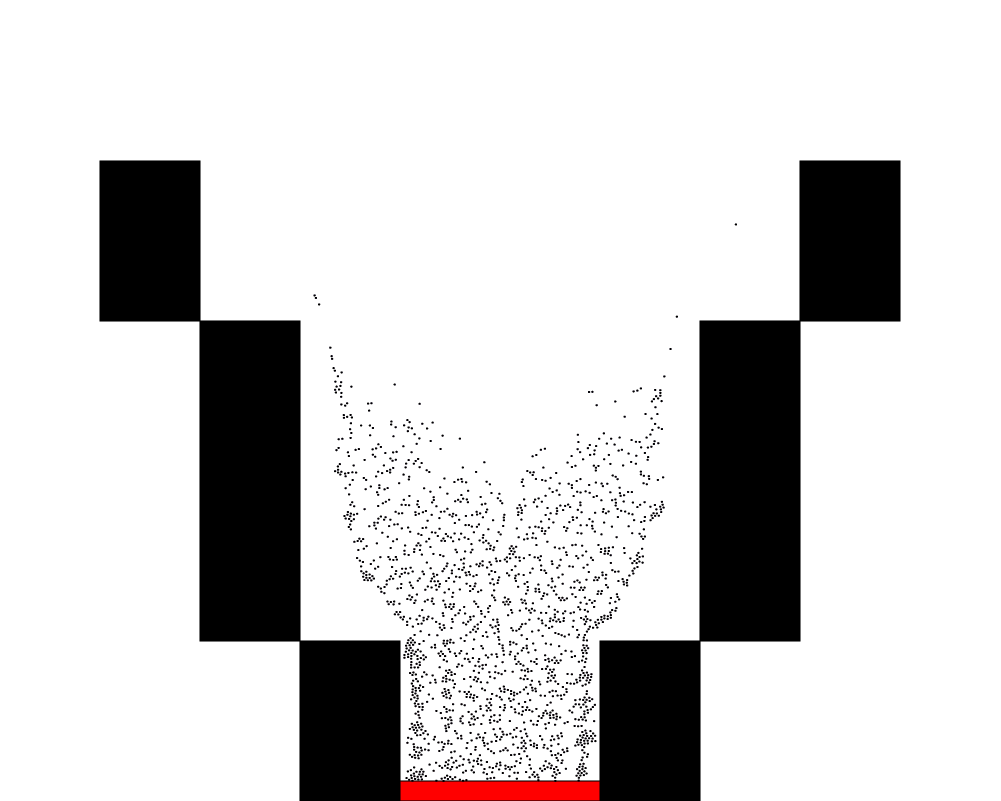}
	\caption{Scenario of Case A after 392 seconds.}
	\label{fig:narrowd3}
\end{minipage}
\end{figure}
\begin{figure}[h]
    \centering
    \includegraphics[width=\textwidth]{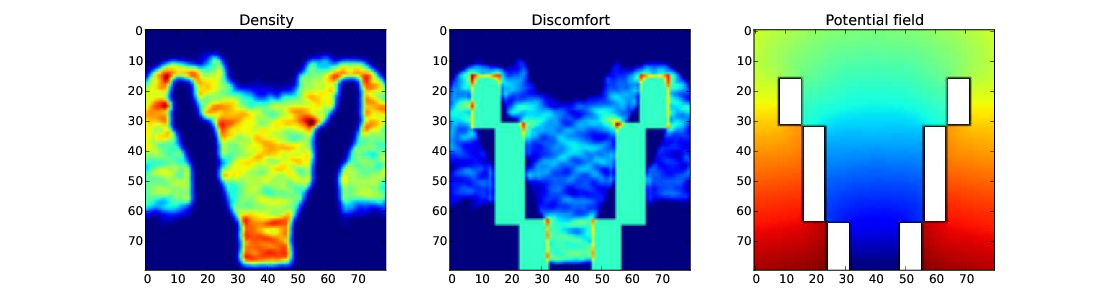}
    \caption{Fields corresponding to Figure~\ref{fig:narrowd1} ($t=58$).}
    \label{fig:narrowdf1}
    \includegraphics[width=\textwidth]{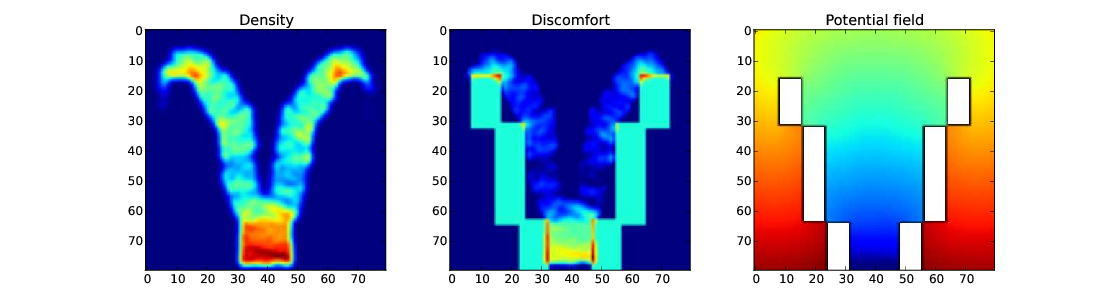}
    \caption{Fields corresponding to Figure~\ref{fig:narrowd2} ($t=257$).}
    \label{fig:narrowdf2}
    \includegraphics[width=\textwidth]{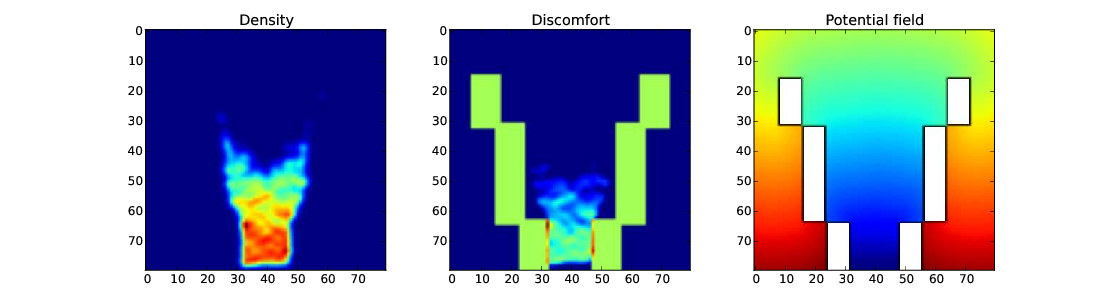}
    \caption{Fields corresponding to Figure~\ref{fig:narrowd3} ($t=392$).}
    \label{fig:narrowdf3}
\end{figure}

\subsection{Quantitative results}
The complete evacuation takes 703 seconds.
Figure~\ref{fig:narrowdf1} to Figure~\ref{fig:narrowdf3} depict the density, potential field and discomfort field at the three simulation snapshots.
This illustrates how the funnel-like obstacles in the scene influence the pedestrian paths.
Figure~\ref{fig:narrow_dyn_time} shows the data from the histogram as a coloured scatter plot of the time to exit. This serves as a measure of (observed) distance to the exit and corresponds with the depicted potential fields.
Figure~\ref{fig:narrow_dyn_delay} shows a scatter plot of the relative delay as a function of initial location.
Relative delay is defined as $1-\frac{\textrm{maximum speed}}{\textrm{average speed}}$ and represents the increase in time a pedestrian remains in the scene due to interaction with other pedestrians.

\begin{figure}[h]
\centering
\begin{minipage}{.45\textwidth}
	\centering
	\includegraphics[width=\textwidth]{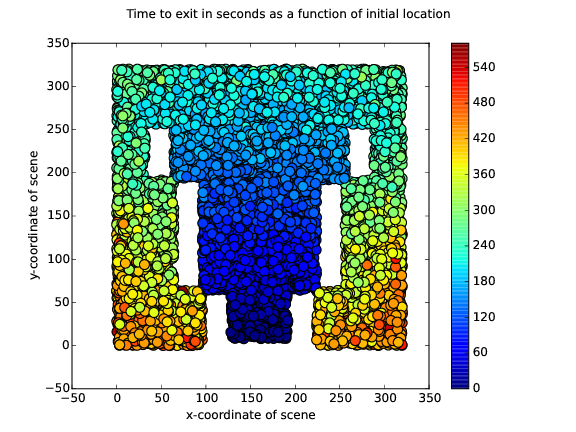}
	\captionof{figure}{Walking time to exit as a function of initial location.}
	\label{fig:narrow_dyn_time}
\end{minipage}%
\hfill
\begin{minipage}{.45\textwidth}
	\centering
	\includegraphics[width=\textwidth]{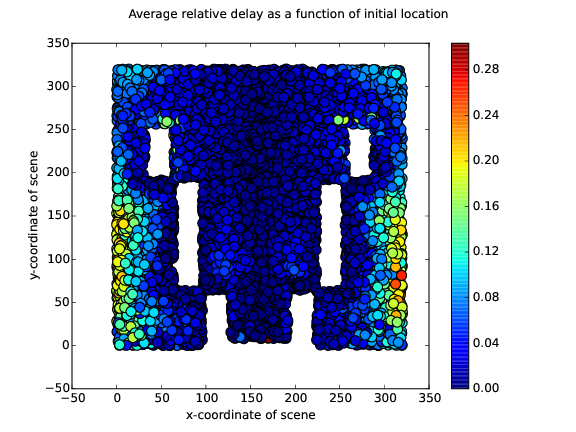}
    \captionof{figure}{Experienced delay as a function of initial location.}
    \label{fig:narrow_dyn_delay}
\end{minipage}
\end{figure}
%\begin{figure}[h]
%    \centering
%    \includegraphics[width=0.5\textwidth]{narrow_dyn_hist.png}
%    \caption{Histogram of leaving times}
%    \label{fig:narrow_dyn_hist}
%\end{figure}
\subsection{Discussion}
We see from the density fields that in spite of the cell-aligned obstacles, the shape of the crowd is smooth on locations with low density. 
As expected, densities increase when the width of the corridor decreases. 
In locations with high density, we see the crowd fills up nearly all the space, and the shape of the crowd adapts itself to the geometry. Yet all pedestrians have smooth and realistic paths, as is visible in Figure~\ref{fig:narrow_dyn_paths}.

\begin{figure}
\centering
\includegraphics[width=0.7\linewidth]{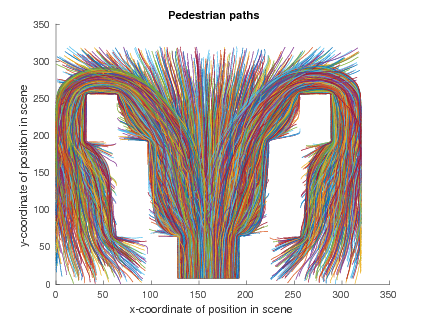}
\caption{Observed paths in the evacuation scene.}
\label{fig:narrow_dyn_paths}
\end{figure}

%From the minimal distance of $d_{min}=0.25$ we obtain a maximum density of 
%\begin{equation}
%    \rho_{\max} =  \frac{2}{ (0.25+2*0.2)^2\sqrt{3}} = 2.73
%\end{equation}
%assuming a pedestrian radius of 0.2 and using \eqref{eq:max_dens_result} from Section~\ref{sec:micro_macro}.
%and comparing that to the width of exit $h=6.4$, we obtain a theoretical maximum throughput of 
%\begin{equation}
%    \frac{\rho_{max}}{h}v = \frac{2.73}{6.4}\cdot1.44=0.61
%    \label{}
%\end{equation}

%Around 4000 time steps, we see a bottleneck in the scenario.

Combining the information from the density plots in Figure~\ref{fig:narrowd3} and the delay plots in Figure~\ref{fig:narrow_dyn_delay}, we see that while congestions occur at bottlenecks (places where the corridor becomes smaller), delays occur at the spaces before that. 
This is crucial information in planning and executing staged evacuations, and determining the location of escape routes.

To investigate the relation between density and mean speed, we rerun the simulation with a different number of pedestrians, while keeping all other parameters intact. 
The mean speed is computed and plotted against the initial density in Figure~\ref{fig:narrow_speed_dens}
While speed declines when density increases, the rate of decrease is quite small when compared to experimental data. In \cite{weidmann93}, measurements are done on the relation between mean speed and density. They observe the mean speed declines linearly with respect to the density, but with a higher rate of change. \\
More experimental data is found in \cite{seyfried10}, but since the experiments done in \cite{weidmann93} correspond to many of the other measurements, we use it as our reference.
\begin{figure}
	\centering
	\includegraphics[width=0.5\linewidth]{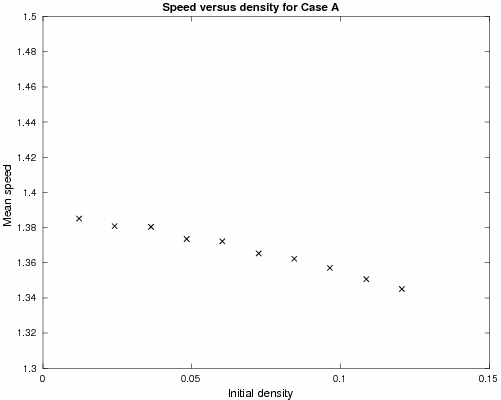}
	\caption{Density and speed measured for 10 simulations, keeping all other parameters equal.}
	\label{fig:narrow_speed_dens}
\end{figure}

Because no potential gradient is defined at obstacle boundaries, in the rare event a pedestrian collides with an obstacle, his motion is arrested.
In this scenario with the given parameters, this happened for $0.7\%$ of the pedestrians. This has been remedied in subsequent runs by finding the stationary pedestrians and move them one step into a random unobstructed direction.
\clearpage
\subsection{Case B: Evacuating a building}
In Case B we investigate the capabilities of this model in indoor environments, dominated by walls and corridors. 
The main focus is observing how the model works in case of more complicated scenarios.
The scene is depicted in Figure~\ref{fig:compd0} and models a (small) crowd leaving a building. We decrease dimensions of the scene, as well as the number of pedestrians. We decrease the minimum distance to $0.15\meter$ to account for indoor walking and observe the effects on the paths of the pedestrians.

\begin{figure}[h]
\centering
\begin{minipage}{.45\textwidth}
	\centering
	\includegraphics[width=0.5\textwidth]{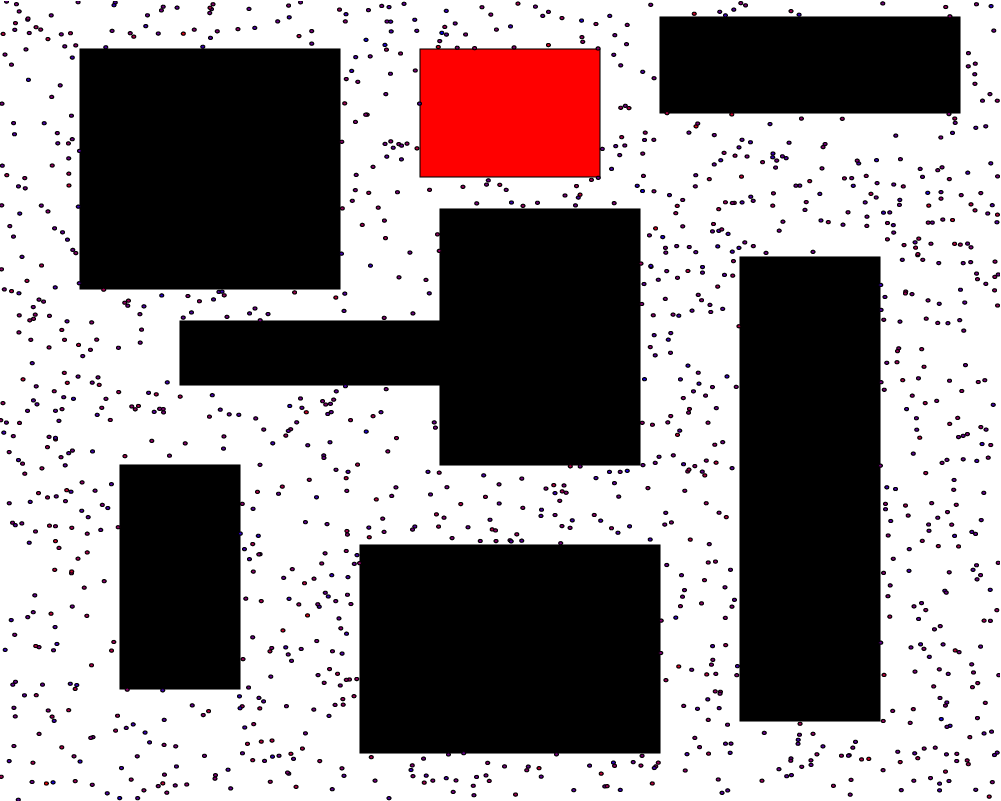}
	\caption{Initial state of the scenario. The red rectangle at the top represents the exit.}
	\label{fig:compd0}
\end{minipage}%
\hfill
\begin{minipage}{.45\textwidth}
	\centering
	\includegraphics[width=0.5\textwidth]{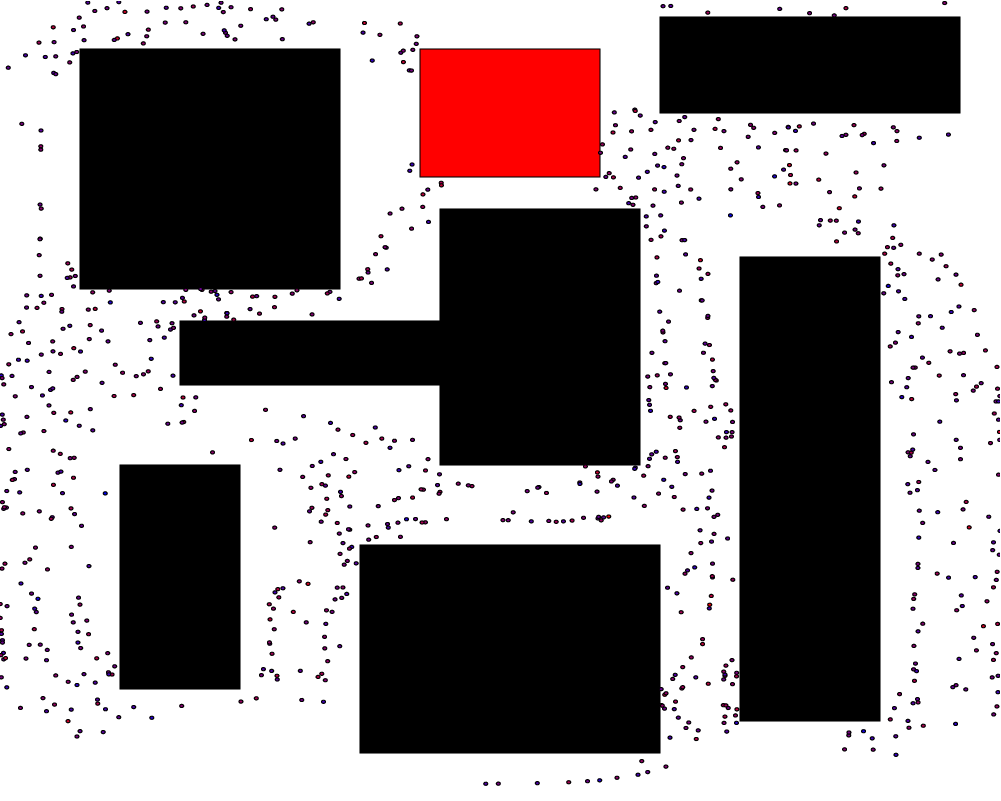}
	\caption{State of the scenario after 13 seconds.}
	\label{fig:compd1}
\end{minipage}
\begin{minipage}{.45\textwidth}
	\centering
	\includegraphics[width=0.5\textwidth]{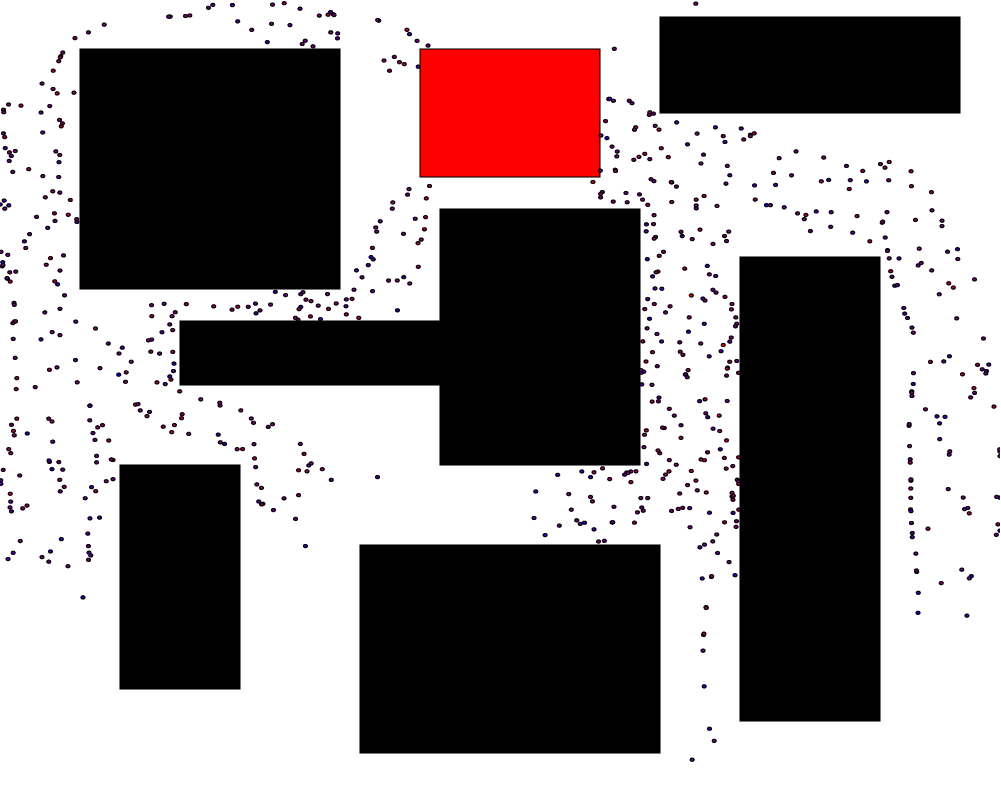}
	\caption{State of the scenario after 38 seconds.}
	\label{fig:compd2}
\end{minipage}%
\hfill
\begin{minipage}{.45\textwidth}
	\centering
	\includegraphics[width=0.5\textwidth]{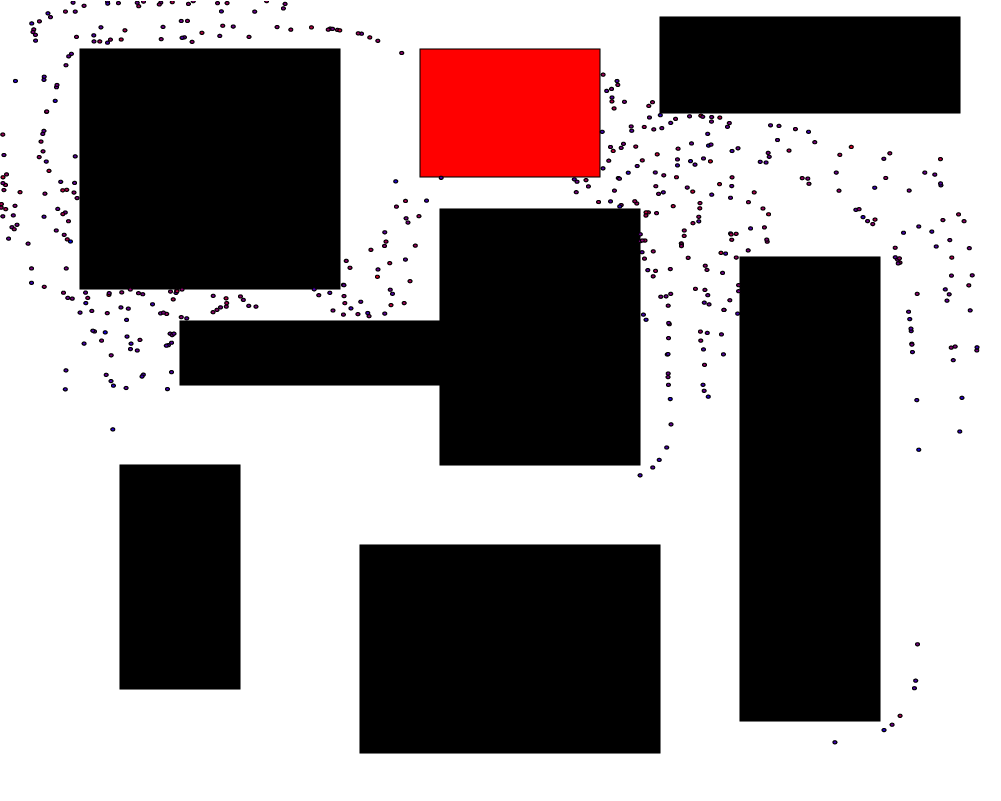}
	\caption{State of the scenario after 71 seconds.}
	\label{fig:compd3}
\end{minipage}
\end{figure}
\begin{figure}[h]
    \centering
    \includegraphics[width=\textwidth]{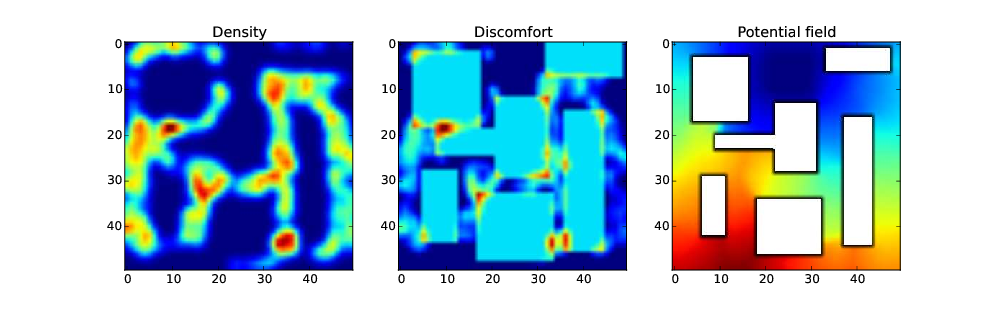}
    \caption{Fields corresponding to Figure \ref{fig:compd1} ($t=13$).}
    \label{fig:compdf1}
    \includegraphics[width=\textwidth]{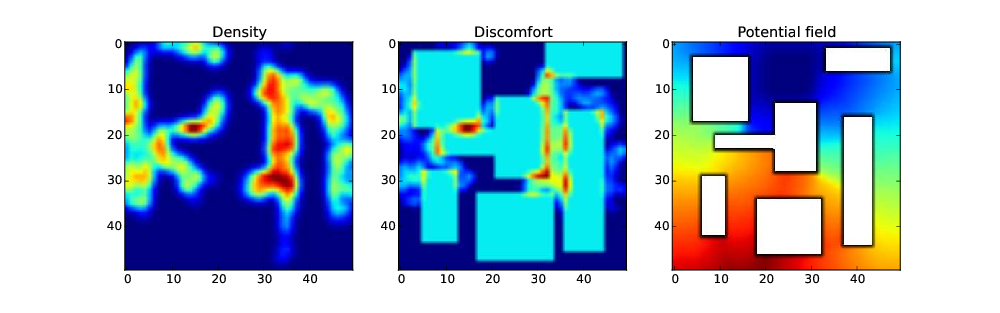}
    \caption{Fields corresponding to Figure \ref{fig:compd2} ($t=38$).}
    \label{fig:compdf2}
    \includegraphics[width=\textwidth]{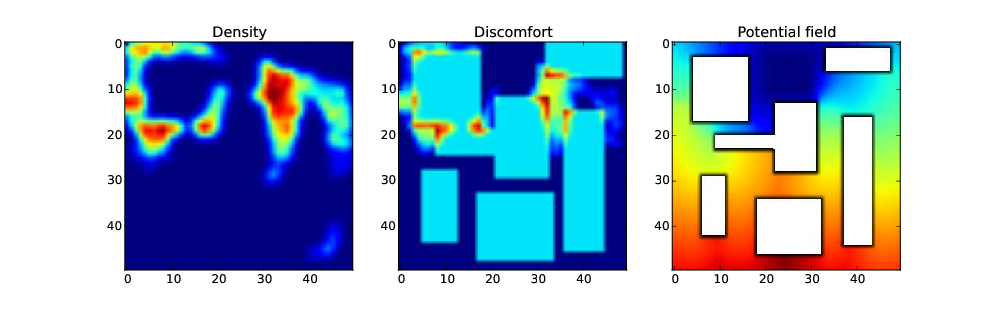}
    \caption{Fields corresponding to Figure \ref{fig:compd3} ($t=71$).}
    \label{fig:compdf3}
\end{figure}

\subsection{Choice of parameters}
We simulate 1200 pedestrians in a scene with size $100\times100\meter\squared$. The obstacles cover a fraction of 0.36 of the space in the scene, so that the average initial density is equal to $0.188\meter\rpsquared$.
We maintain a time step of $\Delta t=0.05\second$ and a grid of $50\times50$. 
\subsection{Quantitative results}
The building clears in 126 seconds. As with the last case, we include the relevant fields at several moments in the simulation, visible in Figure~\ref{fig:compdf1} to Figure~\ref{fig:compdf3}. 
We see that the potential field is highly influenced by the density through the discomfort field. 
We also include the time and delay scatter plots in Figure~\ref{fig:comp_dyn_time} and Figure~\ref{fig:comp_dyn_delay}.
In Figure~\ref{fig:comp_speed_dens} we plotted the mean speed against the initial density for various simulations.

\begin{figure}[h]
	\centering
	\begin{minipage}{.45\textwidth}
		\centering
		\includegraphics[width=\textwidth]{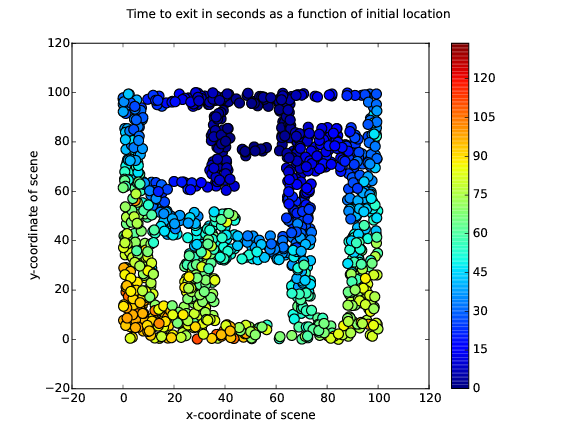}
		\captionof{figure}{Walking time to exit as a function of initial location.}
		\label{fig:comp_dyn_time}
	\end{minipage}%
	\hfill
	\begin{minipage}{.45\textwidth}
		\centering
		\includegraphics[width=\textwidth]{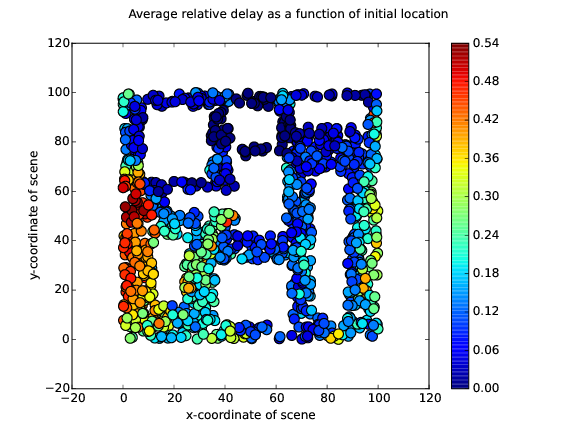}
		\captionof{figure}{Experienced delay as a function of initial location.}
		\label{fig:comp_dyn_delay}
	\end{minipage}
\end{figure}

\begin{figure}
	\centering
	\includegraphics[width=0.7\linewidth]{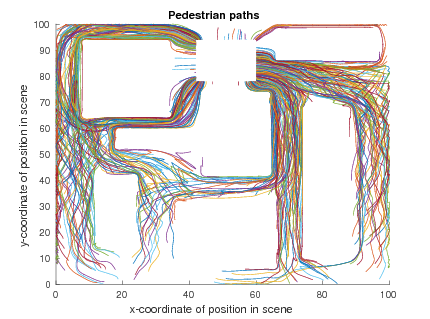}
	\caption{Observed paths from the indoor scene. Upon close inspection, some paths are seen to oscillate.}
	\label{fig:comp_paths}
\end{figure}

\subsection{Discussion}
This scenario shows that the model is capable of handling complex scenarios. If the grid resolution is high enough, then the potential field is able to account for both obstacles and high densities.
However, we see some delay for pedestrians with more than one possibility to reach the exit. 
Since the walking cost evolves with the motion of the pedestrians, the path with the lowest walking cost is also subject to change. This is observed when pedestrians are choosing between two paths with similar costs. When, due to changes in densities, a chosen path becomes more costly than the other one, pedestrians switch directions, as can be observed in Figure~\ref{fig:comp_paths}. This can happen a number of times, depending on the volatility of the crowd. This causes a minor delay for pedestrians at these locations.

This reflects to the knowledge and decision base of the crowd. It implies that a pedestrian exhibits an exhaustive knowledge of the scene and the locations of all other pedestrians in the crowd. 
While the former is to be expected in a familiar environment, the latter is highly unlikely when vision is obstructed by walls and other people.\\
It also implies a lack of anticipation: while some delay might be experienced when choosing the longer of two paths, it will most likely not exceed the delay experienced by continuously switching between the two.

%Upon close inspection of the simulation snapshots, 'waves' of pedestrians are observed in very high density regions. These waves signify a slight instability in solving the Eikonal equation. caused by exceeding the maximum allowed density.

Figure~\ref{fig:comp_speed_dens} shows another effect of lower maximum densities: a decreasing relation between density and average speed, significantly stronger than in Case A. 
We observe the same kind of relation as measured in \cite{weidmann93}, a linear decline in mean speed when density increases. Because of the lower maximum density, the decline in Figure~\ref{fig:narrow_speed_dens} is steeper than in these measurements.

\begin{figure}[!h]
	\centering
	\includegraphics[width=0.4\linewidth]{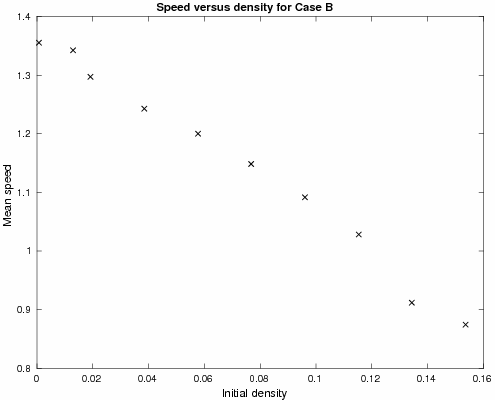}
	\caption{Density and speed measured for 10 simulations, keeping all other parameters equal. Initial density has substantially more influence on mean speed than in Case A.}
	\label{fig:comp_speed_dens}
\end{figure}
\clearpage
\subsection{Application: Evaluating evacuation scenarios at music festival Lowlands}
This implementation was used in \cite{richardson15} for analysing evacuation scenarios on the Dutch music festival Lowlands, commissioned by their event management company LOC7000. 
It provided decent results for simulating large crowds of people between several festival terrains and identified bottlenecks in evacuation routes, congestion points on the festival terrain and estimates for evacuation times.
However, these results cannot be shown due to confidentiality constraints.
\clearpage

% Derivation and results of the interaction potential part
\section{Implementation of a multiscale model: Part II}
\label{sec:crowds2}
While the model in \ref{sec:potential_planner} proves to work quite well in steering large crowds, it suffers from two drawbacks. 
First, it is not robust with respect to high densities; If we want to model large-scale and dense crowds, the simulations become unstable and the particles no longer represent realistic pedestrian behaviour.
Second, the many parameters have no physical interpretation, thus making it difficult to couple to other implementations or real-life crowd data.

\subsection{Interaction potential-based potential}
To retake some control over individual pedestrian motion, we drop the domain potential and create a route planner to explicitly steer the pedestrian towards their goals.
We leave the interaction potential intact, so that we are still able to take advantage of the multiscale modelling structure.

We base the interaction potential function on Darcy's law, mentioned in Section~\ref{sec:macro}. 
We discuss some principles of route planning before proposing our implementation and discussing the pedestrian interaction.

\subsection{Route planning}
\label{sec:path_planning}
Microscopic models need to equip their pedestrians with a path leading to their goal. 
While this might seem obvious, recreating convincing pedestrian routes is far from trivial. 
One has to account for distance to the goal, avoiding of obstacles, avoiding of other (moving) pedestrians, incorporate limited vision angles, all the while maintaining computational efficiency. 
This discipline is called \worddef{path planning} or \worddef{route planning}.
Path planning is a popular topic in robotics and game development. 
We recognise two different ways of planning paths: \worddef{static} and \worddef{dynamic} planning. 

Static planning creates a path from the pedestrians location to his goal, respecting the geometry of the scene and (motionless) obstacles. Other pedestrians and moving obstacles are ignored. Usually, these paths only have to be planned once, for instance at the start of a simulation. 

Dynamic planning is a continuous process. Each time step the direction (or a whole path) of the pedestrian is recomputed. Although this allows for more flexibility, it often is a challenge to incorporate long-range effects and create smooth paths.

Of course, these two methods can be combined to complement each other.
\subsection{Route planning choice}
\label{sec:exp_planner}
A straightforward way to deal with static path planning in domains with obstacles is by computing the \worddef{visibility graph} of a geometry. An elaborate example is described in \cite{kallmann14}.
We base the path planner in our implementation on the visibility graph of the scene, and adapt the resulting paths to become more robust to pedestrian interaction and smooth as to resemble natural paths.
First we elaborate on how to construct and use a visibility graph.
\subsection{Constructing a visibility graph}
Let $M_i \subset \Omega$ for $i=1,\dots,n$ represent the $n$ obstacles present in the domain. Furthermore, to simplify the discussion, assume $\Omega,M_1,M_2,\dots$ to be polygons. Let $G\subset \Omega$ be the pedestrians goal and $x$ be the location of the pedestrian under consideration.

Define $H = (V,A)$ as an undirected weighted graph, the visibility graph. 
Let $V_1$ be a vertex set consisting of vertices of the obstacle polygons $M_1,\dots,M_n$.
Let $V_2$ be a vertex set consisting of pedestrian locations $\vec{x}_{a_1},\dots,\vec{x}_{a_N}$.
We define $V$ as
\begin{equation*}
    V:=V_1 \cup V_2 \cup G,
\end{equation*}
and $A$ as 
\begin{equation*}
    A:=\left\{ (u,v) | u \in V_1 \cup G,v\in V, f(u,v)=0  \right\},
\end{equation*}
where $f(u,v)$ denotes the number of obstacles on the line segment between $u$ and $v$.

This means all obstacle vertices are mutually connected, and all pedestrian positions are connected to the obstacles.
The weight of edge $(u,v)\in A$ is equal to its Euclidean distance $d(u,v)$.\\

Figure~\ref{fig:example_path} and Figure~\ref{fig:example_graph} illustrate a scene and the corresponding visibility graph.
Now, to compute an unobstructed path for a pedestrian at position $x$, use a graph search algorithm (like the $A^*$ search algorithm) to compute the shortest path from $x$ to a vertex from $G$. 
All paths are determined and stored in the preprocessing phase. 
\begin{figure}[h!]
\centering
\begin{minipage}{.45\textwidth}
	\centering
	\includegraphics[width=\textwidth]{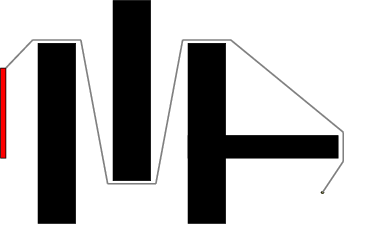}
	\caption{A scene with one pedestrian and his indicative path to the goal.}
	\label{fig:example_path}
\end{minipage}%
\hfill
\begin{minipage}{.45\textwidth}
	\centering
	\includegraphics[width=\textwidth]{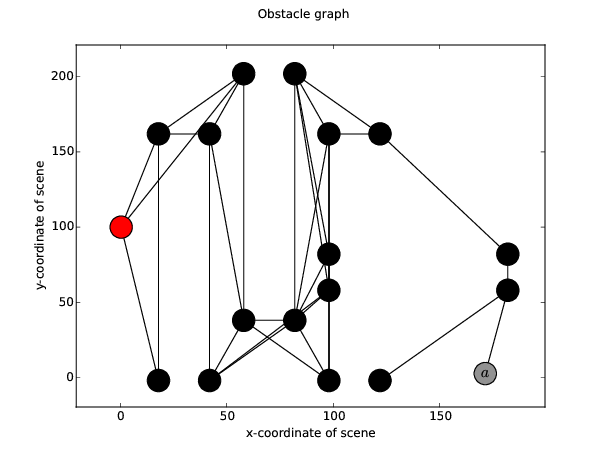}
	\caption{The weighted graph corresponding to the scene in Figure \ref{fig:example_path}. The black nodes correspond to the corners of the obstacles.}
	\label{fig:example_graph}
\end{minipage}
\end{figure}\\
At runtime, each time step the velocity of the next pedestrian is directed towards the next vertex.
\subsection{Intersections with obstacles}
\label{sec:crosses_obstacle}
To construct a visibility graph, we require a way of determining whether two edges in the graph are visible to each other. In other words, we need to check if the line segment between to vertices does not intersect with any obstacles.
This is a common operation in game development, so several method have been developed to solve this problem.
In this section we propose our implementation, derived from a mathematical perspective. 
Using some straightforward linear algebra we discuss theory, implementation and performance.

First, we need to give some definitions:
\begin{newdef}
	Let $\Omega \in \mathbb{R}^2$ be a scene.
	Let $l\subset \Omega$ be a line segment from $\vec{p}$ to $\vec{q}$ with $\vec{p},\vec{q} \in \mathbb{R}^2$. Then the corresponding line $L$ is defined as the set of all $\vec{x}\in\mathbb{R}^2$ satisfying
	\begin{equation}
		\langle\vec{a},\vec{x}\rangle=b,
	\end{equation}
	where $\langle\cdot,\cdot\rangle$ denotes the dot product and $\vec{a}$ and $b$ equal 
	\begin{align}
		\vec{a} &:= \left( -(p_y-q_y),(p_x-q_x)\right),\\
		b   &:=  p_y(p_x - q_x) - p_x(p_y - q_y).
	\end{align}
    \label{def:line}
\end{newdef}
A line separates the set $\Omega$ into two sets $\Omega_{L^+}$ and $\Omega_{L^-}$ such that
\begin{align}
\Omega = \Omega_{L^+} \cup \Omega_{L^-} \cup L,\\
\Omega_{L^+}\cap\Omega_{L^+}=\emptyset.
\end{align}
To determine whether a point $\vec{x}$ belongs to either one of the subsets (or the line itself), the only required operation is computing the dot product with $\vec{a}$.

\begin{newdef}
	Let $\vec{x}$ be a point in $\mathbb{R}^2$. Let $L$ be a hyperplane separating $\Omega$ into subsets $\Omega_{L^+}$ and $\Omega_{L^-}$. These sets are defined as follows:
	\begin{equation}
		\begin{cases}
			\vec{x} \in L&\mbox{ if }(a,x) = b,\\
			\vec{x}  \in \Omega_{L^-} &\mbox{ if }(a,x) < b,\\
			\vec{x}  \in \Omega_{L^+} &\mbox{ if }(a,x) > b.\\
		\end{cases}
		\label{eq:separation}
	\end{equation}
    \label{def:sep}
\end{newdef}
\begin{newdef}
	Let $l$ be a line segment from $\vec{p}$ to $\vec{q}$. The circumscribed rectangle $P_l$ is defined as the smallest rectangle containing $l$:
	\begin{equation}
		P_l := \left\{ w\in\mathbb{R}^2| \min\{p_x,q_x\} \leq w_x \leq \max\{p_x,q_x\}, \min\{p_y,q_y\} \leq w_y \leq \max\{p_y,q_y\}  \right\}.
		\label{eq:bounding_rect}
	\end{equation}
\end{newdef}
An example of a circumscribed rectangle is illustrated in Figure \ref{fig:bounding_rect_ex}.
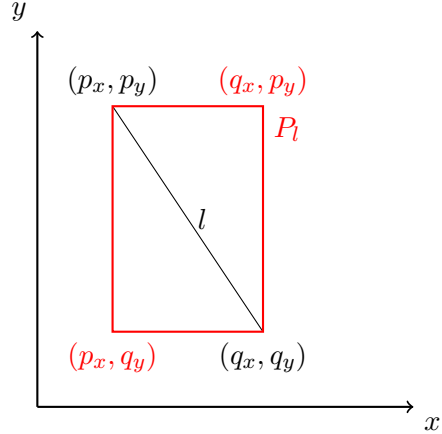
\begin{figure}
	\centering
	\begin{tikzpicture}
		\draw[step=2cm,black,thin] (1,4) node[anchor=south] {$(p_x,p_y)$} -- node[anchor=west]{$l$} (3,1) node[anchor=north] {$(q_x,q_y)$};
		\draw[thick,red] (1,1) rectangle (3,4) node[anchor=north west]{$P_l$};
		\draw[step=2cm,black,thin] (1,1) node[anchor=north,red] {$(p_x,q_y)$};
		\draw[step=2cm,black,thin] (3,4) node[anchor=south,red] {$(q_x,p_y)$};
		\draw[thick,->] (0,0) -- (5,0) node[anchor=north west] {$x$};
		\draw[thick,->] (0,0) -- (0,5) node[anchor=south east] {$y$};
	\end{tikzpicture}
	\captionof{figure}[Caption]{Line segment $l$ and the corresponding circumscribed rectangle $P_l$ with shape $[(p_x,q_y),(q_x,p_y)]$.}
	\label{fig:bounding_rect_ex}
\end{figure}

\begin{newlemma}
	Let $A$ and $B$ be two rectangles with shapes $[\vec{s}_A,\vec{t}_A]$ and $[\vec{s}_B,\vec{t}_B]$. Then the intersection $A\cap B$ is non-empty if and only if
	\begin{equation}
		\vec{t}_{Ai} < \vec{s}_{Bi} \mbox{ or }	\vec{t}_{Bi} < \vec{s}_{Ai}\,, \ i=1,2.
		\label{eq:intersection}
	\end{equation}
	\label{lem:intersect}	
\end{newlemma}
\begin{newthm}
	Let $l\subset \Omega$ be a line segment and $M \subset \Omega$ be an obstacle. Then the following statements are equivalent:
		\begin{enumerate}
			\item $l \cap M = \emptyset$.
			\item $P_l \cap M = \emptyset \mbox{ or } M \subset \Omega_{L^+} \mbox{ or } M \subset \Omega_{L^-}$.
		\end{enumerate}
	\label{thm:sep}
\end{newthm}
\begin{proof}
	(2) $\implies$ (1):\\
	Assume $P_l \cap M = \emptyset$. We know $l \subset P_l$ and therefore $l \cap M = \emptyset$.\\
	Otherwise, without loss of generality, assume $M \subset \Omega_{L^+}$. We know $\Omega_{L^+} \cap l = \emptyset$ so $M \cap l = \emptyset$.

	\ \\
	$\neg (2) \implies \neg (1)$:\\
	Assume $P_l \cap M \neq \emptyset,\, M \not\subset \Omega_{L^+} \mbox{ and } M \not\subset \Omega_{L^-}$.\\
	Let $C := P_l \cap M$. $C$ is the intersection of two rectangular sets and therefore rectangular itself.\\
	Let $C^+:=P_l \cap M \cap \Omega^{L^+}$ and $C^-:= P_l \cap M \cap \Omega^{L^-}$. By construction, we know $C^+$ and $C^-$ are non-empty sets.
	
	\ \\
	We pick $\vec{c}^+ \in C^+$ and $\vec{c}^- \in C^-$ and construct line segment $c$ from $\vec{c}^+$ to $\vec{c}^-$. Since $C^+,C^- \subset C$ and $C$ is a convex set, we know $c \subset C$. But since $\vec{c}^+ \in \Omega_{L^+}$ and $\vec{c}^- \in \Omega_{L^-}$ we know $c$ intersects $l$, so $l \cap M \neq \emptyset$.
\end{proof}
Based on Theorem~\ref{thm:sep}, we obtain a computationally efficient way to find an intersection between some obstacle $M$ and a line segment $l$.

We first compare the circumscribed rectangle of the line segment with the obstacle using Lemma~\ref{lem:intersect}. If these rectangles do not intersect, then we are finished. 
Otherwise, we compute the coordinates of the line using Definition~\ref{def:line}. We determine the location of the obstacle vertices with respect to the line using Definition~\ref{def:sep}.

In the implementation, we restrict ourselves to rectangular obstacles. We allow for aggregation of these rectangles to create more general shapes, as is visible in Figure~\ref{fig:example_path}. 
This is motivated by the fact that many situations can be reproduced by this obstacle modelling technique.
The final part of the discussion will assume (aggregated) rectangular objects.

Figure \ref{fig:crosses_obstacle} depicts an example of an obstacle and a line segment which do not intersect.
% \begin{algorithm}
%     \caption{Find obstacle - line segment intersection}
%     \label{alg:intersection}
%     \begin{algorithmic}[1] % The number tells where the line numbering should start
%     \end{algorithmic} % The number tells where the line numbering should start
%     \Procedure{check\_intersection}{$l,M$}
%     \State \Return true
%     \EndProcedure
%     \Procedure{check\_rectangle\_intersection}{$P_l,M$}
%     \If{$(P_l)_{x1}>M_{x2} \vee M_{x1} > (P_l)_{x2}$}\Comment{No overlap in x direction}
%         \State \Return true
%     \ElsIf{$(P_l)_{y1}>M_{y2} \vee M_{y1} > (P_l)_{y2}$}\Comment{No overlap in y direction}
%         \State \Return true
%     \Else
%         \State \Return false
%     \EndIf
%     \EndProcedure
% 
%     \Procedure{ComputeHyperPlaneValues}{$l,M$}
%     \State $\vec{p} \gets \left( l_{x1},l_{y_1} \right)$
%     \State $\vec{q} \gets \left( l_{x2},l_{y_2} \right)$
%     \EndProcedure
% \end{algorithm}
Noting that any vertex $v$ can be connected to no more than 3 vertices from rectangular obstacle $M_i$ and obstacles can overlap each other, we find an upper bound on the number of edges and vertices in the graph.
With obstacles $M_1,\dots,M_m$ and pedestrians $a_1,\dots,a_n$, graph $(V,A)$ must satisfy 
\begin{align*}
    |V| &\leq n+4m,\\
    |A| &\leq \frac{3m^2}{2} + 3mn,
\end{align*}
which provides us with an upper bound for the number of intersections we have to check.\\
For a $k$-sided polygon, the circumscribed rectangle is computed in $\bigo{k}$ steps. A line segment is a two-sided polygon, and since we model our obstacles as rectangles, no extra computation is needed.
Comparing the boxes takes at most $4$ operations, and identifying the locations of the obstacle vertices with respect to the lines takes another $k$ operations.
These recipes lend themselves for vectorisation and therefore perform well in languages like MATLAB and Python.\\
Using $G$ as a source node, we run Dijkstra's shortest path algorithm once and finish in $\bigo{|A|+|V|\log|V|}$ time. Expressed in $m$ (number of obstacles) and $n$ (number of pedestrians), we obtain
\begin{equation}
    \bigo{\frac{3m^2}{2} + 3mn + (n+4m)\log(n+4m)} = \bigo{m^2+mn+n\log n}.
\end{equation}
We focus on situations where $n\gg m$.

%Say something about the efficiency of this calculation
\begin{figure}[ht]
	\centering
	\begin{tikzpicture}
		\draw[step=2cm,black,thin] (1,4) node[anchor=south] {$(p_x,p_y)$} -- node[anchor=west]{$l$} (3,1) node[anchor=north] {$(q_x,q_y)$};
		\draw[step=2cm,black,thin,dashed] (0,5.5) -- (3.67,0) node[anchor=west]{$L$};
		\draw[thick,red] (1,1) rectangle (3,4) node[anchor=north west]{$P_l$};
		\draw[thick,blue] (0.1,0.5) rectangle (2,2) node[anchor=north west]{$M$};
		\draw[thick,->] (0,0) -- (5,0) node[anchor=north west] {$x$};
		\draw[thick,->] (0,0) -- (0,5) node[anchor=south east] {$y$};
	\end{tikzpicture}
	\captionof{figure}[Caption]{Line segment $l$ and obstacle $M$. There is a non-empty intersection between $P_l$ and $M$, but no intersection between $l$ and $M$ because $M$ lies entirely on one side of the hyperplane $L$.}
	\label{fig:crosses_obstacle}
\end{figure}
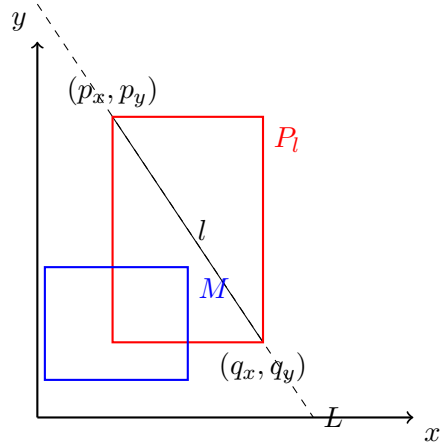

Yet these paths do not have any clearance with respect to obstacles. In addition, if two pedestrians are close, then it is very likely that eventually their paths will be identical.\\
Therefore, we propose an improvement.
\subsubsection*{Exponential waypoints}
When constructing the visibility graph, we increase the size of the obstacles with a fixed margin $m$. This way the paths created have a clearance with respect to the obstacles. In addition, we can enforce pedestrians to avoid extremely narrow paths. \\
At runtime, we do not determine pedestrian directions based on the graph vertices, but by sampling points on the graph edges. The result is visible in Figure~\ref{fig:exp_planner}.
\begin{figure}[h!]
    \centering
    \includegraphics[width=0.6\textwidth]{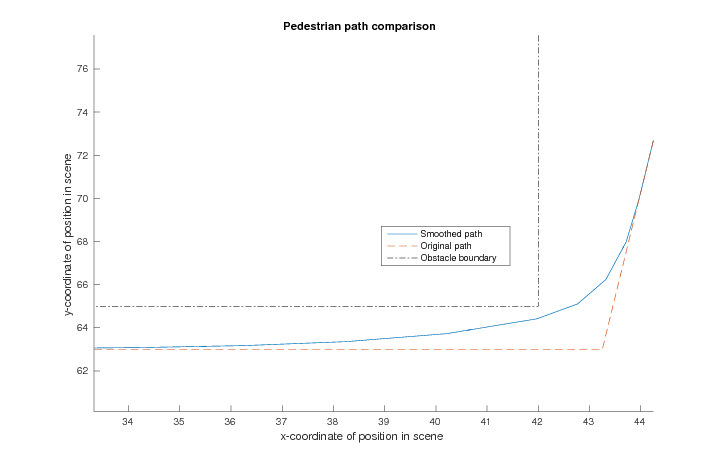}
    \caption{Comparison of original visibility graph paths and the novel paths.}
    \label{fig:exp_planner}
\end{figure}
Instead of choosing the direction of the next vertex, we take a weighted average of a few points ahead on the path.
This produces smooth paths, instead of the piecewise linear paths produced by the visibility graph method alone. 
It also makes the paths more robust to random deviations and other pedestrians: should a segment path be occupied, the pedestrian remains able to move towards the goal.
% style: Do not know how to say this better
\begin{figure}[ht]
    \begin{minipage}{.45\textwidth}
	\centering
	\begin{tikzpicture}
        \draw[<->] (1,0) -- node[anchor=west]{$m$}  (1,1);
        \draw[thin,dashed,->](4,0) -- (0,0);
        \draw[thick,->](5,2) -- (4,0);
        \draw[thin,dashed](6,4) -- (4,0);
        \draw (4,0) circle (2pt);
        \draw[thin,dashed](6,4) -- (4,0);
        \filldraw (5,2) circle (2pt)node[anchor=south east]{$\vec{x}(t)$};
        \draw[thick] (3,4) -- (3,1);
        \draw[thick] (3,1) -- (0,1);
	\end{tikzpicture}
    \captionof{figure}[Caption]{Visibility graph with clearance $m$. Pedestrian at location $\vec{x}(t)$ moves towards obstacle vertex.}
	\label{fig:old_planner}
\end{minipage}%
\hfill
\begin{minipage}{.45\textwidth}
	\centering
	\begin{tikzpicture}
        \def\cx{3.6}
        \def\cy{0.6}
        \draw[thin,dashed,->](4,0) -- (0,0);
        \draw[<->] (1,0) -- node[anchor=west]{$m$}  (1,1);
        \draw[dotted,->] (4.5,1) -- (\cx,\cy);
        \draw[dotted,->] (4,0) -- (\cx,\cy);
        \draw[dotted,->] (3,0) -- (\cx,\cy);
        \draw[dotted,->] (2,0) -- (\cx,\cy);
        \draw (4.5,1) circle (2pt)node[anchor=north west]{$p_{i}$};
        \draw (4,0) circle (2pt)node[anchor=north]{$p_{i+1}$};
        \draw (3,0) circle (2pt)node[anchor=north]{$p_{i+2}$};
        \draw (2,0) circle (2pt)node[anchor=north]{$p_{i+3}$};
        \draw[thick,->](5,2) -- (\cx,\cy);
        \filldraw (\cx,\cy) circle (2pt)node[anchor=south east]{$h_i$};
        \draw[thin,dashed](6,4) -- (4,0);
        \filldraw (5,2) circle (2pt)node[anchor=south east]{$\vec{x}(t)$};
        \draw[thick] (3,4) -- (3,1);
        \draw[thick] (3,1) -- (0,1);
	\end{tikzpicture}
    \captionof{figure}[Caption]{Direction $\vec{h}_i - \vec{x}(t)$ using 4 points. Pedestrian at $\vec{x}(t)$ moves towards a weighted average of positions.}
	\label{fig:new_planner}
\end{minipage}
\end{figure}\\
To determine which points on the path to sample, we use some linear algebra.
Assume path $P$ starts at position $\vec{x}(0)$ and ends in exit $E$. Points on $P$ are sampled with distance $dr$, resulting in $n_P = \ceil{\frac{|P|}{dr}}$ samples. Let $p_i$ denote these sample points.
It holds that
\begin{align}
    \begin{split}
        p_1 &= \vec{x}(0),\\
        p_n &\in E,\\
        d(p_{i}, p_{i-1})&=dr\mbox{ for }i=2,\dots,n_p-1.
    \end{split}
\end{align}
We compute a new direction using the next $n$ points with weights $w=(w_1,\dots,w_n)$. These weights may be chosen in any way as long as $w_k \geq 0$ and $\sum_kw_k=1$.
The direction vector $\vec{h}_i$ is computed by
\begin{equation*}
    \vec{h}_i = \vec{x}(t) - \sum_{k=0}^n w_kb_{i+k+1}.
\end{equation*}
At the start of the simulation, we set $i=1$. We increase $i$ when $p_i$ and $p_{i+n}$ provide conflicting directions; that is when dot product
\begin{equation*}
    \langle \vec{x}(t) - p_i,\vec{x}(t) - p_{i+n}\rangle<0.
\end{equation*}

\subsection{Global interaction}
\label{sec:pressure_interaction}
A hybrid simulation approach similar to the multiscale models we discussed in Section~\ref{sec:macro} is presented in \cite{narain09}. This model only models the interaction between pedestrians, from the hypothesis this is the most expensive and perhaps the most difficult part to simulate.
Again, the crowd is converted to a continuum quantity to interpolate a density and velocity field.
By solving a modified continuity equation, a pressure is introduced in locations where the density exceeds the maximum allowed density. This pressure works on the pedestrians until the density satisfies the allowed density again.

The model uses the continuity equation for mass transport with a unilateral incompressibility constraint based on Darcy's law.
Two aspects of pedestrian interaction are modelled:
\begin{itemize}
\item Swarm behaviour
\item Incompressibility
\end{itemize}

\paragraph{Swarm behaviour}
When a pedestrian observes a higher density in his neighbourhood, his individual velocity converges to the velocity of the crowd.\\
Recall the linear scaling function $L_a^b$ from \eqref{eq:cutoff}.
The actual velocity $\vec{v}_{a_i}$ of pedestrian $a_i$ (located at $(x_{a_i},y_{a_i})$) becomes the average of his desired velocity $\bar{\vec{v}}_{a_i}$ and the crowd's velocity $v(x_{a_i},y_{a_i},\cdot)$, weighted to the local density:

\begin{equation}
    v_{a_i} = v(x_{a_i},y_{a_i},\cdot) + L_0^{\rho_{\max}}\left( \rho(x_{a_i},y_{a_i},\cdot) \right)\left( \bar{v}_{a_i} - v(x_{a_i},y_{a_i},\cdot) \right).
	%v_{a_i} = \left(1-\frac{\rho(x_{a_i},y_{a_i},\cdot)}{\rho_{\max}}\right)\bar{v}_{a_i}+\frac{\rho(x_{a_i},y_{a_i},\cdot)}{\rho_{\max}}v(x_{a_i},y_{a_i},\cdot).
	\label{eq:dens_velo}
\end{equation}

\paragraph{Incompressibility Constraint}
The incompressibility constraint imposes some restrictions on the density and the velocity. We assume a maximum density $\rho_{\max}$. Wherever the density $\rho(x)$ reaches $\rho_{\max}$, the scene is locally saturated, and the pedestrians have to divert from their path. 
This is modelled by a pressure field $p$.

In \cite{narain09} it is stated that using variational calculus the requirement for optimal $v$ respecting $v_{\max}$ arises: 
\begin{equation}
    v = v_{\max}\frac{v-\nabla p}{||v-\nabla p||},
    \label{}
\end{equation}
under the constraints that 
\begin{align}
\label{eq:comp}
    p=0 \implies \rho\leq\rho_{\max},\\
    p>0 \implies \div{v}=0.
\label{eq:dens_compl_pres}
\end{align}
We re-use this idea in our derivation.
Combining this with the original continuity equation in \eqref{eq:cont_equation} we obtain the following system
\begin{equation}\label{eq:pde}
    \frac{\partial \rho}{\partial t} =- \div{\rho(v-\nabla p)}.
\end{equation}
Figure~\ref{fig:example_pressure} shows an example of this system. 
From a density $\rho$ (displayed top left) and initial velocity field $v$ (not shown), accounting for swarm behaviour and the unilateral incompressibility constraint we obtain the pressure $p$ (displayed bottom left), gradient $\nabla p$ (displayed bottom right) and corrected velocity field $v-\nabla p$ (displayed top right).
\begin{figure}[h!]
    \centering
    \includegraphics[width=\textwidth]{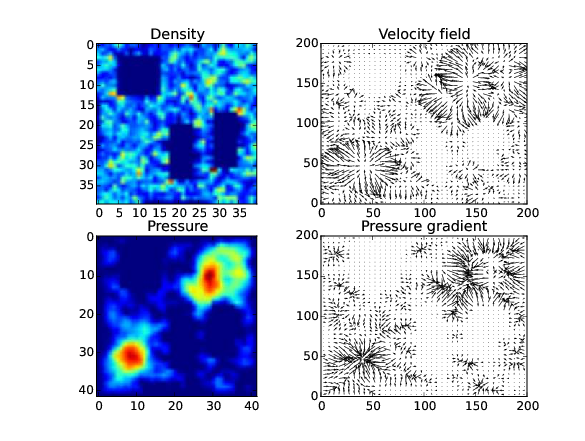}
    \caption{Example of the state variables in the unilateral incompressible quantity.}
    \label{fig:example_pressure}
\end{figure}
\subsection{Boundary conditions}
The boundary conditions in this model should conform with the physical nature of the system. 
Microscopically, we model obstacle boundaries to be impermeable. The motion of particles moving into obstacles is stopped at the boundary.

We also impose boundary conditions on a macroscopic level.
\paragraph{Obstacles}
We prescribe a Dirichlet boundary condition for the pressure. This is motivated by the fact that we want to repel pedestrians from obstacles.

In combination with \eqref{eq:dens_compl_pres} this implies the density at the boundary vanishes. Yet in the actual simulation this manifests itself in the velocity rather than the density.

Let $\partial\Omega_1 \subset \Omega$ denote the boundaries between the scene and the obstacles.

Then the pressure at this boundary is given by
\begin{equation}
    p(\vec{x}) = p_0 \mbox{ for } \vec{x} \in \partial\Omega_1,
    \label{eq:bc_obs}
\end{equation}
with $p_0>0$.

\subsection{Numerical scheme for the continuity equation}
\label{sec:scheme}
To discretise the partial differential equation we use a second order central difference scheme. 
We choose finite difference as the discretisation method since we highly value computational speed over accuracy. Alternatives like finite elements would be under consideration if the computational domain was more complex.
For scalar field $u$ we have the following gradient approximation:
\begin{align*}
\nabla u=\begin{pmatrix}\dfrac{\partial u}{\partial x}\\\dfrac{\partial u}{\partial y}\end{pmatrix}
=\begin{pmatrix}\dfrac{\vec{u}_{(i+1,j)}-\vec{u}_{(i-1,j)}}{2h_x}\\\dfrac{\vec{u}_{(i,j+1)}-\vec{u}_{(i,j-1)}}{2h_y}\end{pmatrix}+\bigo{h^2},
\end{align*}
and for vector field $w$ we define the divergence approximation accordingly:
\begin{equation}
    \div{w}=\dfrac{\partial w_x}{\partial x}+\dfrac{\partial w_y}{\partial y}=\dfrac{\vec{w}_{(i+1,j)}-\vec{w}_{(i-1,j)}}{2h_x}+\dfrac{\vec{w}_{(i,j+1)}-\vec{w}_{(i,j-1)}}{2h_y}+\bigo{h^2}.
    \label{eq:div_approx}
\end{equation}
If we were to discretise $\div{\nabla u}$ using \eqref{eq:div_approx}, then computing the value at cell $(i,j)$ would require cell values from non-adjacent neighbour cells. Therefore, when computing the divergence of the gradient, we use the more compact approximation
\begin{equation*}
    \div{\nabla u} = \dfrac{\vec{u}_{(i+1,j)}-2\vec{u}_{(i,j)}+\vec{u}_{(i-1,j)}}{h_x^2}+\dfrac{\vec{u}_{(i,j+1)}-2\vec{u}_{(i,j)}+\vec{u}_{(i,j-1)}}{h_y^2}+\bigo{h^2}.
\end{equation*}
This way, we have ensured all terms of \eqref{eq:pde} can be computed for every cell not on the boundary.

To solve in time, we use a simple explicit Euler scheme. 
We experimented with a four-stage Runge-Kutta scheme, but the improvement was not noticeable. Because of \eqref{eq:comp}, the system corrects itself in case of densities exceeding the maximum. 
For time integration schemes more elaborate than the Euler scheme, this correction occurs in the first stage, deeming the other stages superfluous.

\eqref{eq:pde} can be rewritten to 
\begin{equation}
    \frac{\partial \rho}{\partial t} = -\div{pv} + \div{\rho \nabla p} = -\div{pv} + \rho \Delta p + \nabla \rho \nabla p
\end{equation}
From that, we obtain numerical scheme:
\begin{equation}
	\begin{split}
		\dfrac{\gvec{\rho}^{n+1}_{(i,j)} - \gvec{\rho}^{n}_{(i,j)}}{\Delta t} =
		% div(\gvec{\rho}*v)\Delta t
		&-\dfrac{\gvec{\rho}^{n}_{(i+1,j)}\vec{v}^{n}_{x(i+1,j)}-\gvec{\rho}^{n}_{(i-1,j)}\vec{v}^{n}_{x(i-1,j)}}{2h_x} \\
		&-\dfrac{\gvec{\rho}^{n}_{(i,j+1)}\vec{v}^{n}_{y(i,j+1)}-\gvec{\rho}^{n}_{(i,j-1)}\vec{v}^{n}_{y(i,j-1)}}{2h_y}\\
		%div(\gvec{\rho}*\grad \vec{p})\Delta t
		&+\frac{1}{h^2_x}\left(\frac{1}{4}(\gvec{\rho}^{n}_{(i+1,j)}-\gvec{\rho}^{n}_{(i-1,j)})(\vec{p}^{n}_{(i+1,j)}-\vec{p}^{n}_{(i-1,j)})+ \gvec{\rho}^{n}_{(i,j)}(\vec{p}^{n}_{(i+1,j)}-2\vec{p}^{n}_{(i,j)}+\vec{p}^{n}_{(i-1,j)})\right)\\
		&+\frac{1}{h^2_y}\left(\frac{1}{4}(\gvec{\rho}^{n}_{(i,j+1)}-\gvec{\rho}^{n}_{(i,j-1)})(\vec{p}^{n}_{(i,j+1)}-\vec{p}^{n}_{(i,j-1)})+ \gvec{\rho}^{n}_{(i,j)}(\vec{p}^{n}_{(i,j+1)}-2\vec{p}^{n}_{(i,j)}+\vec{p}^{n}_{(i,j-1)})\right).
	\end{split}
	\label{eq:scheme}
\end{equation}
It is convenient to express this scheme in terms of matrices and vectors. Not only does this provide us with an efficient way to implement the scheme, it also sets us up for an efficient way of solving the PDE (as explained in Section \ref{sec:lcp}).

Before reformulating the scheme, we introduce \worddef{Kronecker} products and vector gradients.
\subsubsection{Kronecker product}
Let $A \in \mathbb{R}^{m\times n}$ and $B \in \mathbb{R}^{p \times q}$. The Kronecker product $A\otimes B\in \mathbb{R}^{mp \times nq}$ is defined as
\begin{equation*}
	A\otimes B := \begin{pmatrix}
		a_{11}B &\cdots & a_{1n}B\\
		\vdots & \ddots & \vdots\\
		a_{m1}B & \cdots & a_{mn}B
	\end{pmatrix}.
\end{equation*}
The Kronecker product will prove valuable in notation and computation of the discretisation matrices.
% \subsubsection{Hadamard Product}
% Let $A \in \mathbb{R}^{m\times n}$ and $B \in \mathbb{R}^{m \times n}$. The Hadamard product $A\circ B\in \mathbb{R}^{m \times n}$ is the elementwise product of $A$ and $B$:
% \begin{equation*}
% 	A\circ B = \begin{pmatrix}
% 		a_{11}b_{11} &\cdots & a_{1n}b_{1n}\\
% 		\vdots & \ddots & \vdots\\
% 		a_{m1}b_{m1} & \cdots & a_{mn}b_{mn}
% 	\end{pmatrix}.
% \end{equation*}
% Each of the terms \eqref{eq:scheme1}, \eqref{eq:scheme1}, \eqref{eq:scheme3}, and \eqref{eq:scheme4} can be expressed separately by a combination of Kronecker products.
\subsubsection{Vector difference operator}
In Section \ref{sec:scheme} we defined the discretisation of the gradient. We would like to compute a gradient for every cell in the scene, even for the boundary. We introduce a directional \worddef{difference operator} that aids us in computing the gradient approximation. First we surround the scene with a extra layer of cells. These virtual cells are meant to ensure the presence of 8 adjacent cells for all the cells in the scene. We fix the density and velocity in these virtual cells to 0. 
For a discrete field $\vec{u}\in \mathbb{R}^{N_xN_y}$ let the auxiliary extended field be denoted by $\vec{\tilde{u}} \in \mathbb{R}^{(N_x+2)(N_y+2)}$ defined such that:
\begin{equation}
	\vec{\tilde{u}}_{(i,j)} := \begin{cases}
		\vec{u}_{(i,j)}&\mbox{if } i \in \left\{ 1,\dots,N_x \right\},j \in \left\{ 1,\dots,N_y \right\} \\
		0&\mbox{otherwise}
	\end{cases}.
	\label{def:gradient}
\end{equation}
We define difference operators $\D_x, \D_y:\mathbb{R}^{N_xN_y}\to\mathbb{R}^{N_xN_y}$ as follows:
\begin{align*}
	\left( \D_x\vec{u} \right)_{(i,j)} &= \vec{\tilde{u}}_{(i+1,j)}-\vec{\tilde{u}}_{(i-1,j)},\\
	\left( \D_y\vec{u} \right)_{(i,j)}  &= \vec{\tilde{u}}_{(i,j+1)}-\vec{\tilde{u}}_{(i,j-1)}.
	%&\forall i\in \left\{ 1,\dots,N_x \right\},\,\forall j \in \left\{ 1,\dots,N_y \right\}
\end{align*}
\subsection{Matrix composition}
With the definitions from Section~\ref{sec:scheme}, we can succinctly denote the discretisation scheme.
We first define two tridiagonal matrices $P_{m},Q_{m}\in \mathbb{R}^{m\times m}$:
\begin{align*}
	P_{m} &:= \begin{pmatrix}
		0 & 1 &  & &  \\
		-1 & 0 & 1 &   &  \\
		  & -1 & \ddots & \ddots &  \\
		  &  & \ddots & \ddots &1 \\
		 & &  & -1 & 0
	\end{pmatrix},\\
	Q_{m} &:= \begin{pmatrix}
		-2 & 1 &  & &  \\
		1 & -2 & 1 &   &  \\
		  & 1 & \ddots & \ddots &  \\
		  &  & \ddots & \ddots &1 \\
		 & &  & 1 & -2
	\end{pmatrix}.
\end{align*}
Let $I_m$ be the identity matrix of rank $m$. Let $\diag :\mathbb{R}^n\to\mathbb{R}^{n\times n}$ be the operator that converts a vector $\vec{p}$ to a diagonal matrix:
\begin{equation}
	\diag(\vec{p}) = \begin{pmatrix}
		\vec{p}_1\\
		&\vec{p}_2\\
		&&\ddots\\
		&&&\vec{p}_n
	\end{pmatrix}.
	\label{def:diag}
\end{equation}
Finally, we define the scheme with two matrices for each divergence term in \eqref{eq:scheme}.
\begin{align*}
	A_x &:= \frac{1}{4h_x^2}\diag(\D_x\gvec{\rho}^n)(P_{N_x}\otimes I_{N_y}),\\
	A_y &:= \frac{1}{4h_y^2}\diag(\D_y\gvec{\rho}^n)(I_{N_x}\otimes P_{N_y}),\\
	B_x &:= \frac{1}{h_x^2}\diag(\gvec{\rho})(Q_{N_x}\otimes I_{N_y}),\\
	B_y &:= \frac{1}{h_y^2}\diag(\gvec{\rho})(I_{N_x}\otimes Q_{N_y}).
\end{align*}
The final matrix $C$ becomes 
\begin{equation}
	C := A_x + A_y + B_x + B_y.
	\label{eq:total_matrix}
\end{equation}
We define vector $\vec{b}$ to express the flux term:
\begin{equation}
	\vec{b}: = -\left( \frac{1}{2h_x}\D_x(\vec{v}^n_{x(i,j)}\gvec{\rho}^n_{(i,j)}) + \frac{1}{2h_y}\D_y(\vec{v}^n_{y(i,j)}\gvec{\rho}^n_{(i,j)})\right).
	\label{def:vector}
\end{equation}
Combining \eqref{eq:total_matrix} and \eqref{def:vector} we obtain the following matrix-vector system
\begin{equation}
	\gvec{\rho}^{n+1}=\gvec{\rho}^{n}+(C\vec{p}^n + \vec{b})\Delta t.
	\label{eq:matr_vec_scheme}
\end{equation}
\subsection{Reformulation of numerical scheme to linear complementary problem}
\label{sec:lcp}
\eqref{eq:pde} is linear in both $\vec{p}$ and $\gvec{\rho}$. Combining this with the complementarity conditions \eqref{eq:dens_compl_pres} and \eqref{eq:dens_compl_pres} we solve this system by reformulating it to fit a linear complementarity problem (LCP).
\begin{newdef}
For matrix $M \in \mathbb{R}^{n\times n}$ and vector $\vec{q}\in \mathbb{R}^n$, a LCP has the following general form. Find $\vec{z}\in \mathbb{R}^n$ such that
\begin{equation}
    \begin{split}
	\vec{w}&=M\vec{z}+\vec{q},\\
	\vec{w}&\geq\vec{0},\vec{z}\geq\vec{0},\\
        w_iz_i&=0\textrm{ for all }i\in\left\{1,\dots,n  \right\}.
    \end{split}
    \label{eq:def_lcp}
\end{equation}
\end{newdef}
Positive semi-definiteness (PSD) of $M$ is a sufficient condition to solve this problem, regardless of $q$.
We choose the following expressions:
\begin{equation}
\begin{split}
	\vec{w} &= \rho_{\max}-\gvec{\rho}^{n+1},\\
	\vec{q} &= \rho_{\max}-\gvec{\rho}^n+\left(D_x(\gvec{\rho}^n\vec{v}^n)+D_y(\gvec{\rho}^n\vec{v}^n)\right)\Delta t,\\
	\vec{z} &= \vec{p}^n,\\
	M &= -C\Delta t.
	\end{split}
	\label{eq:lcp}
\end{equation}
Using the expressions in \eqref{eq:lcp}, we have reformulated the discretisation of the PDE system to an LCP.
\subsection{Existence of a solution to the LCP}
\label{sec:test_lcp}
It is difficult to formulate the conditions for $\gvec{\rho},\vec{v},\vec{p}$ under which \eqref{eq:def_lcp} has a solution. Since the fact that $M$ is PSD is enough to guarantee a solution, we perform a numerical analysis on $M$ in several distinct cases.
We recall from basic matrix theory that a symmetric matrix is PSD if and only if all of its eigenvalues are non-negative.
As long as the ratio $\frac{\gvec{\rho}}{\rho_{\max}}$ was below 1.4 (and the area of maximum density violation was small), $M$ was close enough to these conditions to find an approximate solution to the LCP.
Our experiments suggest that the existence of a solution depends on this ratio, rather than on the absolute difference between observed density and maximum density.

Below follows a closer examination of the LCP in two cases: one where a scene has a low crowd density, uniformly distributed over the scene ($M_{\textrm{I}}$), and one where the density locally exceeds the maximum density with approximately a factor 1.4 while elsewhere the density is negligible ($M_{\textrm{II}}$). In both cases, the matrices have dimensions $48^2\times48^2$.

Figure~\ref{fig:eig_plot} shows a log-plot of the real part of the spectrum of $M_{\textrm{I}}$ and $M_{\textrm{II}}$ (the set of eigenvalues, denoted by $\sigma(M_{\textrm{I}})$ and $\sigma(M_{\textrm{II}}$). Other details on the LCP system are found in Table~\ref{tab:lcp}.
We observe that aside from small perturbations, the matrices have all non-negative real eigenvalues.
\begin{figure}[h]
    \centering
    \includegraphics[width=0.55\textwidth]{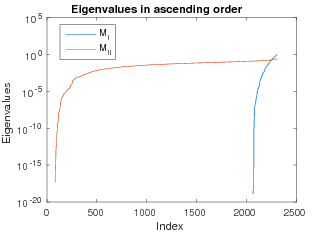}
    \caption{Log-plot of the spectrum of $M_{\textrm{I}}$ and $M_{\textrm{II}}$. The latter only has around 300 positive eigenvalues.}
    \label{fig:eig_plot}
\end{figure}

$M_{\textrm{I}}$ and $M_{\textrm{II}}$ are almost (but not entirely) symmetric. To establish a measure of `symmetricness' of matrix $M$ we compute the antisymmetric part $M_S$ with 
\begin{equation*}
    M_S  = \frac{1}{2}\left( M-M^T \right),
\end{equation*}
and compare the induced two-norm of $M_S$ with $M$.
\begin{table}
    \centering
    \begin{tabular}{|c|c|c|}
        Property &$M_{\textrm{I}}$&$M_{\textrm{II}}$ \\
        $\max\left( \frac{\gvec{\rho}}{\rho_{\max}} \right)$& 0.384 & 1.411\\
        $max(|\operatorname{Im}(\sigma(M_{\textrm{I}})))$& 4.017e-03   &6.138e-03 \\
        min(re(eigenvalues))&                   -2.2108e-08 & 1.94e-17\\
        $\frac{||M_{S}||_2}{||M||_2}$&          0.0791      &0.096 \\
        $||f_{{\textrm{FB}}}(\gvec{\rho},\vec{p})||_2$ &  3.339e-11&1.338  \\ 
    \end{tabular}
    \caption{Additional information on matrices $M_{\textrm{I}}$ and $M_{\textrm{II}}$ and their spectra.}
    \label{tab:lcp}
\end{table}

Finally, we measure the correctness of the found LCP solutions using the Fischer-Burmeister function. This function was proposed in \cite{fischer92} and is a common tool to measure complementarity. It is defined as 
\begin{equation}
    f_{\textrm{FB}}(x,y) = x+y-\sqrt{x^2+y^2}.
    \label{}
\end{equation}

\begin{figure}[!h]
    \centering
    \includegraphics[width=\textwidth]{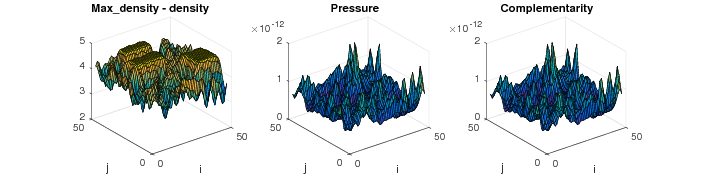}
    \caption{Density gap, pressure, and their complementarity for $M_{\textrm{I}}$. The complementarity is sufficiently satisfied.}
    \label{fig:fb_vals1}
\end{figure}
\begin{figure}[!h]
    \centering
    \includegraphics[width=\textwidth]{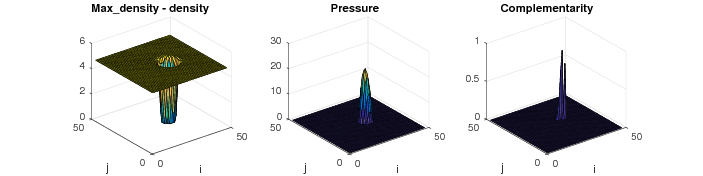}
    \caption{Density gap, pressure, and their complementarity for $M_{\textrm{II}}$. The approximation struggles with satisfying the complementarity where the crowd is densest.}
    \label{fig:fb_vals2}
\end{figure}
Figure~\ref{fig:fb_vals1} and Figure~\ref{fig:fb_vals2} show the density, pressure and Fisher-Burmeister values for $M_{\textrm{I}}$ and $M_{\textrm{II}}$.
We conclude that in the general occurring cases the LCP is solvable. Since the pressure acts as a negative feedback, for simulations where the initial density satisfies the maximum density, the solver should have no problem in finding solutions.
\subsubsection{Quadratic Program Solver}
There are several ways to solve linear complementary problems. If matrix $M$ is positive definite, then LCP's can be solved by a quadratic program (QP) solver, of which there exist many. 

From the assumption that $M$ is positive definite, this LCP can be rewritten to a quadratic program.

Any LCP of the form \eqref{eq:def_lcp} can be rewritten to a QP as follows: Minimise $f(\vec{z})$ where
\begin{equation}
	f(\vec{z}) = \vec{z}^T\left( M\vec{z}+\vec{q} \right),
\end{equation}
with the constraints:
\begin{align}
	M\vec{z} + \vec{q}\geq\vec{0},\\
	\vec{z}\geq\vec{0}.
\end{align}
$\vec{z}$ is a solution to the LCP if and only if $f(\vec{z})=0$. 
These constraints assert $f(\vec{z}) \geq 0$. 
We solve it using the Python library \texttt{cxvopt}.
While this provides accurate results, finding a solution takes a relative long time (approximately 5.2 seconds for $M_{\textrm{I}}$ in the tests of Section~\ref{sec:test_lcp}). This can greatly be improved by using a different solver.
\subsubsection{Projected Gauss-Seidel}
We implement a solver using the projected Gauss-Seidel method to compare to the quadratic program solver. 
We first provide a short summary on the Gauss-Seidel method and afterwards comment on its results in comparison to the QP solver.

The projected Gauss-Seidel (PGS) method is an iterative solution method used in solving LCP's. Each iteration $k$ a new solution approximation $\vec{z}^k$ is computed using the approximations from iteration $k-1$.
We provide a deduction along the lines of \cite{erleben13}.

Observe the system posed in \eqref{eq:def_lcp}. When applying an iterative method, we are looking for $\vec{z}^k$ such that 
\begin{equation}
    \begin{split}
        M\vec{z}^k+\vec{q}&\geq\vec{0},\\
	   \vec{z}^k&\geq\vec{0},\\
       \left(\vec{z}^k\right)^T\left(M\vec{z}^k+\vec{q}\right)&=0.
    \end{split}
    \label{eq:def_itlcp}
\end{equation}
We split matrix $M$ in lower triangular (and diagonal) part $L^*$ and upper triangular part $U$. Furthermore, let $\vec{c}^k:=U\vec{z}^k+\vec{q}$.
Then from \eqref{eq:def_itlcp} we obtain the fixed point formulation
\begin{align}
        L^*\vec{z}^{k+1}+\vec{c}^{k}&\geq\vec{0}\label{eq:pgs_1},\\
        \vec{z}^{k+1}&\geq\vec{0}\label{eq:pgs_2},\\
        \left(\vec{z}^{k+1}\right)^T\left(L^*\vec{z}^{k+1}+\vec{c}^k\right)&=0.\label{eq:def_itpgs}
\end{align}
This reformulation allows us to take advantage of forward substitution in computing the solution approximations.

Assuming both \eqref{eq:pgs_1} and \eqref{eq:pgs_2}, the complementarity condition in \eqref{eq:def_itpgs} is equivalent to
\begin{equation*}
    \min\left(\vec{z}^{k+1},L^*\vec{z}^{k+1}+\vec{c}^k\right)=\vec{0},
\end{equation*}
and by reducing both sides with $\vec{z}^{k+1}$ and multiplying with $-1$ we obtain
\begin{equation*}
    \max\left( \vec{0}, -L^*\vec{z}^{k+1}-\vec{c}^k + \vec{z}^{k+1} \right)=\vec{z}^{k+1},
\end{equation*}
Now looking at the $i$th entry of vector $\vec{z}^{k+1}$ we have two possibilities
\begin{equation*}
    \vec{z}^{k+1}_i=0\textrm{ or }L^*\vec{z}^{k+1}_i=-\vec{c}^k_i.
\end{equation*}
Re-substituting $\vec{c}^k$ and inverting $L^*$ we then obtain $\vec{z}^{k+1}_i$ has to satisfy

\begin{equation*}
    \vec{z}^{k+1} = \max\left( \vec{0},\left( L^{*-1}\left(-U\vec{z}^k-\vec{q} \right) \right) \right).
\end{equation*}

The algorithm is presented below. When applying forward substitution, all the required values from the current and the previous iteration step can be stored in the same vector.
\begin{algorithm}
    \caption{Solving the LCP with projected Gauss-Seidel method.}
    \label{alg:pgs}
    \begin{algorithmic}[1]
        \Procedure{ProjectedGaussSeidel}{$M,\vec{q},\vec{x}_0$}
        \State set stopping parameters $\eps,N_{\max}$
        \State $\vec{x} \gets \vec{x}_0$
        \State $n\gets \mathrm{length}(\vec{q})$
        \State $\vec{w}\gets M\vec{x} + \vec{q}$
        \While{$(\vec{w} < - \eps \boolor |(\vec{w},\vec{x})| > \eps) \booland k < N_{\max}$}\Comment{$\vec{x} \geq \vec{0}$ by construction}
            \State $k=k+1$
            \For{$i=1,\dots,n$}
                \State $r \gets -\vec{q}_i - (M_{i,\cdot},\vec{x})+M_{i,i}z_i$ 
                \State $\vec{x}_i \gets \max\left(0,\frac{r}{M_{i,i}}\right)$
            \EndFor
            \State $\vec{w}\gets M\vec{x} + \vec{q}$
        \EndWhile
        \State \Return $\vec{x}$
        \EndProcedure
    \end{algorithmic}
\end{algorithm}

We implement Algorithm~\ref{alg:pgs} in FORTRAN such that it can be called from the Python framework. With respect to the quadratic program solver, this increases computational speed with 10-14 times. This comparison is made with both solvers employing dense matrix structure. (It must be noted that \texttt{cvxopt} is a third party C implementation with a broader scope, and ours is a tailored FORTRAN implementation).
On top of that, the algorithm is warm-started with pressure values of previous iterations, which are usually close to the solution of the current iteration. This benefits the computational speed even more.

The cost is paid with accuracy and a smaller range of solvable LCP problems. Numerical experiments we executed tell us the LCP solver requires all positive eigenvalues in matrix $M$ to solve the system.
We can enforce this by increasing the density with a small value (we used 0.01). This ensures $M$ only has strictly positive eigenvalues.

After implementing the algorithms with sparse matrices we are able to approximate the solution to the LCP for $\vec{x},\vec{w}$ having $10000$ entries in 0.012 seconds.\\
We use a Compressed Sparse Row format to represent our sparse matrices. Of the various available sparse matrix structures, this one performs very well for our implementation, given that the PGS algorithm frequently extracts matrix rows for dot product calculations.
\subsection{Pressure impact on velocity}
After solving the LCP posed in Section~\ref{sec:lcp} for $\vec{z}$ we obtain pressure $\vec{p}^n = \vec{z}$.
We compute the gradient approximation $\nabla\vec{p}^n$ and subtract it from the velocity to enforce Darcy's law. After that, we normalise the velocity to $v_{\max}$.
The result is a final grid velocity $\tilde{v}^n$ satisfying
\begin{equation}
    \tilde{v}^n = v_{\max} \frac{\vec{v}^n-\nabla\vec{p}^n}{||\vec{v}^n-\nabla\vec{p}^n||}.
	\label{eq:finvelocity}
\end{equation}
We use $\tilde{v}^n$ to steer the velocity of the pedestrians. First we interpolate the final crowd velocity $\tilde{v}^n_{a_i}$ at the pedestrians location $(x_{a_i},y_{a_i})$ by applying bilinear interpolation from the 4 surrounding cell values. 
Recall the elaboration provided in Section~\ref{sec:micro_macro}.
We determine the individual velocity $\vec{v}^{n+1}_{a_i}$ by computing the weighted average as described in \eqref{eq:dens_velo}.

\section{Simulation results: Part II}
\label{sec:results2}
This section discusses the results obtained with the model discussed in this chapter: we use the planner from Section~\ref{sec:exp_planner} augmented with the interaction potential discussed in Section~\ref{sec:pressure_interaction}.
As with the previous model, we create two test scenarios to evaluate the implementation. This time, we use the same domain for both scenarios, but vary the initial and boundary conditions.
We present two cases:
\begin{itemize}
    \item Case C: we simulate an evacuation of an outside scenario with several obstacles and exits.
    \item Case D: we simulate a traffic scenario in which pedestrians traverse through the scene.
\end{itemize}
\subsection{Case C: Evacuation of a plaza}
In Case C we investigate the capacities of the planner and the interaction potential to deal with scenes with multiple exits and obstacles.
The scenario is displayed in Figure~\ref{fig:case_c:0} to Figure~\ref{fig:case_c:3} for various moments in time.
\subsection{Choice of parameters}
The scenario has an area of $200\times200\meter\squared$ and is initiated with $5000$ pedestrians.
We assume a minimum distance of 0.7\meter.
The LCP is solved on a grid with $100\times 100$ cells and the time step is fixed to 0.05\second.
\subsection{Quantitative results}
Although 95\% of the pedestrians reached the exit after 182 seconds, it takes 1556 seconds to clear the scene.
As with the previous simulation results, we show plots of the state variables at time of the snapshots in Figure~\ref{fig:case_c_field:1} to Figure~\ref{fig:case_c_field:3}.
The measured delay is plotted in Figure~\ref{fig:case_c:delay} and the time to exit is plotted in Figure~\ref{fig:case_c:time}.
We plot the paths in Figure~\ref{fig:case_c:paths}.

To obtain a measure of which areas pedestrians prefer, we plot a density heatmap $R:\Omega\to \mathbb{R}$ in Figure~\ref{fig:case_c:logdens}, defined as
\begin{equation}
    R(\vec{x}) = \int_0^T\log(1+\rho(\vec{x},t))dt.
    \label{eq:density_heatmap}
\end{equation}
This quantity can be interpreted as a macroscopic representation of the paths in Figure~\ref{fig:case_c:paths}.

\begin{figure}[h]
\centering
\begin{minipage}{.45\textwidth}
	\centering
	\includegraphics[width=0.7\textwidth]{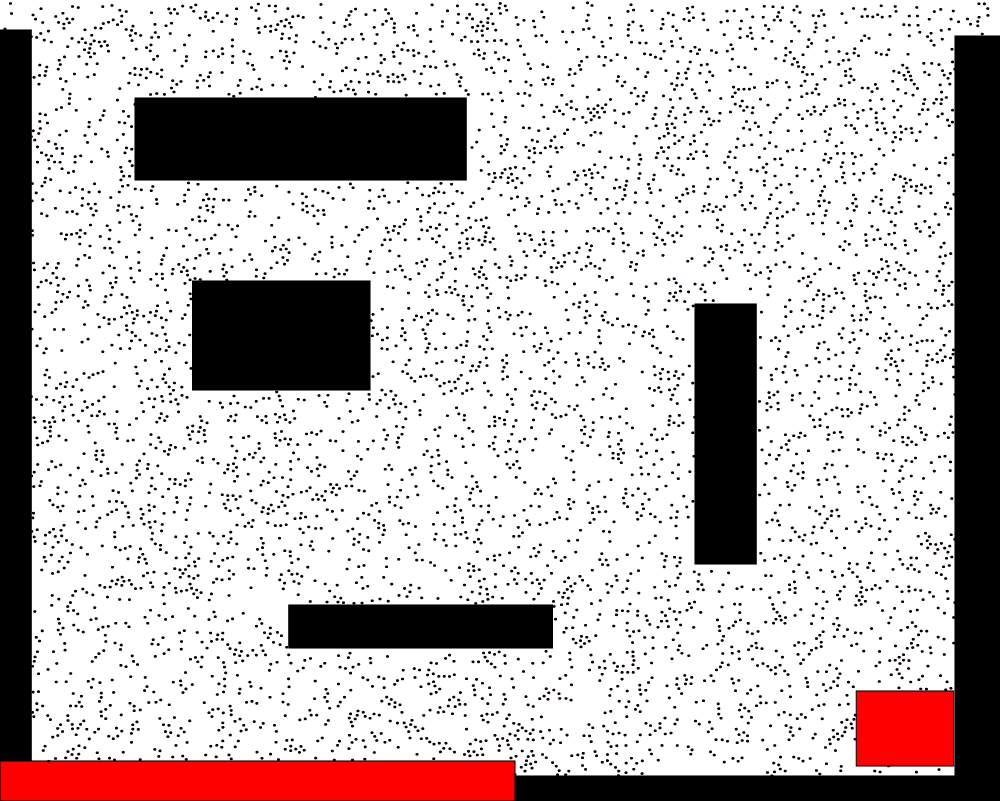}
    \caption{Initial state of the scenario of Case C.}
	\label{fig:case_c:0}
\end{minipage}%
\hfill
\begin{minipage}{.45\textwidth}
	\centering
	\includegraphics[width=0.7\textwidth]{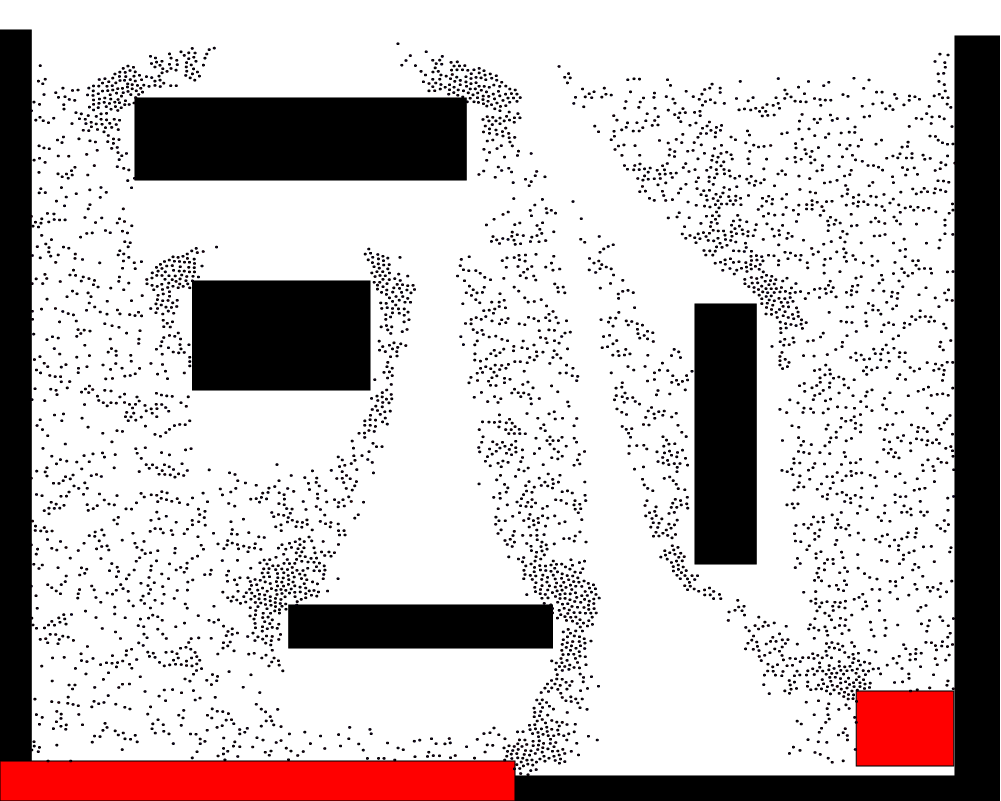}
	\caption{Scenario of Case C after 14 seconds.}
	\label{fig:case_c:1}
\end{minipage}
\begin{minipage}{.45\textwidth}
	\centering
	\includegraphics[width=0.7\textwidth]{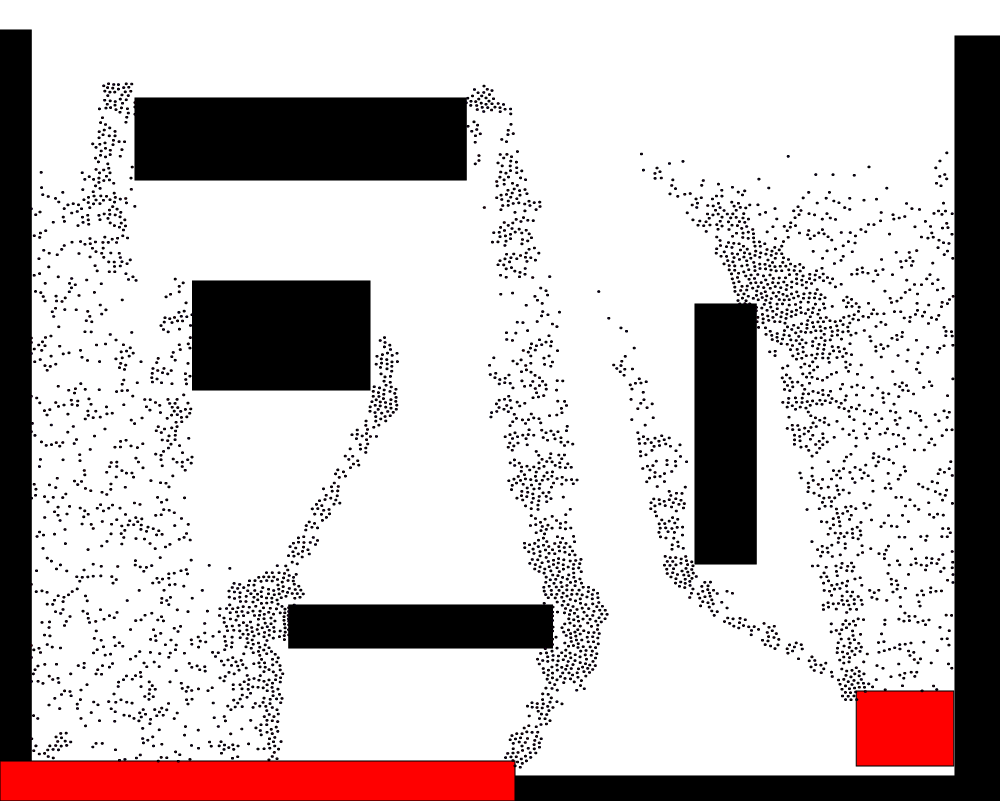}
	\caption{Scenario of Case C after 36 seconds.}
	\label{fig:case_c:2}
\end{minipage}%
\hfill
\begin{minipage}{.45\textwidth}
	\centering
	\includegraphics[width=0.7\textwidth]{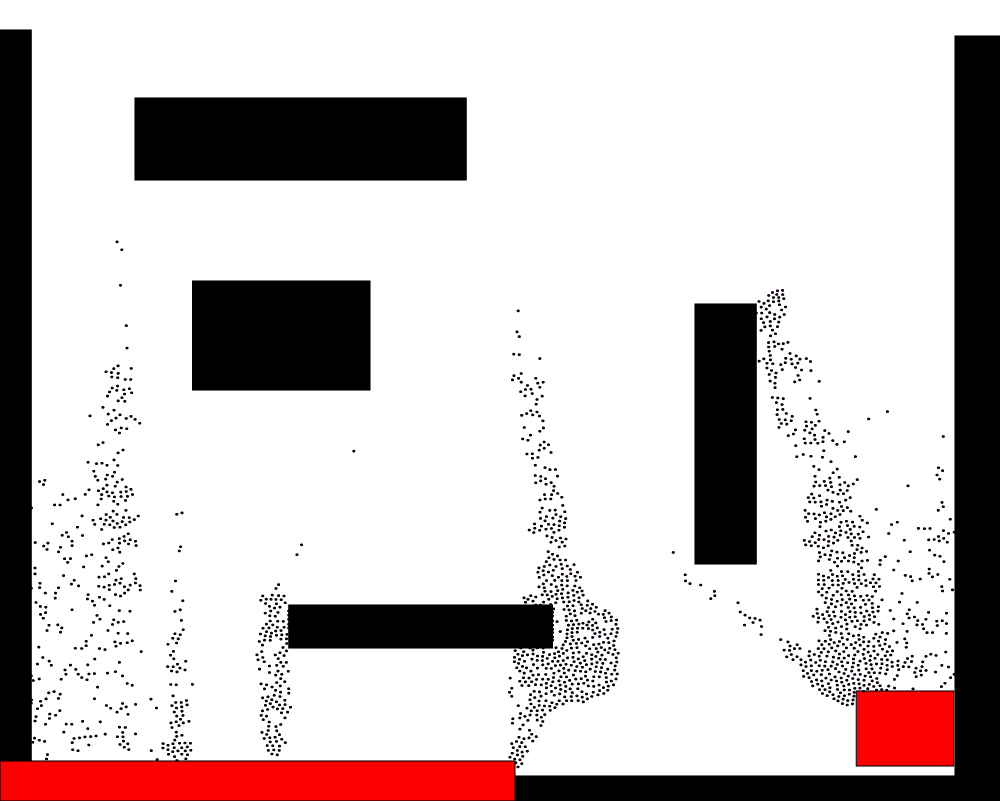}
	\caption{Scenario of Case C after 90 seconds.}
	\label{fig:case_c:3}
\end{minipage}
\end{figure}

\begin{figure}[h]
\centering
\begin{minipage}{.45\textwidth}
	\centering
	\includegraphics[width=\textwidth]{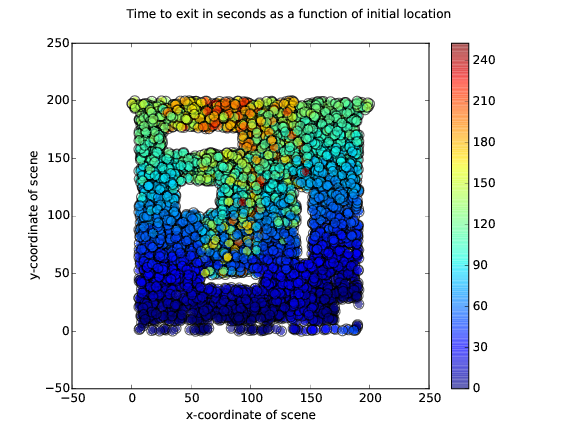}
	\captionof{figure}{Walking time to exit for Case C. Walking times are cut off at 240 seconds, to prevent the outliers from dominating the plot.}
	\label{fig:case_c:time}
\end{minipage}%
\hfill
\begin{minipage}{.45\textwidth}
	\centering
	\includegraphics[width=\textwidth]{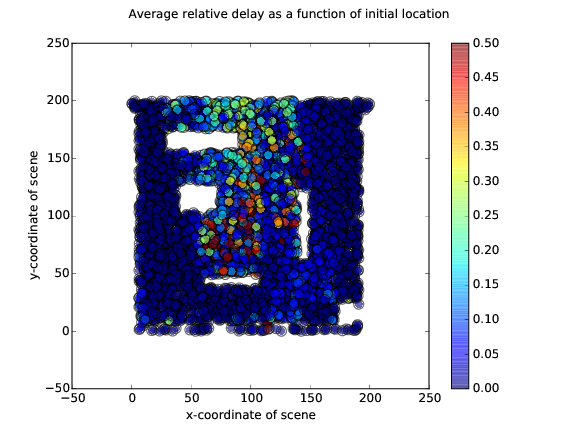}
    \captionof{figure}{Experienced delay for Case C. Delay is cut of at 50\%.}
    \label{fig:case_c:delay}
\end{minipage}
\end{figure}

\begin{figure}[!h]
\centering
\begin{minipage}{.45\textwidth}
	\centering
    \includegraphics[width=\linewidth]{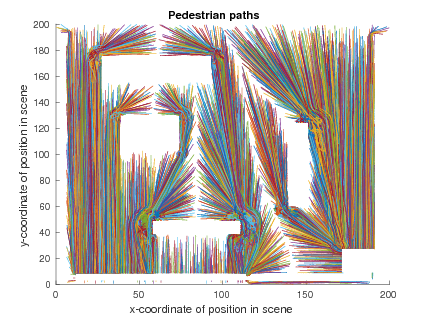}
    \caption{Observed paths for Case C.}
    \label{fig:case_c:paths}
\end{minipage}%
\hfill
\begin{minipage}{0.45\textwidth}
    \centering
    \includegraphics[width=\textwidth]{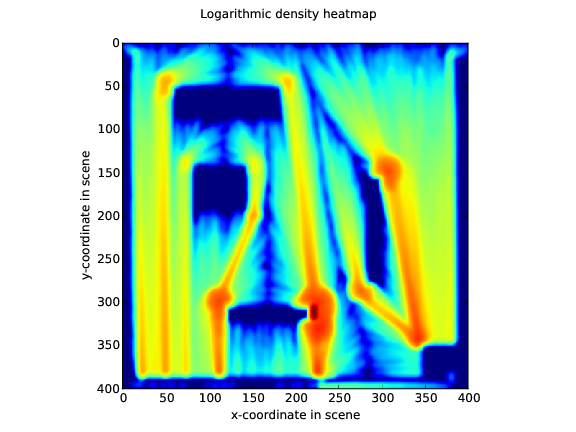}
    \caption{Density heatmap for case C.}
    \label{fig:case_c:logdens}
\end{minipage}
\end{figure}

\begin{figure}[h]
    \centering
    \includegraphics[width=0.7\textwidth]{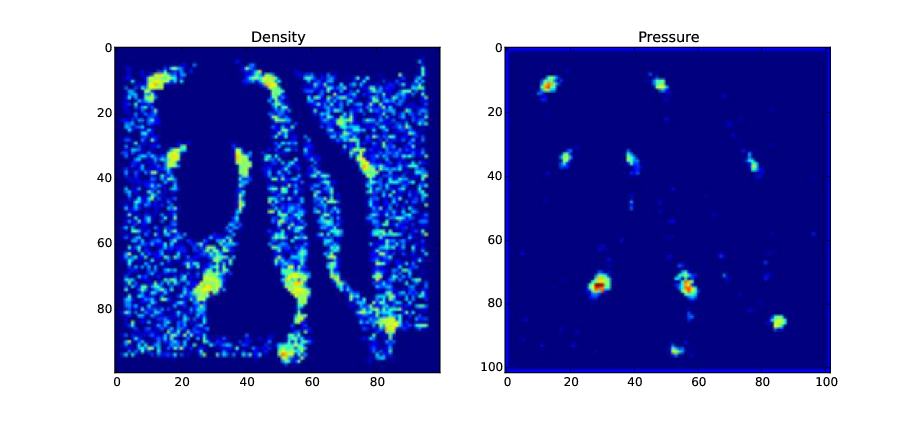}
    \caption{Density (left) and pressure (right) corresponding to Figure~\ref{fig:case_c:1} ($t=14$).}
    \label{fig:case_c_field:1}
    \includegraphics[width=0.7\textwidth]{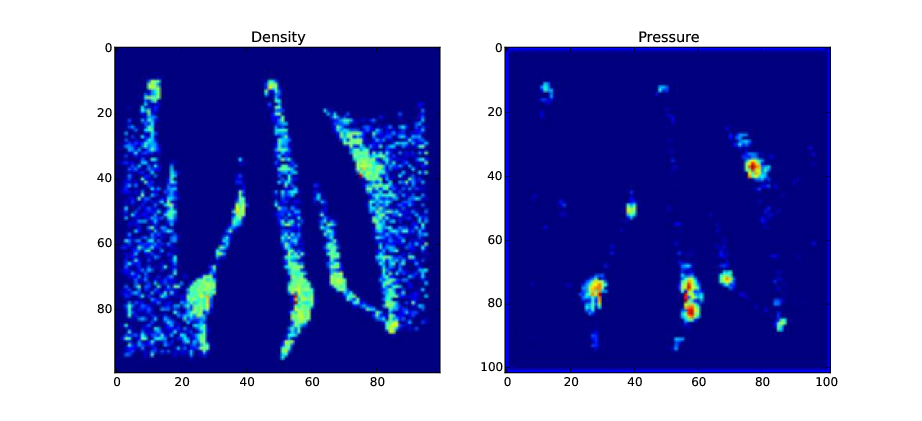}
    \caption{Density (left) and pressure (right) corresponding to Figure~\ref{fig:case_c:2} ($t=36$).}
    \label{fig:case_c_field:2}
    \includegraphics[width=0.7\textwidth]{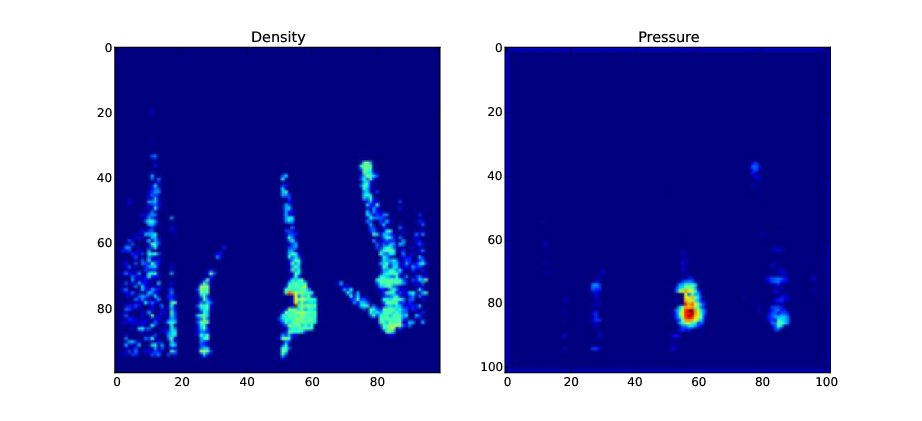}
    \caption{Density (left) and pressure (right) corresponding to Figure~\ref{fig:case_c:3} ($t=90$).}
    \label{fig:case_c_field:3}
\end{figure}
\subsection{Discussion}
It takes quite a long time to completely clear all pedestrians from the scene. 
This is caused by conflicts between the path planner and the interaction potential. 
Figure~\ref{fig:case_c:3} shows an example of such a conflict. 
Since the planner works with waypoints, many pedestrians try to reach the vicinity of the same locations in the scene.
But when the density in certain locations exceeds its maximum, the pressure exerts a repulsive effect on those pedestrians.
While this validates that the interaction potential works well, it causes bottlenecks around waypoints if the number of passing pedestrians is high.
% Shows some overcompensation. There the paths were too flexible, here they are too rigid.
In other locations, the pressure performs well in maintaining acceptable densities. The path planner ensures that the oscillating paths seen in Case B are no longer present; as is seen in Figure~\ref{fig:case_c:paths}, paths are mostly straight, except in region of high densities. 
Another advantage with respect to the simulation of Case A and B is that this simulation remains stable for high densities.

This is paid with a cost: in this implementation, the bottleneck is the computation of the individual pedestrian paths that enter the scene.
This was not necessary when using the domain planner, causing this simulation to be slower.
Simulation of one time step takes between 0.06 and 0.45 seconds (depending on how often the maximum density is exceeded).
Using the engine of the previous simulation, one time step takes a steady 0.2 seconds.

\subsection{Case D: One-way traffic simulation}
In Case D, we investigate the same domain as in Case C, but this time no people are present at the start of the scenario. Instead they enter the scene through one of the entrances. The exits remain in place. The scenario is visible in Figure~\ref{fig:case_d:0}.
We are interested in the effect of inflow conditions on the state and results of the simulation. We also want to check how the minimum distance is respected.

This case is essentially different than its predecessors. The inflow conditions ensure the simulation does not finish; instead, it reaches an equilibrium in which the inflow matches the outflow and pedestrians follow roughly the same paths.

\subsection{Choice of parameters}
The dimensions of the scenario remain unchanged. Two of the obstacles in Figure~\ref{fig:case_c:0} are now entrances.
Both entrances have an inflow according to \eqref{eq:inflow_prob} of rate $\lambda=10$ pedestrians per second.
In order for the pressure to kick in, the maximum density has to be chosen much lower than the density resulting from the minimum distance.
The minimum distance in this simulation is set to $2.0\meter$ which result in a minimum distance of $1.4\meter$ enforced by the interaction potential.

\subsection{Quantitative results}
The simulation ran for 366 seconds. Over the course of the simulation, 2965 pedestrians were simulated, of which 514 were present when the simulation ended.
Figure~\ref{fig:case_d:0} shows the equilibrium plot of these simulations.
Figure~\ref{fig:case_d_fields} shows the density and pressure of this equilibrium.

\begin{figure}[h]
\centering
\begin{minipage}{.45\textwidth}
	\centering
	\includegraphics[width=\textwidth]{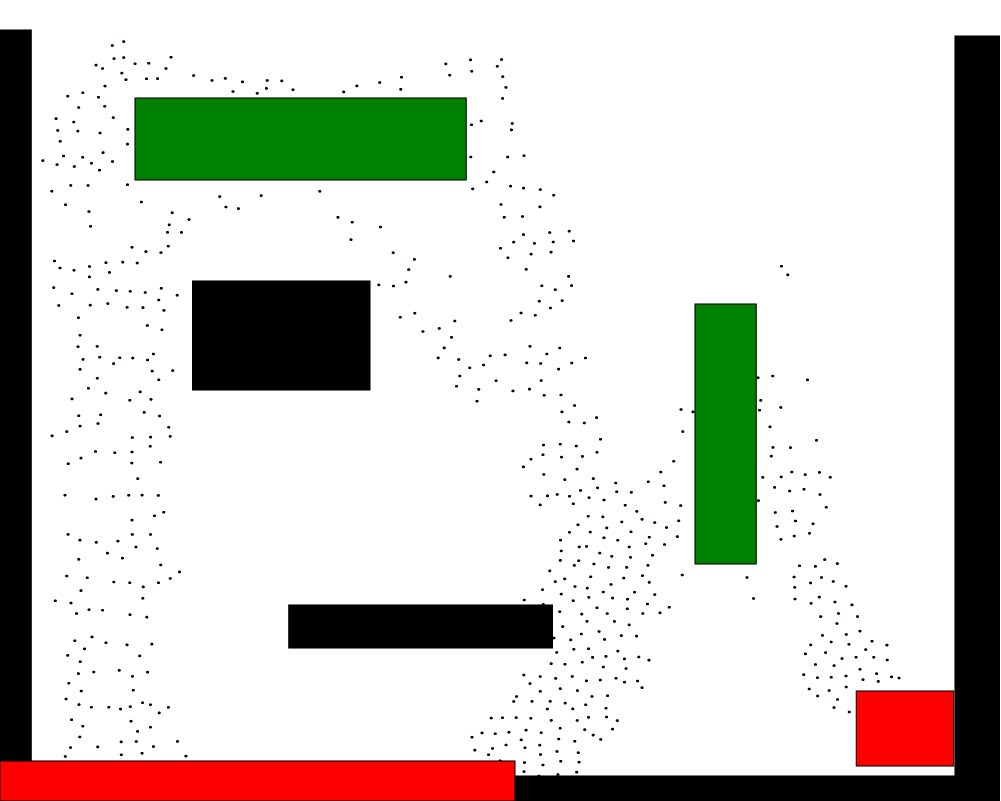}
	\captionof{figure}{Visualisation of Case D after 366 seconds.}
    \label{fig:case_d:0}
\end{minipage}%
\hfill
\begin{minipage}{.45\textwidth}
	\centering
	\includegraphics[width=\textwidth]{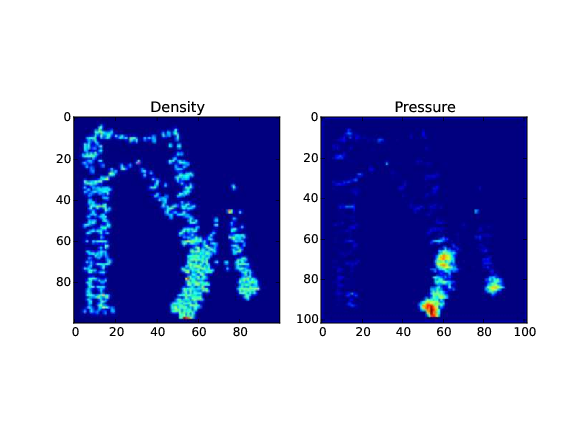}
    \captionof{figure}{Density and pressure for $t=366$.}
    \label{fig:case_d_fields}
\end{minipage}
\end{figure}

To check how the interaction potential influences the minimum distance, we plot the number of particles closer than some distance $r$ at the end of the simulation in Figure~\ref{fig:mde_vio}. The figure also shows the same graph for a \worddef{control simulation}, a simulation with exactly the same parameters, but where the interaction potential is suppressed so particles ignore each other.

\begin{figure}[h]
	\centering
    \includegraphics[width=0.5\textwidth]{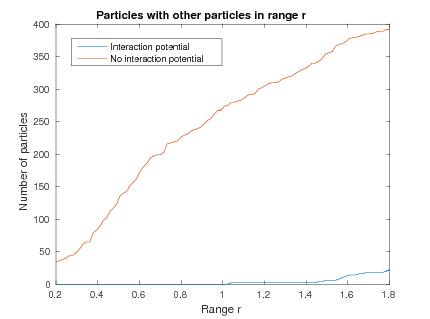}
    \captionof{figure}{Number of particles with neighbours within a certain range for the equilibrium situation at $t=366$. The interaction potential respects the interpedestrian distance.}
    \label{fig:mde_vio}
\end{figure}

Over average, the simulation violates the minimum distance for 0.50\% of the particles, while in the control simulation this occurs for 43\% of the particles.

The density heatmap is displayed in Figure~\ref{fig:case_d:logdens:1}, with the density heatmap for the control simulation in Figure~\ref{fig:case_d:logdens:1}.

\begin{figure}[h]
\centering
\begin{minipage}{.45\textwidth}
	\centering
	\includegraphics[width=\textwidth]{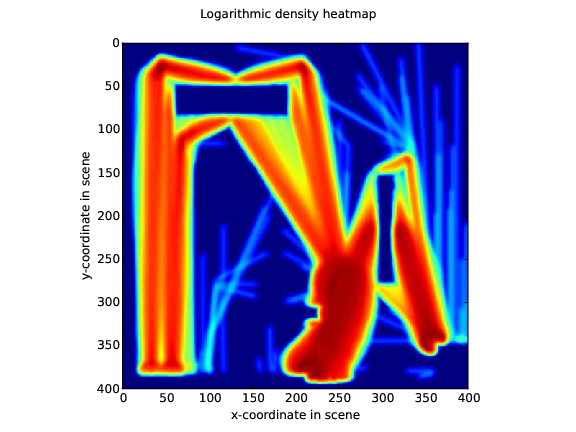}
    \captionof{figure}{Density heatmap for Case D with interaction potential.}
    \label{fig:case_d:logdens:1}
\end{minipage}%
\hfill
\begin{minipage}{.45\textwidth}
	\centering
	\includegraphics[width=\textwidth]{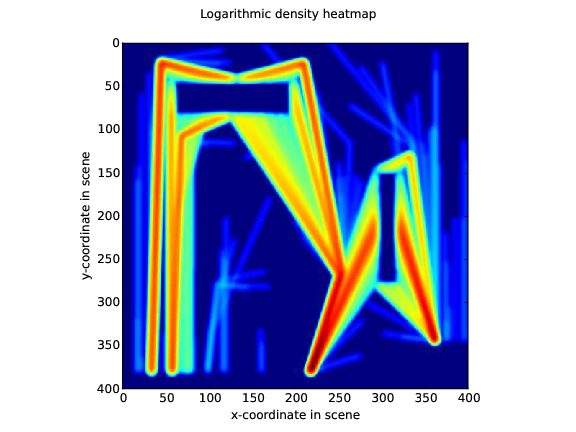}
    \captionof{figure}{Density heatmap for Case D without interaction potential.}
    \label{fig:case_d:logdens:2}
\end{minipage}
\end{figure}

\subsection{Discussion}
This system attains an equilibrium situation based on the inflow of the entrances.
While we don't prescribe an exact outflow, this does not mean the simulation is unconditionally stable. 
The interaction potential limits the total throughput, as can be observed from the heatmap in Figure~\ref{fig:case_d:logdens:1}.
With higher inflow values, the system would overflow.

While the interaction potential does not adhere exactly to the relation between minimum distance and maximum density posed in \eqref{eq:max_dens_result}, it still does a very good job in maintaining distances between pedestrians.
The heatmap shows us that the interaction potential succeeds in modelling repulsion between pedestrians by spreading their densities over larger parts of the scene.

However, the exact relation between the minimum distance and maximum density should be explored further, since this system does not satisfy the requirements of a closest packing structure.

% Validation and comparison of the results
\chapter{Validation and comparison of simulation results}
\label{chap:validation}
In this chapter we discuss simulation results from \emph{Mercurial} in specific crowd configurations. 
We validate the results by comparing observed pedestrian behaviour to phenomena described in the literature of pedestrian dynamics.
After that, we compare the two simulations from Section~\ref{sec:results} and Section~\ref{sec:results2} and disuss their strengths and weaknesses.

\section{Case E: Narrowing corridor}
\label{sec:case_e}
The first scenario is depicted in Figure~\ref{fig:narrow_0}. We initialise all pedestrians in the top section of the scene, The only exit is in the bottom section, so the pedestrians have to follow the funnel-like corridor.
The goal of this scenario is to investigate how well the simulation deals with many aggregated obstacles and slowly increasing densities.

\subsection{Choice of parameters}
We spawn 1000 pedestrians in the upper side of the scenario. Each pedestrian has a maximum speed of $2.0\meter\per\second$. The scene has a size of $70 \times 70\meter\squared$ while the exit has a width of 14\meter. 
The distance from the start of the funnel to the exit is 56\meter.
We suppress the influence of the global fluid solver to investigate the quality of the path planner and the collision avoidance. 
We maintain a minimal distance between pedestrians of 0.3\meter. The simulation is run with a time step of 0.05\second.
\subsection{Quantitative results}
\begin{figure}[h]
\centering
\begin{minipage}{.45\textwidth}
	\centering
	\includegraphics[width=0.5\textwidth]{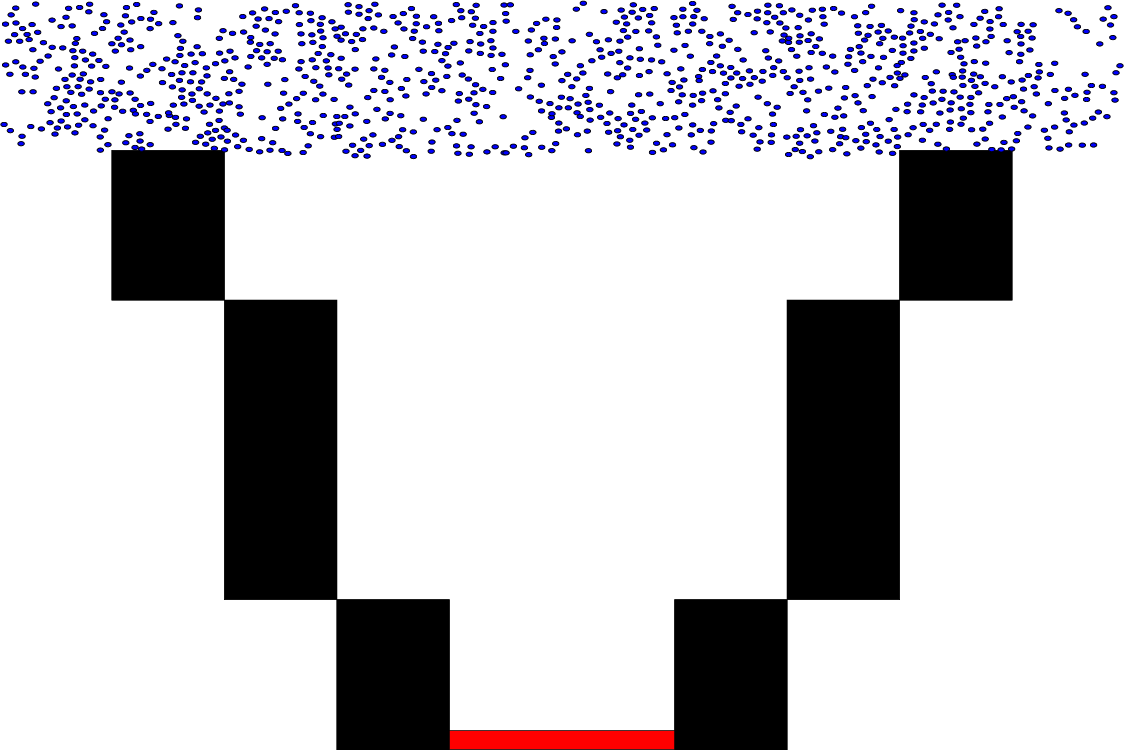}
	\caption{Initial state of the scenario. All pedestrians have to exit through the red rectangle South.}
	\label{fig:narrow_0}
\end{minipage}%
\hfill
\begin{minipage}{.45\textwidth}
	\centering
	\includegraphics[width=0.5\textwidth]{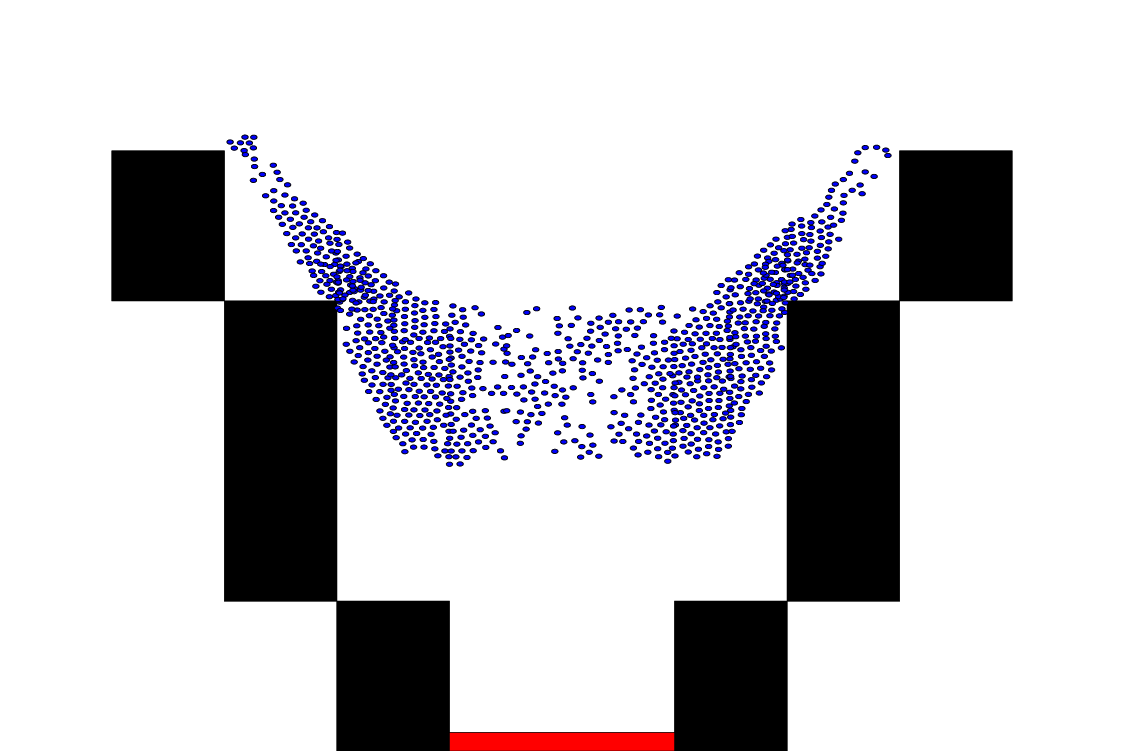}
	\caption{State of the scenario after 14 seconds.}
	\label{fig:narrow_14}
\end{minipage}
\begin{minipage}{.45\textwidth}
	\centering
	\includegraphics[width=0.5\textwidth]{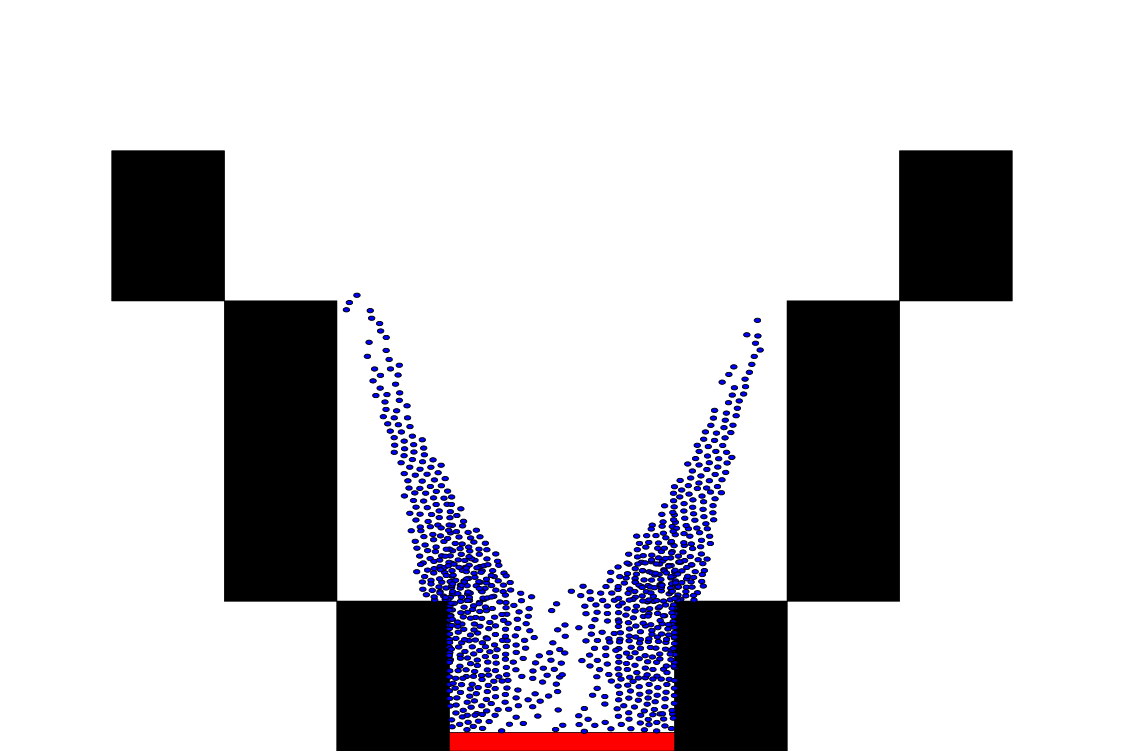}
	\caption{State of the scenario after 28 seconds.}
	\label{fig:narrow_28}
\end{minipage}
\end{figure}
Figure~\ref{fig:narrow_14} and Figure~\ref{fig:narrow_28} show the scene after 14 respectively 28 seconds. It is difficult to capture the dynamic nature of a simulation in a single image; a movie is always preferable. Still, we try to give an impression of the characteristic motions of the pedestrians by showing the state of the crowd in several states of evacuation.
Notice the crowd congestion takes place close to the obstacles, while in the centre density has barely increased. Left and right we observe trails of pedestrians pushed back for exceeding the maximum density. The density is highest near the corners of the obstacles. This is visible in the density plot in Figure~\ref{fig:narrow_14_field}.
\begin{figure}[h]
    \centering
    \includegraphics[width=0.7\textwidth]{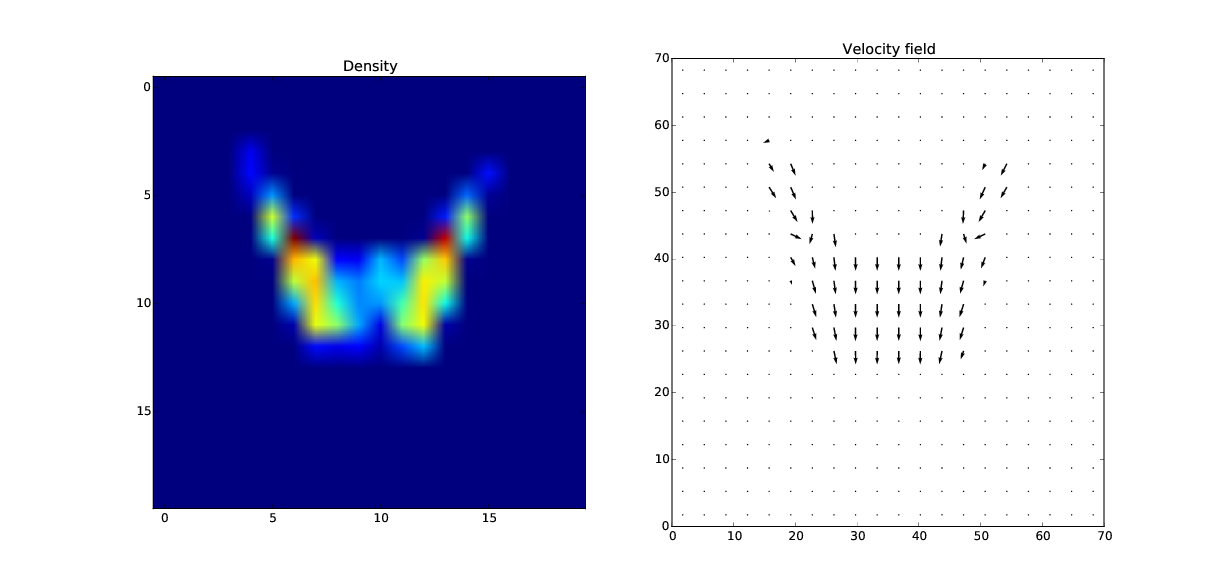}
	\caption{Density field plot corresponding to Figure~\ref{fig:narrow_14}.}
    \label{fig:narrow_14_field}
\end{figure}
The scene is cleared after $53$ seconds, almost twice as long as the first pedestrian needs to reach the exit. This shows that the amount of congestion has a large effect on the evacuation of the scene. 

To support this observation, Figure~\ref{fig:narrow_observed} shows a histogram of the pedestrians exit times. We observe a widespread distribution, while the planned times in the scene (plotted in Figure~\ref{fig:narrow_planned}) are a lot more concentrated. Notice that some pedestrians reach the exit faster than planned, indicating that they have exceeded their maximum velocity. This occurs due to the collision avoidance in combination with the increasing density, pushing some pedestrians forward.
\begin{figure}[h]
\centering
\begin{minipage}{.45\textwidth}
    \centering
    \includegraphics[width=\textwidth]{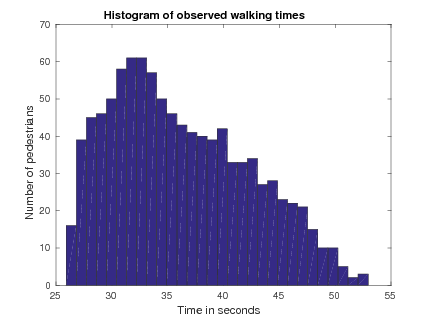}
    \caption{Histogram of the time each pedestrian spent in the scene. This also serves as a measure of throughput for the exit.}
    \label{fig:narrow_observed}
\end{minipage}%
\hfill
\begin{minipage}{.45\textwidth}
    \centering
    \includegraphics[width=\textwidth]{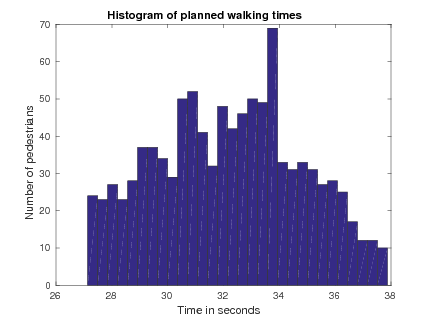}
    \caption{Histogram of the time each pedestrian plans to be in the scene, according to their path to exit. This represents the walking time in case of no other pedestrians.}
    \label{fig:narrow_planned}
\end{minipage}
\end{figure}

In Figure~\ref{fig:narrow_time} we plot the time to exit as a function of the initial location in a scatter plot. This plot reveals 'hot spots' of locations which are bound to have a long evacuation time.
In Figure~\ref{fig:narrow_delay} we plot the relative delay as a function of initial location. 
\begin{figure}[h]
\centering
\begin{minipage}{.45\textwidth}
	\centering
	\includegraphics[width=\textwidth]{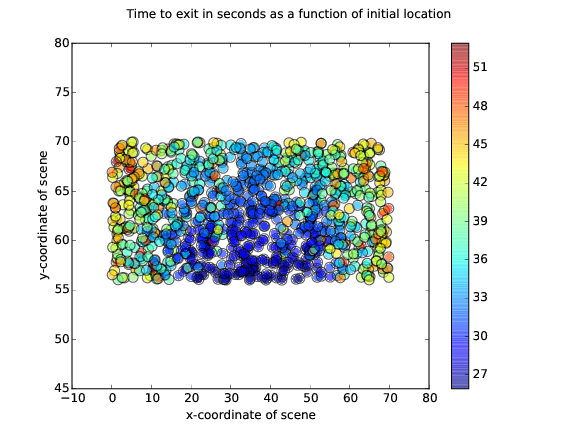}
	\captionof{figure}{Walking time to exit as a function of initial location.}
	\label{fig:narrow_time}
\end{minipage}%
\hfill
\begin{minipage}{.45\textwidth}
	\centering
	\includegraphics[width=\textwidth]{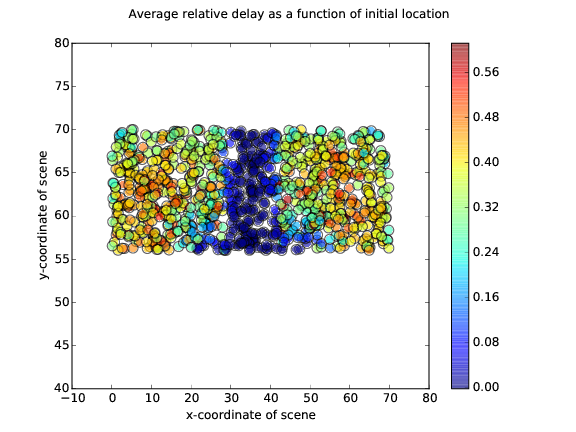}
	\captionof{figure}{Experienced delay as a function of initial location.}
	\label{fig:narrow_delay}
\end{minipage}
\end{figure}

The pedestrian walking times are smoothly distributed. The pedestrians spawned in the bottom centre exit first, and the time spent in the scene increases in a radially symmetric fashion.\\
The delay has a less continuous distribution. This is caused by the fact that all pedestrians plan a path directly to the exit, but only the pedestrians in the centre are able to maintain that path.
All the other pedestrians have to divert from their path, creating a significant delay. This is also visible in the plot of pedestrian paths, provided in Figure~\ref{fig:narrow_paths}.\\
\begin{figure}[h!]
	\centering
	\includegraphics[width=0.7\textwidth]{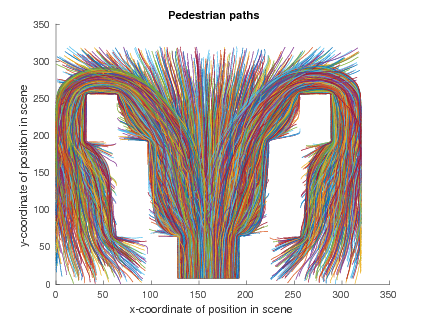}
	\caption{Paths from Case E their initial position to the exit.}
	\label{fig:narrow_paths}
\end{figure}\\
We run multiple simulations, each time increasing the number of pedestrians spawned. When comparing the density to the average delay we obtain the plot in Figure~\ref{fig:delay_vs_density}.\\
\begin{figure}[h!]
    \centering
    \includegraphics[width=0.7\textwidth]{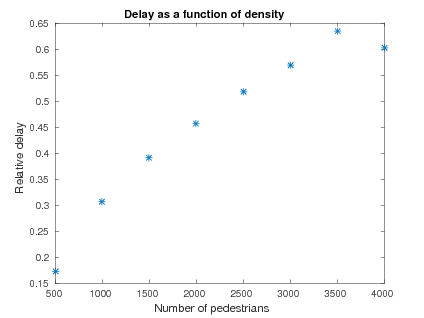}
    \caption{Delay measured over several simulations with varying densities.}
    \label{fig:delay_vs_density}
\end{figure}\\
Of course, the delay increases with the density. 
But Figure~\ref{fig:delay_vs_density} show that for high density values the increase in delay becomes smaller. 
It would be interesting to examine the limiting behaviour of the delay.
This would require running the simulations with extremely high densities.
These values fall outside the model range, so they cannot provide relevant information.
\subsection{Discussion}
In the attempt to model a funnel, a limitation of the obstacle setup and the path planning algorithms is that all obstacle faces have to be horizontal or vertical.
Straight lines with different angles, like a diagonally placed wall, can only be approximated by placing a number of smaller obstacles on that line. 
While the path planner ensures pedestrians are correctly manoeuvred to the exit, when the number of pedestrians becomes large, their paths become less smooth and the path planning takes more time.
In the simulation this is observed in Figure~\ref{fig:narrow_14}, where the pedestrians near the obstacles have difficulty passing their checkpoint. This is caused by having the same checkpoints for all pedestrians, while only a small number of pedestrians can simultaneously be present in that location.

Although it is not unlikely to have congestions near the boundaries of the funnel, the location of the congestions seems unintuitive since it could easily be avoided by the pedestrians.
A strong point of the algorithm is that in spite of fixed angles for the walls and obstacles, the pedestrians path angle is not restricted. This means that smooth diagonal paths are generated, respecting the pedestrians intent to reach the exit as soon as possible. This is visible in Figure~\ref{fig:narrow_paths}, where the paths for the leftmost and rightmost pedestrians are diagonal at first, but become increasingly vertical as they approach their goal.

Observing the propagation of the crowd, apart from curves, their paths follow the shape of the funnel. However, the paths in the centre of the funnel are more straight than to be expected for a crowd this dense. This is caused by the low density in the centre of the crowd, as visible in Figure~\ref{fig:narrow_14}. No congestion happens in the centre, so the pedestrians there do not need to deviate from their original paths.

The low density in the centre is caused by the inflexibility of the path. Pedestrians are able to deviate from the planned path as long as they pass within a certain radius of their checkpoints.
This means that in case of congestions, pedestrians will wait until the blocked path is free, instead of passing around the blockage to regions with a lower density. This causes both the low density in the centre as well as the trails of pedestrians near the edges of the funnel.
This behaviour is reminiscent of laminar flow of a fluid through a pipe; high velocity in the centre, low velocity near the walls.

Finally, we examine the rate of pedestrians leaving the exit. Figure~\ref{fig:narrow_time} shows that the time spent in the scene is lowest for people closest to the exit and from there increases gradually. This is in accordance with the histogram in Figure~\ref{fig:narrow_observed}. \\
Besides the distribution of pedestrian exit times, the histogram shows us something else; the maximum throughput of the exit. Looking at the shape of the histogram, after the first pedestrian reaches the exit, the throughput increases up to almost 4 pedestrians per time step.
After reaching this maximum, the throughput gradually decreases until the final pedestrian leaves the scene.
This is consistent with Helbing and Still. They observe a normal distributed throughput in their evacuation and corridor simulations. While the time distribution is obviously skewed, the amount of variation seems consistent.
Nevertheless, in this model we miss the outliers a normal distribution would have. This is partly caused by the fact that the pedestrian walking speeds are drawn from a uniform distribution.

We compare the density versus delay plot with the results found in \cite{guy10}. Apart from the outlier at 3500 pedestrians, the same behaviour is observed: a higher density implies a more delay, but only up to a certain density. The simulation in \cite{guy10} has been executed for higher, unsafe densities and they show asymptotic behaviour in the delay as a function of density.
\section{Case F: Moving dense crowd}
\label{sec:case_f}
The second scenario models an open space with a packed crowd. Upon starting the simulation, the crowd collectively begins to move towards the exit.
The scenario with the initial crowd configuration is depicted in Figure~\ref{fig:race_0}.
The width of the exit is chosen significantly smaller than the comfort radius of the crowd (which is visible in Figure~\ref{fig:race_11}. This way, we observe the effects of the exit size on the configuration of the pedestrians. 
\begin{figure}[h!]
\centering
\begin{minipage}{.5\textwidth}
	\centering
	\includegraphics[width=\textwidth]{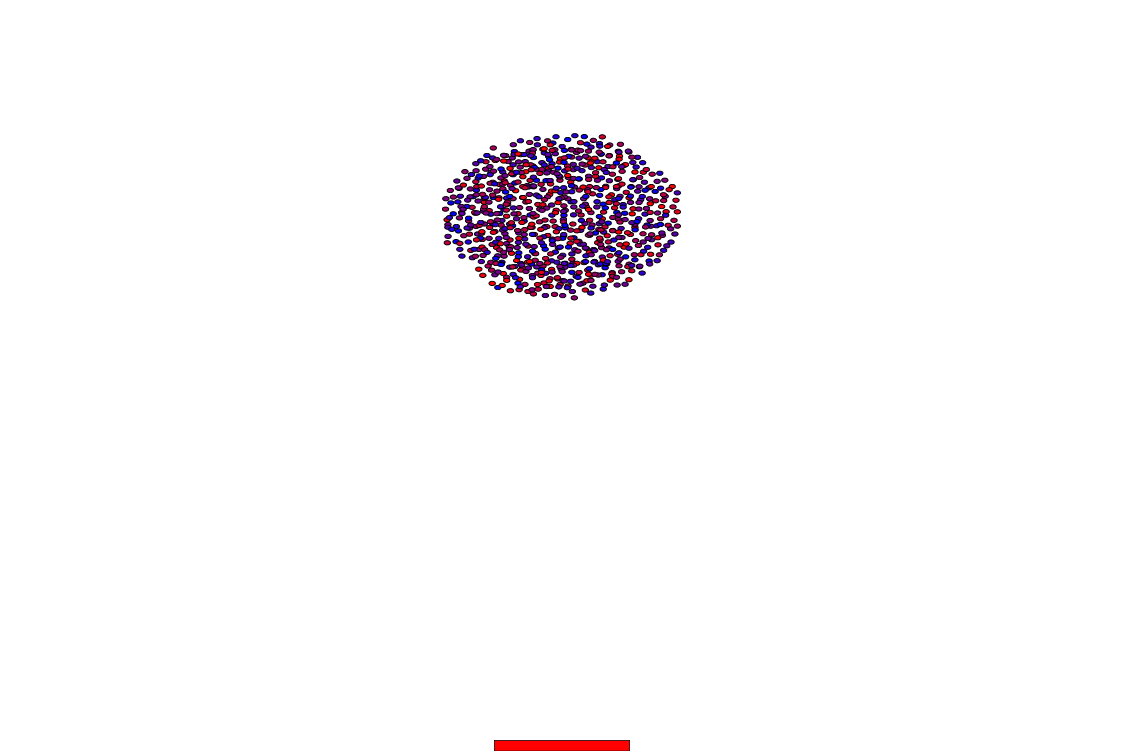}
	\caption{Initial pedestrian distribution in scenario 2.}
	\label{fig:race_0}
\end{minipage}%
\begin{minipage}{.5\textwidth}
    \centering
	\includegraphics[width=\textwidth]{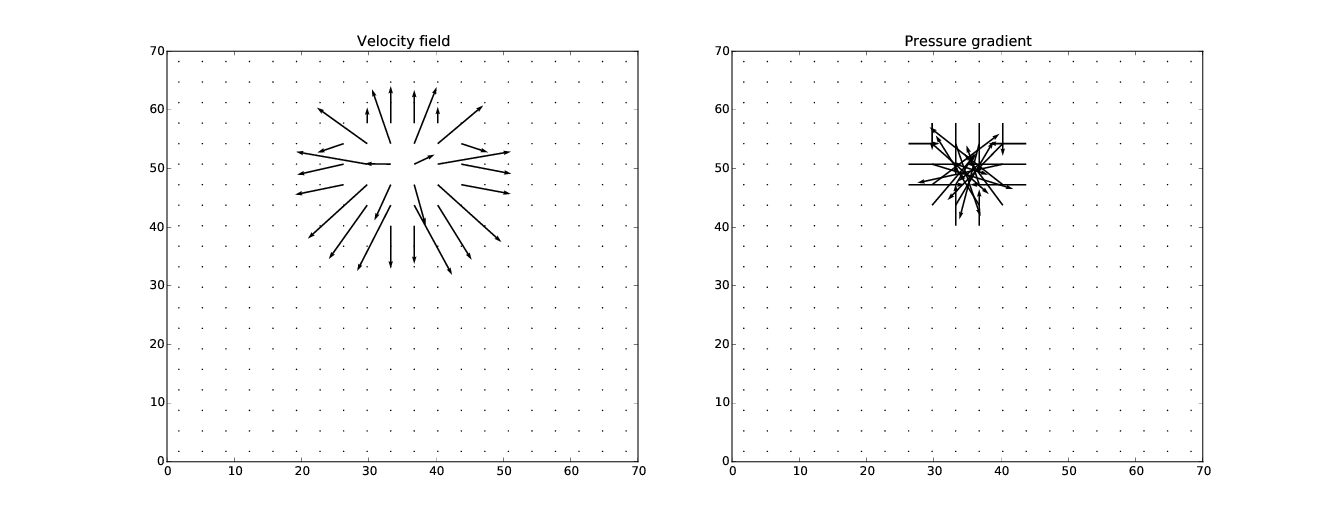}
    \caption{Pressure gradient and velocity field corresponding to Figure~\ref{fig:race_0}.}
	\label{fig:race_0_fields}
\end{minipage}
\begin{minipage}{.5\textwidth}
	\centering
	\includegraphics[width=\textwidth]{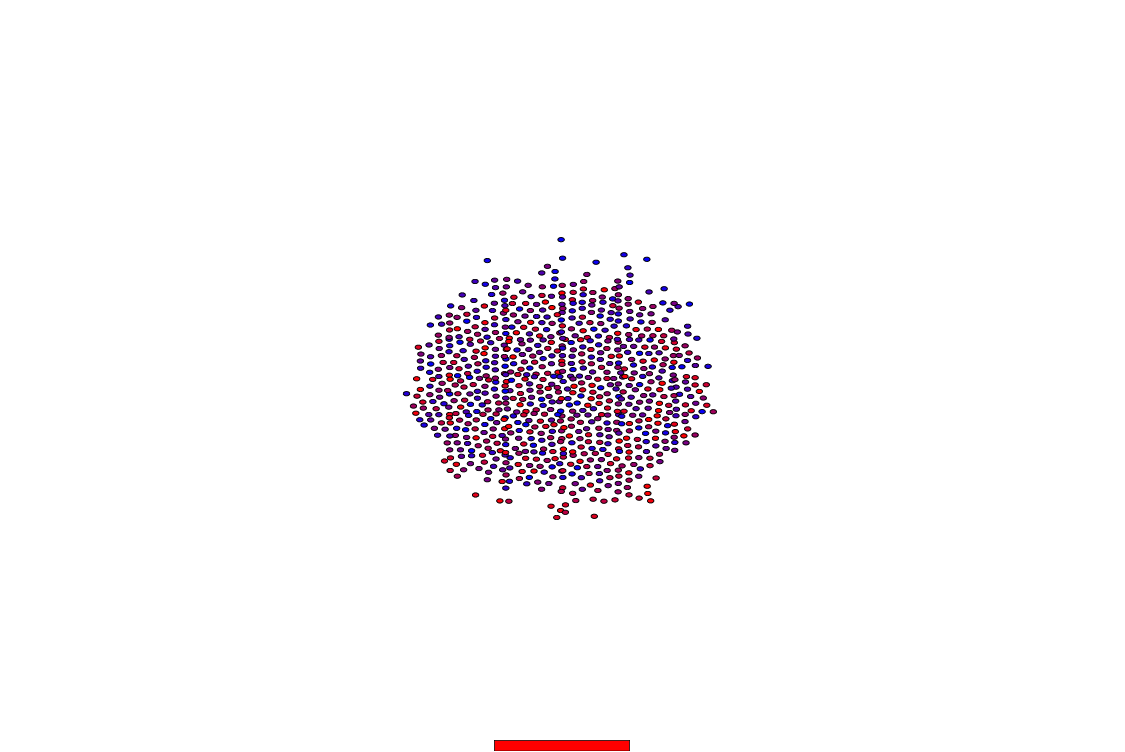}
	\caption{Scene after 11 seconds. \\The crowd expands to a comfortable density.}
	\label{fig:race_11}
\end{minipage}%
\begin{minipage}{.5\textwidth}
	\centering
	\includegraphics[width=\textwidth]{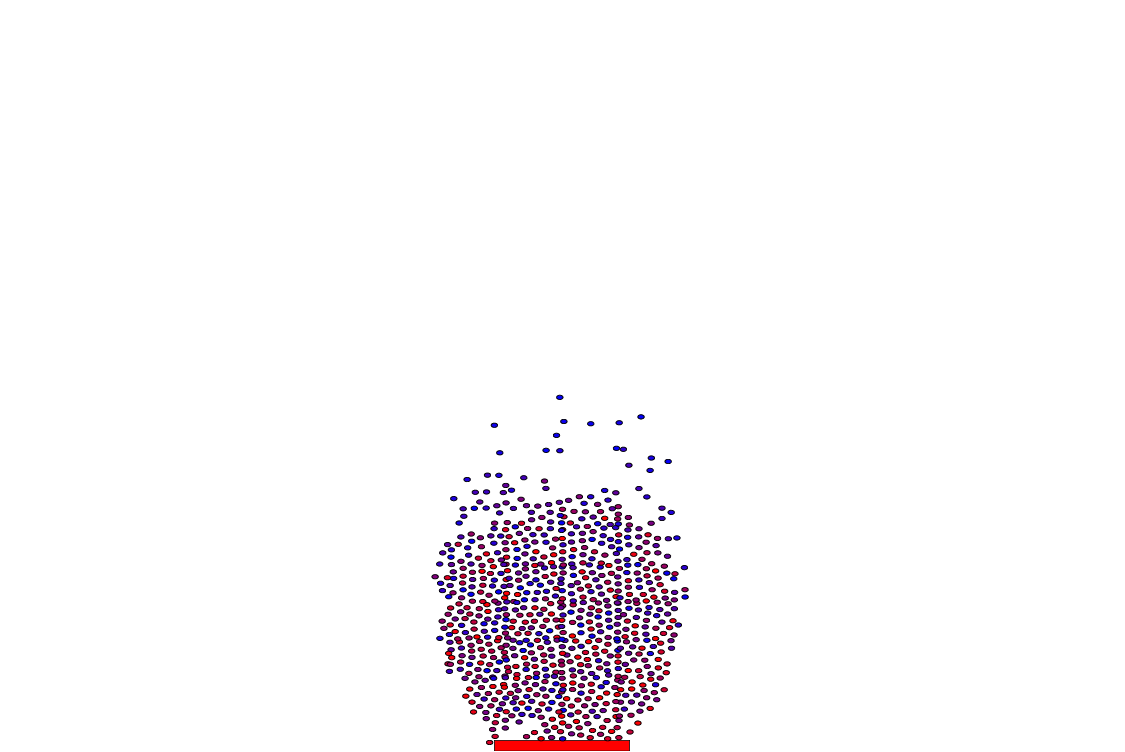}
	\caption{Scene after 25 seconds. The crowd reshapes to fit through the exit.}
	\label{fig:race_25}
\end{minipage}
\end{figure}

\subsection{Choice of parameters}
We initialise a scene with the same size parameters, $70 \times 70 \meter\squared$. 
In this scene, we assume a minimum distance of $0.1\meter$.
With the assumption that a pedestrian has a radius 0.2, this corresponds to a maximum density of $\rho_{\max} = 3.4\meter\rpsquared$, according to \eqref{eq:max_dens_result}.
We spawn a dense crowd in a circle with a radius of 7 meters, as shown in Figure~\ref{fig:race_0}. If we want to satisfy the maximum density, then up to $3.4\cdot \pi*7^2\approx523$ people can be present. 
We spawn an initial crowd of 800 people to violate the maximum density, but remain within physical boundaries. The initial crowd is distributed uniformly, resulting in an average density of $\frac{800}{7^2*\pi} = 5.2\meter\rpsquared$, corresponding with an minimal distance of approximately $0.07\meter$. We exceed the maximum density to test the influence of the pressure.

This time, we pick a pedestrians maximum speed from a uniform distribution of the interval $[1,2]\meter\per\second$. The density is interpolated by a kernel with smoothing length $h=1.75\meter$ meters and support radius $2h$. 
Roughly, this means a pedestrian only feels pressured by other pedestrians within a distance of $3.5$\meter.
The exit has a width of $8.4$\meter.
\subsection{Quantitative results}
Figure~\ref{fig:race_11} and Figure~\ref{fig:race_25} show two snapshots of the simulation at a time of respectively 11 and 25 seconds. 
In Figure~\ref{fig:race_11} the radius of the crowd has increased. 
This is a direct effect of the pressure depicted in Figure~\ref{fig:race_0_fields}: while the crowd exceeds the maximum density, the pressure gradient is computed and subtracted from the velocity. 
In Figure~\ref{fig:race_25} we observe that while time progresses, the crowd gains a more oval shape. 
This is caused by both the path planner and the pressure field. 
The path planner herds the crowd into a form that fits through the exit, while the pressure field ensures the pedestrians do not exceed maximum density by slowing down the pedestrians in the centre of the crowd.

The scene is cleared after 57.5 seconds. We plot the distribution of walking times in Figure~\ref{fig:race_planned}. The spread is wider than in Scenario 1, although the planned walking times (in Figure~\ref{fig:race_planned}) are more concentrated. This is a consequence of the pressure, which causes a significant delay for the pedestrians in the centre. This observation is supported by Figure~\ref{fig:race_delay}, where the relative delay for the outer pedestrians is low, and increases inward.

In Figure~\ref{fig:race_time}, the time scatter plot is depicted. This shows the expected increase in time spent the further a pedestrian is located to the exit. Notice the effect of varying pedestrian speeds; two pedestrians' initial locations might be close, in contrast to their exit times.
\begin{figure}[h]
\centering
\begin{minipage}{.45\textwidth}
	\centering
	\includegraphics[width=\textwidth]{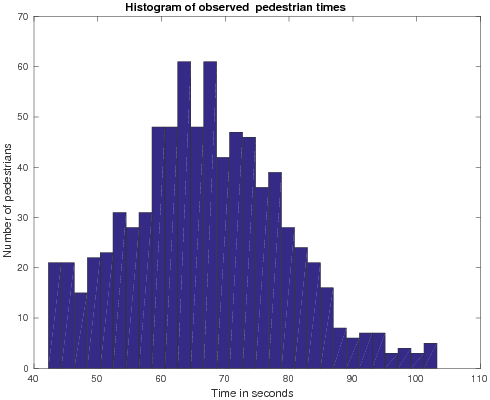}
	\caption{Observed walking time histogram.}
	\label{fig:race_observed}
\end{minipage}%
\hfill
\begin{minipage}{.45\textwidth}
    \centering
    \includegraphics[width=\textwidth]{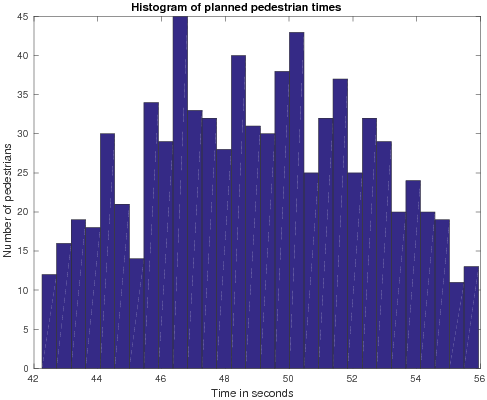}
    \caption{Planned walking time histogram.}
    \label{fig:race_planned}
\end{minipage}
\end{figure}
\\
\begin{figure}[h]
\centering
\begin{minipage}{.45\textwidth}
	\centering
	\includegraphics[width=\textwidth]{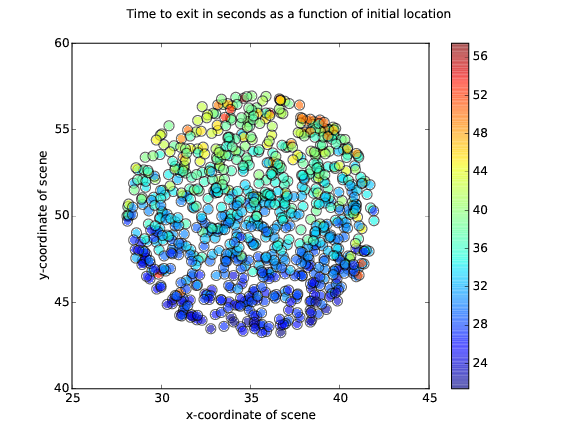}
	\captionof{figure}{Walking time to exit as a function of initial location.}
	\label{fig:race_time}
\end{minipage}%
\hfill
\begin{minipage}{.45\textwidth}
	\centering
	\includegraphics[width=\textwidth]{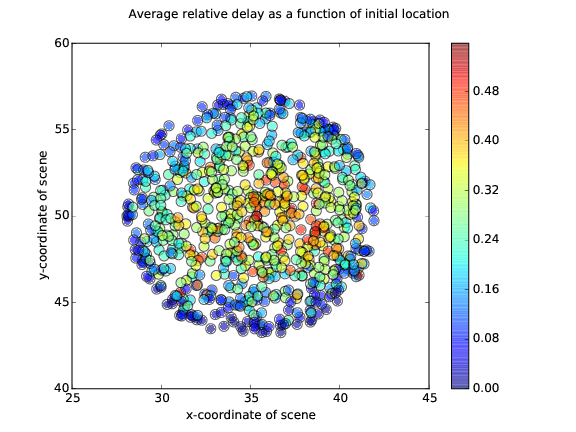}
	\captionof{figure}{Experienced delay as a function of initial location.}
	\label{fig:race_delay}
\end{minipage}
\end{figure}

\begin{figure}[h!]
	\centering
	\includegraphics[width=0.7\textwidth]{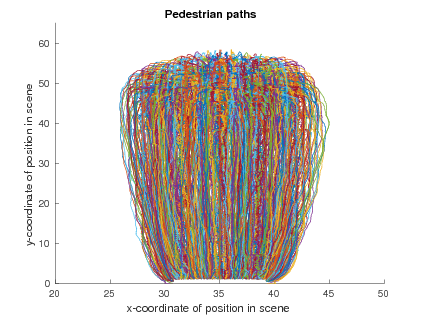}
	\caption{Paths of each pedestrian of case F from their initial position to the exits.}
	\label{fig:race_paths}
\end{figure}
\subsection{Discussion}
Figure~\ref{fig:race_0} to \ref{fig:race_25} show the evolution of a crowd under pressure. 
The grid is relatively coarse, but the bilinear interpolation does a good job in evaluating the pressure forces for each pedestrian, resulting in a radially symmetric dispersing crowd. 
We also observe the effect of the combination of the maximum density and the size of the exit on the global nature of the crowd.
Figure~\ref{fig:race_delay} shows the people at the edge of the crowd move faster than the people in the centre. 
The pedestrians in the centre of the crowd have a reduced freedom of motion, and have to adapt their velocities to the slowest in the group.
This is in accordance with the results in \cite{guy10}, where they report of the 'Edge-Effect Phenomena' which causes a speed drop of approximately 30\% for pedestrians in the centre of a crowd. Our maximum speed drop exceeds that value, but it should be noted that the speed range of the pedestrians is quite large.

No congestion takes place at the exit. While the crowd respects the maximum density, no constraint was set for the throughput of the exit. When such constraints are set, clogging is observed, like in \cite{helbing95}.
Finally, the observed paths of the pedestrians are smooth and we observe random deviations in the paths.

\section{Comparison}
\label{sec:comparison}
We compare the performance of the model based on the domain potential, discussed in Section~\ref{sec:crowds1} and the model based on the interaction potential, discussed in Section~\ref{sec:crowds2}.

The pedestrian interaction between the models is quite different.
The domain potential causes the pedestrians to follow a minimum cost principle. The chosen paths try to minimize discomfort, and to that extent, avoid regions with higher densities.
This leads to smooth paths and crowds that quickly occupy the entire accessible domain, but it also causes oscillating behaviour in these paths if densities are volatile. This is undesired, as this switching of direction in a real life setting does not adhere to the minimal cost principle; pedestrians have to spend energy to change course.

When obstacles are taken into account in the discretisation, they are neatly avoided by the pedestrians. The fast marching method ensures that even for non-convex obstacles pedestrians are able to find their destination.

The interaction potential causes this spreading to happen only on a local basis when the maximum density is violated. The minimum distance that pedestrians try to attain does not depend on the observed density.
Repulsion is modelled with a unilateral incompressibility constraint which is only activated when the density exceeds some threshold. This means deadlocks can occur when multiple pedestrians have the same location as their destination.
Pedestrians move closer to the destination, the density around the destination increases, which causes the pressure to activate and repel the crowd from that destination.
If almost none of the pedestrians managed to reach this destination, the density does not decrease and the same behaviour occurs the next time step.
The result is a {pulsing}-like behaviour. 
Because this is inherent to the way interaction is modelled, the only way to alleviate is to ensure that in simulations the smoothing length is chosen significantly smaller than the size of destinations and the width of corridors.
In complex geometries like buildings this requires a high resolution grid, but we have seen in Section~\ref{sec:lcp} that solving large sparse LCP systems should pose no problem.

Using the domain potential, it is a challenge to find a well performing set of parameters. Most of the coefficients have no physical meaning and that makes it difficult to modify the dimensions of the domain or the pedestrian-specific parameters.
The interaction potential is defined solely by the macroscopic quantities like pressure, density and velocity. This provides a better intuition as to how the parameters should be chosen.

\section{Combination}
It is possible to strip down the domain potential method in Section~\ref{sec:crowds1} in such a way that can be combined with the interaction potential.
When we make the unit cost field independent of the density, we are able to model inhomogeneous domains and have a fully continuous path planner. 
Interaction can then be modelled at runtime by using a interaction potential. 

Computationally, this is advantageous. Since the path planner is now a function independent of time and defined on the entire domain $\Omega$ it only needs to be computed once. At runtime, only the interaction potential needs to be evaluated, making for very efficient simulations.
In addition, because the domain potential is defined everywhere, the pulsing-like behaviour for specific waypoints is removed.

We showcase one simulation in which we provide the results using a combined planner.
\section{Case G: Large indoor domain}
We simulate a scenario corresponding to a similar implementation in \cite{hoeven16}, in which a similar crowd dynamics simulation was constructed.
This scenario was built to compare and validate the two implementations.
The scenario is shown in Figure~\ref{fig:combi_scene}. Pedestrians are spawned from each of the nine entrances.

\subsection{Choice of parameters}
The domain has a size of $30\times30\meter\squared$ and is simulated with a time step of $\Delta t=0.05\second$.
We use a $100\times100$ grid and simulate the system for 157 seconds, in which the simulation has attained a stable state with 560 pedestrians present. 
We prescribe a minimal distance of $0.4\meter$.

\subsection{Results}
Most of the results are similar to the other test cases and are therefore omitted. 
We want to focus on pedestrian distribution in the domain and the minimum distance violations.
We show the paths in Figure~\ref{fig:combi_paths} and the density heatmap in Figure~\ref{fig:combi_heatmap}. The minimum distance violations are shown in Figure~\ref{fig:combi_mde}.

\begin{figure}[h]
\centering
\begin{minipage}{.45\textwidth}
	\centering
	\includegraphics[width=\textwidth]{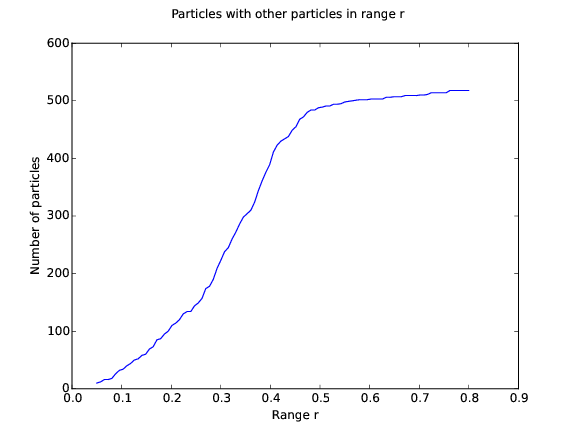}
    \caption{Neighbours within range $r$ for various values of $r$.}
    \label{fig:combi_mde}
\end{minipage}%
\hfill
\begin{minipage}{.45\textwidth}
	\centering
    \includegraphics[width=\textwidth]{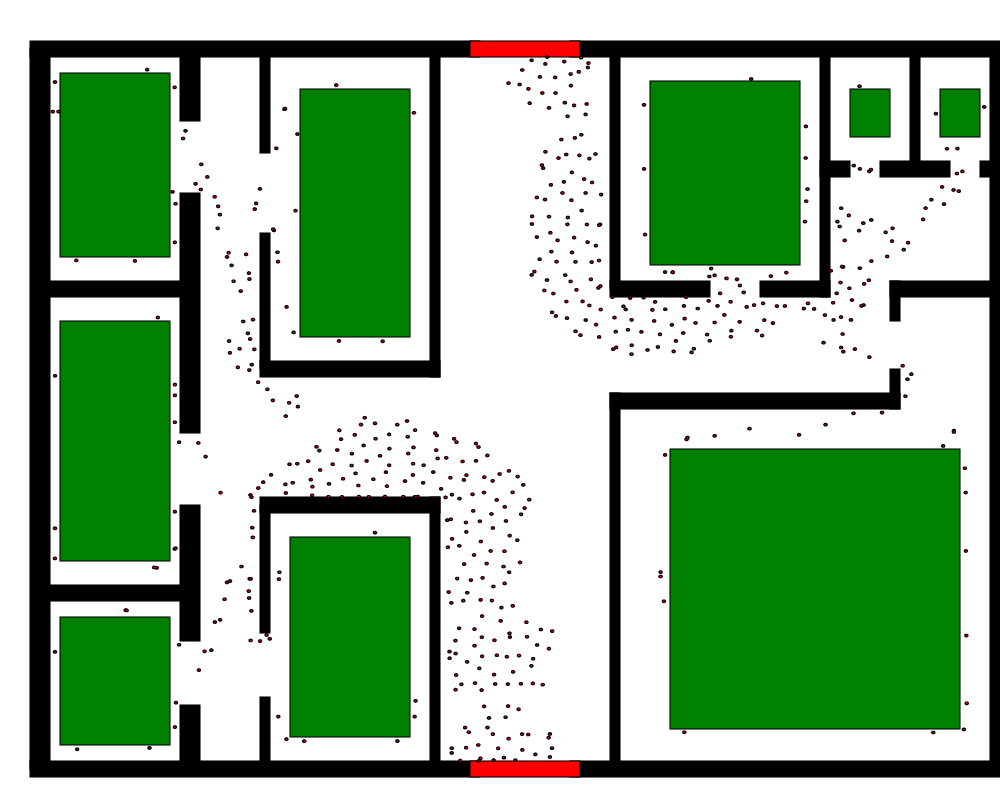}
    \caption{Scenario used for simulation combining potentials.}
    \label{fig:combi_scene}
\end{minipage}
\end{figure}

\begin{figure}[h]
\centering
\begin{minipage}{.45\textwidth}
	\centering
	\includegraphics[width=\textwidth]{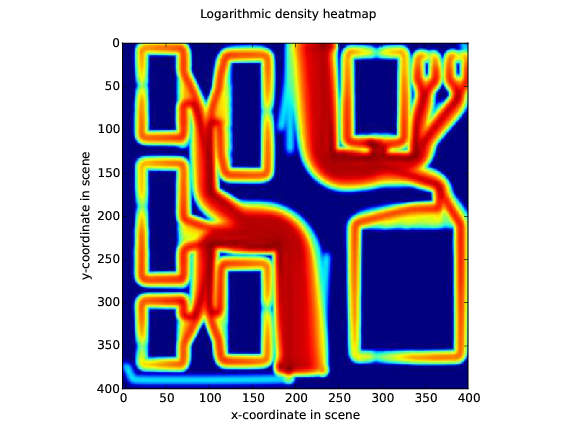}
    \captionof{figure}{Density heatmap for Case G with interaction potential.}
    \label{fig:combi_heatmap}
\end{minipage}%
\hfill
\begin{minipage}{.45\textwidth}
	\centering
	\includegraphics[width=\textwidth]{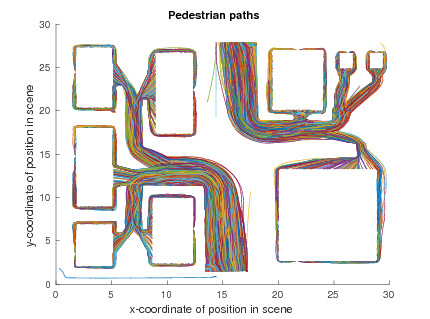}
    \captionof{figure}{Density heatmap for Case G without interaction potential.}
    \label{fig:combi_paths}
\end{minipage}
\end{figure}
Figure~\ref{fig:combi_paths} shows the paths are smooth and non-oscillatory, while the heatmap in Figure~\ref{fig:combi_heatmap} shows the pedestrians spread out over the scene as a result of interaction.

Figure~\ref{fig:combi_mde} shows a violation of the minimum density. However, this is to be expected in a simulation domain as complex as this one.
In spite of the complex geometry and the high interaction, one iteration takes approximately $0.56\second$. This means the simulation runs virtually interactive. This shows the huge speed-up with respect to the previous two implementations.

% Conclusion
\chapter{Conclusions}
With this chapter, we conclude this thesis. Here we shortly summarize the results achieved in this work, discuss some of the limitations and identify opportunities for future research.
\section{Summary of results}
In Chapter~\ref{chap:analysis}, we started by exploring the notion of transport in particle systems on a microscopic and macroscopic level.
By defining interaction and domain potentials, we saw that we are able to model interaction and inhomogeneous domains on both modelling levels.
Using interaction potentials, we showed that under certain conditions, limiting behaviour of large-scale particle models converges to solutions of multispecied interactive transport systems. 

We continued to investigate the coupling between these modelling scales in Chapter~\ref{sec:micro_macro}. 
We used an SPH-based interpolation method to translate microscopic measure to a macroscopic representation.
We showed that if we related the smoothing length of this method to the minimum distance between particles, it is possible to obtain a consistent approximation of the microscopic system on macroscale.

In Chapter~\ref{chap:crowds} we provided a literature review for the most popular way of modelling crowds. 
We gave examples of particle models, lattice models and PDE-like models, and discussed their strengths and weaknesses.
We implemented two models, one based on a domain potential and one based on an interaction potential, and used these implementations to simulate various scenarios.

From the simulation results we concluded that the domain potential is a valuable tool in modelling inhomogeneous domains, but is less capable in dealing with dynamical features like preventing congestions.
The interaction potential is better suited for this tasks, but needs a path planner that is able to cope with sudden changes in direction.
We implemented a new path planner based on a visibility graph, creating smoothed shortest paths to the goal.
This path planner works well in combination with the interaction potential for stable simulations, but has a difficult time correctly representing crowd motion in areas with high densities and obstacles.

In Chapter~\ref{chap:validation} we assessed the results and validated them by comparing observed effects to results of experiments found in literature.
We found the simulation replicated various interaction effects observed in real life experiments and other simulations.
\section{Future research}
This thesis leaves many aspects open for further exploration. Below we mention a few extensions.
\begin{itemize}
\item \textbf{Multiple kinds of inhomogeneities}\\
In our simulations, we restricted ourselves to domains with inaccessible objects, essentially modelling domains with two permeability values.
This can be generalised to a range of permeability values to represent different kinds of inhomogeneities. 
Macroscopically, several models exist that are able to represent these situations, but on a microscopical level this is less understood.
\item \textbf{Different kinds of interaction}\\
We focused on using the interaction potential to model short-range repulsion. 
In other interaction-dominated models, like predator-prey or leader-follower models, (long-range) attraction plays an important role as well.
With the right macroscopic formulation, these situations could be modelled as well.
\item  \textbf{Multiple populations}\\
Our implementations only involve pedestrians in single populations. 
The interactions between different types of pedestrians are known to cause phenomena like lane formation and vortices.
Including social dynamics into the crowd dynamics context can make for more realistic simulations.
\item \textbf{Different implementation language for Mercurial}\\
\emph{Mercurial} is built with performance in mind, but having an implementation in a compiled language makes a lot of difference. 
Huge speed-ups could also be gained in parallel implementations. 
Both the microscopic as the macroscopic model use local information in computing the evolution of the system. 
This means the computational domain could be subdivided on different processors to gain large speed-ups in simulations.
\end{itemize}

\bibliography{literature} 
\bibliographystyle{alpha}
\begin{appendices}
% Paper published with Hong and Muntean
\chapter{Draft for publication}
\label{chap:paper}
The following paper is a draft submitted to the \emph{European Physics Journal} for the \emph{Special Topics} edition on statistical mechanics in inhomogeneous domains.\\
It is a result of the collaboration between a researcher from the University of Warwick and the author and main supervisor of this thesis.



# Discrete and continuum links to a nonlinear coupled transport problem of interacting populations

M.H. Duong[1], A. Muntean[2], and O.M. Richardson[3]

[1] University of Warwick, UK
[2] Karlstad University, Sweden
[3] Eindhoven University of Technology, NL

**Abstract.** We are interested in exploring interacting particle systems that can be seen as microscopic models for a particular structure of coupled transport flux arising when different populations are jointly evolving. The scenarios we have in mind are inspired by the dynamics of pedestrian flows in open spaces and are intimately connected to cross-diffusion and thermo-diffusion problems holding a variational structure. The tools we use include are a suitable structure of the relative entropy controlling TV-norms, the construction of Lyapunov functionals and particular closed-form solutions to nonlinear transport equations, a hydrodynamics limiting procedure due to Philipowski, as well as the construction of numerical approximates to both the continuum limit problem in 2D and to the original interacting particle systems.

## 1 Introduction

The starting point of the results presented in this paper is the following question[1]:
*Can one design a system of interacting particles that converges in some suitable limit to the following system of nonlinearly coupled system of transport equations:*

$$\begin{cases} \partial_t u = (u(u+v)_x)_x, & \text{(1a)} \\ \partial_t v = (v(u+v)_x)_x, & \text{(1b)} \end{cases}$$

*with initial conditions $u(0) = u_0(x)$ and $v(0) = v_0(x)$ $(x \in \mathbb{R}^3)$?* Here $u$ and $v$ refer to mass concentrations of some chemical species which are participating in a non-competitive manner in a joint transport process. The background of the question (and interest of Mimura) is connected to the role pheromones play in influencing the aggregation phenomenon, one of the main survival mechanisms in insects, birds and animal colonies; we refer the reader to [6] for more on this context. It is worth noting the coupled structure of the transport fluxes resembles situations arising in

---

[1] This question was posed by Prof. M. Mimura (Meiji, Tokyo, Japan) to A. Muntean during a visit at Meiji University

cross-diffusion and thermo-diffusion. Compare [3] for the thermodynamical founda-
tions of cross- and thermo- diffusion and [9] for a nice paper illustrating the role, of
cross-diffusion mechanisms towards pattern formation in chemical systems. Our own
interest in this framework targets at the fundamental understanding of well-observed
optimal self-organization behaviours (e.g lane formation in counter-flows) exhibited
by the motion of pedestrian flows (cf. e.g. [7] and references cited therein).

Interestingly, due to the symmetry in the structure of the equations, the system (1a)
– (1b) admits a direct interpretation from the porous media theory point of view,
which later turns out to be very useful in understanding mathematically the particle
system origin of this transport problem.

Assume now that $u$ and $v$ denote two populations (of pedestrians, ants, chemical
species, etc.) that like to travel together. Think, for instance, of a pair of large fami-
lies of individuals that wish to reach perhaps a common destination or target, under
the basic assumption that besides some kind of social pairwise repulsion and adher-
ence to the same drift there are no other interactions in the crowd made of the two
populations. This basic situation can be modelled as a system of continuity equa-
tions

$$\partial_t u + \mathrm{div}(uw) = 0,$$
$$\partial_t v + \mathrm{div}(vw) = 0,$$

where $w$ is the common drift to which the two populations adhere. The velocity vector
$w$ is assumed now to comply with Darcy's law

$$w = -\frac{K}{\mu}\nabla p. \tag{2}$$

In (2), $\frac{K}{\mu} \in (0,\infty)$ denotes the permeability coefficient (usually a tensor for a het-
erogeneous region) and $p$ is the total (social) pressure in the system. Now, making
the ansatz on the structure of the pressure

$$p = \mu(u+v),$$

and then summing up the above continuity equations we obtain the system (1a)-(1b),
where for simplicity we take $K \equiv 1$.

The rest of the paper is organised as follows. In Section 2, we provide some analytical
analysis of (1a)-(1b) by transforming the system to an equivalent one, showing the
local well-posedness, constructing a special class of solutions and proving the preserva-
tion of relative entropy and consequences. In Section 3, we introduce a stochastically
many-particle system to approximate (1a)-(1b). Finally, Section 4 presents numerical
illustrations of both the particle system and continuum model.

## 2 Analytical results

In this section, we provide a couple of analytical results on the continuum model. We
first transform the system (1a)-(1b) to an equivalent one. Using this transformation,
we ensure in a straightforward way the local existence of classical solutions. In addi-
tion, we construct a special class of solutions and show remarkable properties of these
solutions, especially concerning the preservation of the relative entropy.

### 2.1 An equivalent system

Define $w := u + v$, then $w$ solves the following equation:

$$\partial_t w = \frac{1}{2}\partial_{xx}(w^2), w(0) = u_0(x) + v_0(x). \tag{3}$$

We transform the system (1a)-(1b) for $(u, v)$ to the following system for $(w, u)$

$$\begin{cases} \partial_t w = \dfrac{1}{2}\partial_{xx}(w^2), & w(0) = w_0(x), \\[2mm] \partial_t u = \partial_x(uw_x), & u(0) = u_0(x). \end{cases} \tag{4a} \tag{4b}$$

Conversely, suppose that $(w, u)$ satisfies the system (4a)-(4b). Then $(u, v)$, where $v = w - u$, satisfies the original system (1a)-(1b). Therefore, the two systems are equivalent.

The transformation has two advantages. First, the new system (4a)-(4b) is only one-sided coupled in the sense that one can solve (4a) independently to obtain $w$, and then substitute to find $u$ from (4b) with $w$ given. Second, (4a) is the famous Boussinesq's equation of groundwater flow, while (4b) is the standard continuity equation. Both equations have been studied extensively and have a rich literature. Therefore, we can apply existing methods and techniques to handle them.

### 2.2 A general solution to the continuity equation by the method of characteristic

Let $V$ be a given velocity field and $f_0 : \mathbb{R} \to \mathbb{R}$ be a given function. We first seek solutions for the following general continuity equation

$$\frac{\partial f}{\partial t} + \frac{\partial}{\partial x}[V(x, t)f] = 0 \text{ with } f(x, 0) = f_0(x). \tag{5}$$

We consider the following ODE

$$\frac{d}{dt}X(t) = V(X(t), t), \quad X(0) = x. \tag{6}$$

The solution of this ODE is $X(t) = F(x, t)$. Conversely, we also can regard $x$ as a function of $X(t)$, i.e., $x = G(X(t), t)$, where $G : \mathbb{R} \times \mathbb{R} \ni (y, t) \mapsto G(y, t) \in \mathbb{R}$ and $G(y, 0) = y$.

**Lemma 1 (Solving the continuity equation, see e.g.[1])** *The function*

$$f(x, t) = f_0(G(x, t))\frac{\partial G}{\partial x}(x, t) = -\frac{f_0(G(x, t))}{v(x, t)}\frac{\partial G}{\partial t}(x, t) \tag{7}$$

*solves the continuity equation* (5).

### 2.3 Classical solutions

The first result of this paper is about local existence of classical solutions of (4a)-(4b). Let $T > 0$ be sufficiently large but fixed and let $(w_0, u_0)$ be given. We say that the couple $(w, u)$, where $w, u : [0, T] \times \mathbb{R} \mapsto \mathbb{R}$ is a classical solution to the system (4a)-(4b) if $w, u \in C^{2,1}([0, T], \mathbb{R})$ and satisfy (4a)-(4b).

**Theorem 1** *Suppose that $w_0$ and $u_0$ are continuous functions in $\mathbb{R}$ with*

$$\varepsilon \leq w_0(x) \leq \frac{1}{\varepsilon},$$

*for some $\varepsilon > 0$ and all $x \in \mathbb{R}$. There exists $T^* \in (0, T)$ such that the system (4a)-(4b) has a classical solution in $C^{2,1}([0, T^*], \mathbb{R})$*

*Proof* This theorem is a direct consequence of [10, Theorem 3.1] for the (global) existence of the Boussinesq's solution and the Peano's theorem for the local existence of the characteristic trajectory.

### 2.4 A special class of solutions

Due to the special structure of the system (4a)-(4b), namely (4a) being the Boussinesq's equation and (4b) being the continuity equation, we are able to construct a special class of solutions. We consider a solution profile of quadratic functions for $w$ and then find $u$ accordingly. The idea of the former has appeared in the literature, see for instance [13].

*Step 1.* We rewrite (4a) as

$$\partial_t w = |\nabla w|^2 + w \Delta w. \tag{8}$$

We consider solutions of (8) of the form

$$w(t, x) = A(t) - B(t) x^2 \text{ where } w(0, x) = w_0(x). \tag{9}$$

By substituting this form in (8), we obtain the following systems

$$\frac{dA}{dt} = -2A(t)B(t), \quad A(0) = w_0(0) =: a,$$

$$\frac{dB}{dt} = -6B(t)^2, \quad B(0) = w_0(1) - w_0(0) =: b.$$

Solving this system yields

$$A(t) = a(6bt + 1)^{-\frac{1}{3}}, \quad B(t) = \frac{b}{6bt + 1}. \tag{10}$$

Therefore,

$$w(t, x) = a(6bt + 1)^{-\frac{1}{3}} - \frac{b}{6bt + 1} x^2. \tag{11}$$

*Step 2.* Substituting (11) back into (4b), we obtain the following continuity equation for $u$:

$$\partial_t u + \partial_x \left[ \frac{2bx}{6bt + 1} u \right] = 0.$$

We now apply Lemma 1 to solve this equation. The ODE (6) becomes

$$\frac{d}{dt}X(t) = \frac{2bX(t)}{6bt+1}, \quad X(0) = x,$$

which gives

$$X(t) = x(6bt+1)^{\frac{1}{3}}, \quad \text{and hence,} \quad G(y,t) = \frac{y}{(6bt+1)^{\frac{1}{3}}}.$$

Therefore, we obtain

$$u(x,t) = u_0(G(x,t))\frac{\partial G}{\partial x}(x,t) = u_0\left(\frac{x}{(6bt+1)^{\frac{1}{3}}}\right)\frac{1}{(6bt+1)^{\frac{1}{3}}}.$$

Concluding, we have established a special solution of the system (4a)-(4b) as follows

$$w(t,x) = a(6bt+1)^{-\frac{1}{3}} - \frac{b}{6bt+1}x^2, \quad \text{and} \quad u(x,t) = u_0\left(\frac{x}{(6bt+1)^{\frac{1}{3}}}\right)\frac{1}{(6bt+1)^{\frac{1}{3}}},$$

for some $a, b \in \mathbb{R}$. If one considers non-negative solutions, one should take the positive parts of these expressions.

### 2.5 Preservation of the relative entropy and consequences

We observe that the system is symmetric in the sense that if we swap $u$ and $v$ in (1a)-(1b) then the system remains unchanged. Therefore, if the initial data $u_0$ and $v_0$ are equal, then it is expected that $u$ and $v$ will be equal at any later time, which is a necessary condition for uniqueness. Two mathematical questions naturally arise:

(i) How to prove this observation rigorously?

(ii) If $u_0$ and $v_0$ are not equal, can we still quantify the distance between $u(t)$ and $v(t)$ in terms of the initial data?

In this section, we provide affirmative answers to these questions using the concept of relative entropy and the total variation metric. We generalise the results of this section (and the previous one) to a more general system in Section 2.6. It will become clear that the structure of the system matches nicely with the notation of the relative entropy.

We now recall the definition of the relative entropy, the total variation metric, as well as he relationship between the two.

Let $f(x)\,dx$ and $g(x)\,dx$ be two probability densities. The relative entropy of $f$ with respect to $g$ is defined by

$$H(f||g) = \int f(x)\log\frac{f(x)}{g(x)}\,dx.$$

The total variation distance between $f(x)\,dx$ and $g(x)\,dx$ is defined simply as the $L^1$-norm of $f - g$, namely

$$TV(f,g) := ||f - g||_{L^1} = \int |f(x) - g(x)|\,dx.$$

Note that the relative entropy is always non-negative and it is equal to 0 if and only if $f = g$. Although it is not a distance (it does not satisfy the triangle equality nor the symmetry condition), it is a useful quantity to measure the difference between two probability measures and has been used extensively in the literature. In addition, it also provides an upper-bound for the total variation distance $TV(f, g)$ by Pinsker's inequality, see for instance [11], as

$$TV(f, g) \leq \sqrt{2H(f||g)}. \tag{12}$$

**Theorem 2** *The function $t \mapsto H(u(t)||v(t))$ is constant, i.e., for any $0 < t < T^*$ it holds*

$$H(u(t)||v(t)) = H(u_0||v_0). \tag{13}$$

*Proof* We calculate the time-derivative of $t \mapsto H(u(t)||v(t))$ as follows (the time variable $t$ is dropped in the right-hand side for simplicity of notation)

$$\begin{aligned}
\frac{d}{dt} H(u(t)||v(t)) &= \frac{d}{dt} \int u \log \frac{u}{v} \, dx \\
&= \int \left[ \left( \log \frac{u}{v} \right) \partial_t u + u \frac{(v \partial_t u - u \partial_t v)/v^2}{u/v} \right] dx \\
&= \int \left[ \left( \log \frac{u}{v} \right) \partial_t u - \frac{u}{v} \partial_t v \right] dx \\
&= \int \left[ \left( \log \frac{u}{v} \right) [u(u+v)_x]_x - \left( \frac{u}{v} \right) [v(u+v)_x]_x \right] dx \\
&= \int (u+v)_x \left[ -u \frac{\partial_x(u/v)}{u/v} + v \partial_x(u/v) \right] dx = 0.
\end{aligned}$$

**Corollary 1** *For any $0 < t < T^*$, it holds that*

$$TV(u(t)||v(t)) \leq \sqrt{2H(u_0||v_0)}.$$

*Proof* This inequality is direct consequence of the Pinsker's inequality (12) and Theorem 2.

**Corollary 2** *Suppose that $u_0 = v_0$, then $u = v = \frac{1}{2} w$. When $w$ solves the equation* (4a), $(u, v)$ *is the unique solution to the system* (1a)-(1b).

Heuristically, note that if $u_0 = v_0$ and $(u, v)$ satisfies the system (1a)-(1b), then so does $(v, u)$. To guarantee the uniqueness, it follows that $u = v = \frac{1}{2} w$. Theorem 2 offers a much stronger result.

*Proof* This is a direct consequence of Theorem 2 (or Corollary 1). Since $u_0 = v_0$, we have $H(u_0||v_0) = 0$. Then it follows from Theorem 2 that $H(u(t)||v(t)) = 0$ for all $t > 0$, which in turn implies that $u(t) = v(t) = \frac{1}{2} w(t)$ for all $t > 0$.

## 2.6 Generalisations

It is worth noting that Theorem 1 and Theorem 2 can be extended to a more general system of the form

$$\begin{cases} \partial_t u = \left[ u \big[ f(u+v) \big]_x \right]_x, & u(0) = u_0(x), \quad (14a) \\[2ex] \partial_t v = \left[ v \big[ f(u+v) \big]_x \right]_x, & v(0) = v_0(x). \quad (14b) \end{cases}$$

The transformed system for $(w, u)$, where $w = (u + v)$, now becomes

$$\begin{cases} \partial_t w = \partial_x(w \partial_x f(w)), & w(0) = w_0(x), \\ \partial_t u = \partial_x(u \partial_x f(w)), & u(0) = u_0(x). \end{cases}$$

For instance, when $f(z) = z^{m-1}$, then the equation for $w$ becomes

$$\partial_t w = \frac{m-1}{m} \partial_{xx}(w^m).$$

This is the standard porous medium equation. Hence, we can apply [10, Theorem 3.1] again; thus Theorem 1 still holds true. We now show that Theorem 2 can also be extended to general shapes of $f$.

**Theorem 3** *Suppose that $u, v$ are classical solutions to the general system* (14a)-(14b). *Then the function $t \mapsto H(u(t) || v(t))$ is constant.*

*Proof* Similar computations as in the proof of Theorem 2 give

$$\frac{d}{dt} H(u(t) || v(t)) = \int \partial_x f(u + v) \Big( -u \frac{\partial_x(u/v)}{u/v} + v \partial_x(u/v) \Big) dx = 0.$$

*Remark 1* We note that the common relation that makes the relative entropies in both Theorem 2 and Theorem 3 vanish is

$$-u \frac{\partial_x(u/v)}{u/v} + v \partial_x(u/v) = 0.$$

Tracing back this relation in the calculations, this property appears because of the combination of three ingredients: the formula of the relative entropy, the symmetry of the system, and the formulas of the continuity equations. The last two properties together form the structure of the system.

i) The continuity equations provide that

$$\partial_t u = \partial_x[uX] \quad \text{and} \quad \partial_t v = \partial_x[vY],$$

ii) The symmetry of (1a) and (1b) means that $X = Y$ (so that if we swap $u$ and $v$, the system is unchanged).

In other words, we find that the relative entropy is constant essentially due to the structure of the system.

## 3 Particle system approach

In this section, we introduce a many-particle system that includes coupled weakly interacting stochastic differential equations. We formally show that the empirical measures associated to this system converge to solutions of the original system (1a)-(1b). The rigorous proof will be given in a separate paper.

We consider the following particle system,

$$dX^{i,\varepsilon} = -\frac{1}{n}\sum_{j=1}^{n}\left(V_{\varepsilon}'(X_t^{j,\varepsilon} - X_t^{i,\varepsilon}) + V_{\varepsilon}'(Y_t^{j,\varepsilon} - X_t^{i,\varepsilon})\right)dt + \varepsilon dW_t^i$$

$$dY^{i,\varepsilon} = -\frac{1}{n}\sum_{j=1}^{n}\left(V_{\varepsilon}'(X_t^{j,\varepsilon} - Y_t^{i,\varepsilon}) + V_{\varepsilon}'(Y_t^{j,\varepsilon} - Y_t^{i,\varepsilon})\right)dt + \varepsilon dW_t^{n+i},$$

(16)

for $i = 1, \ldots, n$, where $\{W^i\}_{i=1}^{2n}$ are independent standard Wiener processes, $V_{\varepsilon}'$ is a sequence of smooth functions which are chosen later on. Note that the system above can be seen as a generalisation of [8] to our model of coupled interactions of two species. In [5,4], the authors studied similar systems but in the absence of the stochastic noise.

We define the following empirical measures

$$u_t^{n,\varepsilon} = \frac{1}{n}\sum_{i=1}^{n}\delta_{X_t^{i,\varepsilon}}, \qquad v_t^{n,\varepsilon} = \frac{1}{n}\sum_{i=1}^{n}\delta_{Y_t^{i,\varepsilon}}.$$

(17)

We now formally derive the system (1a)-(1b) in two steps

**Step** 1: *Hydrodynamic limit*, as $n$ tends to infinity:

$$u_t^{n,\varepsilon} \rightharpoonup u_t^{\varepsilon}, \qquad v_t^{n,\varepsilon} \rightharpoonup v_t^{\varepsilon},$$

where $(u^{\varepsilon}, v^{\varepsilon})$ solves a system which depends on $V_{\varepsilon}'$ and with some viscous terms.

**Step** 2: *Viscosity limit*, as $\varepsilon$ tends to 0:

$$u_t^{\varepsilon} \rightharpoonup u_t, \qquad v_t^{\varepsilon} \rightharpoonup v_t,$$

where $(u, v)$ solves the original system.

The derivation explains the choice of scalings occurring in the many-particle system. Now, we perform the first step.

**Step 1 (Hydrodynamic limit):** Let $f$ be a smooth function. By definition (17) of the empirical measure $u_t^{n,\varepsilon}$, we have

$$\langle f, u_t^{n,\varepsilon}\rangle := \int f(x)u_t^{n,\varepsilon}(dx) = \frac{1}{n}\sum_{i=1}^{n}f(X_t^{i,\varepsilon}).$$

Using Ito's lemma and definitions of the empirical measures in (17), we derive that

$$d\langle f, u_t^{n,\varepsilon}\rangle = \langle -f'V_{\varepsilon}' * (u_t^{n,\varepsilon} + v_t^{n,\varepsilon}) + \frac{\varepsilon^2}{2}f'', u_t^{n,\varepsilon}\rangle dt + \varepsilon\sum_{i=1}^{n}f'(X_t^{i,\varepsilon})dW_t^i,$$

$$d\langle f, v_t^{n,\varepsilon}\rangle = \langle -f'V_{\varepsilon}' * (u_t^{n,\varepsilon} + v_t^{n,\varepsilon}) + \frac{\varepsilon^2}{2}f'', v_t^{n,\varepsilon}\rangle dt + \varepsilon\sum_{i=1}^{n}f'(Y_t^{i,\varepsilon})dW_t^i.$$

By taking the expectation, the Brownian terms vanish, and we obtain that

$$\partial_t\mathbb{E}\langle f, u_t^{n,\varepsilon}\rangle = \mathbb{E}\langle\partial_x\big[u_t^{n,\varepsilon}V_{\varepsilon}' * (u_t^{n,\varepsilon} + v_t^{n,\varepsilon})\big]\rangle + \frac{\varepsilon^2}{2}\partial_{xx}u_t^{n,\varepsilon}, f\rangle,$$

$$\partial_t\mathbb{E}\langle f, v_t^{n,\varepsilon}\rangle = \mathbb{E}\langle\partial_x\big[v_t^{n,\varepsilon}V_{\varepsilon}' * (u_t^{n,\varepsilon} + v_t^{n,\varepsilon})\big]\rangle + \frac{\varepsilon^2}{2}\partial_{xx}v_t^{n,\varepsilon}, f\rangle.$$

Suppose that $u_t^{n,\varepsilon} \xrightarrow{n\to\infty} u_t^{\varepsilon}, v_t^{n,\varepsilon} \xrightarrow{n\to\infty} v_t^{\varepsilon}$ where $u_t^{\varepsilon}$ and $v_t^{\varepsilon}$ are deterministic profiles. Then the pair $(u_t^{\varepsilon}, v_t^{\varepsilon})$ satisfies the following equations

$$\partial_t \langle f, u_t^{n,\varepsilon} \rangle = \langle \partial_x \big[ u_t^{n,\varepsilon} V_{\varepsilon}' * (u_t^{n,\varepsilon} + v_t^{n,\varepsilon}) \big] + \frac{\varepsilon^2}{2} \partial_{xx} u_t^{n,\varepsilon}, f \rangle.$$

$$\partial_t \langle f, v_t^{n,\varepsilon} \rangle = \langle \partial_x \big[ v_t^{n,\varepsilon} V_{\varepsilon}' * (u_t^{n,\varepsilon} + v_t^{n,\varepsilon}) \big] + \frac{\varepsilon^2}{2} \partial_{xx} v_t^{n,\varepsilon}, f \rangle.$$

which are respectively weak formulations of

$$\partial_t u_t^{\varepsilon} = \partial_x \big[ u_t^{\varepsilon} V_{\varepsilon}' * (u_t^{\varepsilon} + v_t^{\varepsilon}) \big] + \frac{\varepsilon^2}{2} \partial_{xx} u_t^{\varepsilon},$$

$$\partial_t v_t^{\varepsilon} = \partial_x \big[ v_t^{\varepsilon} V_{\varepsilon}' * (u_t^{\varepsilon} + v_t^{\varepsilon}) \big] + \frac{\varepsilon^2}{2} \partial_{xx} v_t^{\varepsilon},$$

**Step 2(Viscosity limit):** Assume that $u_t^{\varepsilon} \rightharpoonup u_t,\ v_t^{\varepsilon} \rightharpoonup v_t,\ V_{\varepsilon} \to \delta$ and such that the diffusive terms vanish in the limit $\varepsilon \to 0$, where $(u, v)$ solves the original system. Then we formally get

$$\partial_t u = \partial_x [u \partial_x (u + v)],$$

$$\partial_t v = \partial_x [v \partial_x (u + v)],$$

which is exactly the system (1a)-(1b). To show rigorously the viscosity limit, we rely on the techniques presented in [2]. To keep a concise presentation, we omit to complete the line of the arguments here and postpone them to a forthcoming paper.

# 4 Numerical simulations

In this section, we illustrate numerically in 2D the solution of both the limit continuum problem and of the many-particle system.

## 4.1 Continuum system

We naturally extend the one-dimensional model to two dimensions using the following notations: The populations $u = u(x, y, t)$ and $v = v(x, y, t)$, where $u, v : \Omega \times [0, T^*] \to \mathbb{R}$, with $\Omega \subset \mathbb{R}^2$, satisfy

$$\begin{aligned} \partial_t u &= \nabla \cdot (u \nabla (u + v)) \\ \partial_t v &= \nabla \cdot (v \nabla (u + v)) \end{aligned} \text{ for } (x, y), t \in (0, T^*], \tag{18}$$

with boundary conditions

$$\begin{aligned} n \cdot u \nabla (u + v) &= 0 \\ n \cdot v \nabla (u + v) &= 0 \end{aligned} \text{ on } \partial \Omega, t \in (0, T^*). \tag{19}$$

These boundary conditions ensure a mass-conservative system in combination with the finite volume scheme. We simulate this system with a finite volume discretisation on an equidistant grid where the fluxes adhere to a flux limiter. To approximate $u$

and $v$, we use a first order upwind discretisation. Fluxes are approximated with a second order central difference approximation.

The semi-discrete system of ODE's is non-stiff. We use an explicit integration method to maintain positivity and to put no constraints on the Jacobi matrix. Together with the finite volume space discretisation, this allows for discontinuous initial data. We use a four-stage Runge-Kutta integration scheme to perform the time integration.

### 4.2 Particle system

Recall the particle system formulation from (16). This system is simulated in the same space as (18). As a potential function $V_\varepsilon$, we use

$$V_\varepsilon(r/\varepsilon) = \frac{1}{c\varepsilon^2\sqrt{2\pi}}e^{(-r/(c\varepsilon))^2},\tag{20}$$

where $r \in \mathbb{R}$ represents the interparticle distance, $c$ is the interaction range parameter and $\varepsilon > 0$, modelling a repulsive effect for $r > 0$. This potential formulation is chosen because of its consistency with the potential function description from [8]. In addition, with the stochastic factors included in the system, it is able to model repulsive behaviour equivalent to the continuum system.

Given a particle configuration in time $t$, we use the Euler-Maruyama method (a stochastic variant of the Euler time-integration method) to compute the configuration in $t + \Delta t$. The positions in time step $t_k$ are updated with:

$$\Delta X^{i,\varepsilon} = -\frac{1}{n}\sum_{j=1}^n \left(V_\varepsilon'(X_t^{j,\varepsilon} - X_t^{i,\varepsilon}) + V_\varepsilon'(Y_t^{j,\varepsilon} - X_t^{i,\varepsilon})\right)\Delta t + \varepsilon\sqrt{\Delta t}W_i$$

$$\Delta Y^{i,\varepsilon} = -\frac{1}{n}\sum_{j=1}^n \left(V_\varepsilon'(X_t^{j,\varepsilon} - Y_t^{i,\varepsilon}) + V_\varepsilon'(Y_t^{j,\varepsilon} - Y_t^{i,\varepsilon})\right)\Delta t + \varepsilon\sqrt{\Delta t}W_{n+i},$$

Here $W_i$ are samples of a standard normal distribution. This term emerges from the distribution of the standard Wiener process: $W_t \sim \mathcal{N}(0, t)$. We model conservation of mass by implementing reflective boundaries. With these boundaries, we aim to mimic the zero-flux boundaries in the continuum system.

We impose the conditions $N < 1/(c\varepsilon) < 1/\varepsilon$ so we can observe the limiting effects of the particle system. This way we ensure each particle has enough neighbours in its interaction radius, whilst also balancing interaction forces with stochastic forces. We compute a density or concentration by smoothing the particle positions with a Gaussian kernel $\Phi_h$. This allows us to compare the particle system to its continuum counterpart. This empirical measure $\mu_h(t)$ for particles $X_1, \ldots, X_N$ is defined as $\mu_h(t) = \sum_i^N \delta_{X_t} * \Phi_h$, where $h$ represents the smoothing length of the kernel.

Figure 1 to 9 show the results of the two simulations for various points in time for parameters $N_U = N_V = 1000$, $\varepsilon = 0.3$ and $c = 0.3$. The interpolation kernel has a smoothing length of 0.15. The continuum $w$ is simulated on an equidistant $20 \times 20$ grid. We observe similar diffusive and repulsive behaviour in both simulations.

As long the inequalities between $N$, $c$ and $\varepsilon$ are respected, these similarities are spotted for any $N$ and $t$. It is not immediately obvious how the parameters in the particle simulation should be chosen to mimic the continuum system. However, this varies per phenomenon modelled.



% Simulation software
\chapter{Simulation software: \emph{Mercurial}}
\label{chap:mercurial}
This chapter provides some details on \emph{Mercurial}, the framework developed to perform the simulations in this thesis. It enables users to create inhomogeneous domains and model the transport of particles in these domains on multiple scales. It also provides various options for visualisations and gathering results.
We discuss the programming environment \emph{Mercurial} was developed in, the architectural structure of the framework and highlight some of the features.
The implementation is open source and hosted on GitHub. It can be found on \texttt{https://github.com/0mar/mercurial} or obtained by contacting the author or the supervisors of this thesis.

\section{Programming environment}
\emph{Mercurial} is programmed in Python, a high-level interpreted object-oriented programming language. 
Python is a popular language in scientific computing due to its well-maintained numerical libraries like \texttt{numpy} and \texttt{scipy}, and due to its simple syntax compared to other programming languages.
Its ease of implementation lends itself well for building prototypes.
In addition, since the language and its libraries are open source, it allows us to make \emph{Mercurial} framework open-source as well.

Python allows for rapid prototyping of applications, and is easily extensible with different features. 
However, because Python is an interpreted languages, some operations are very slow in comparison to other (compiled) languages.
This overhead can be circumvented by \worddef{extending} Python with modules built in C/C++ or FORTRAN, two relatively low-level compiled languages.
The \texttt{numpy}/\texttt{scipy} stack are examples of such modules.
For this reason, most of the computationally intensive parts are (whenever possible) performed using the \texttt{numpy}/\texttt{scipy} libraries, whose performance is comparable to pure C/C++ and generally exceeds MATLAB performance.

Not all operations fit in the framework of these third-party libraries. If this is the case for some operation that becomes a bottleneck when executed in Python, we implement the operation in FORTRAN.
Using the tool \texttt{f2py} we are able to convert FORTRAN subroutines to compiled Python libraries. 
A big advantage of this work flow is that this FORTRAN code is able to directly interface with objects from \texttt{numpy}, so no data conversion is necessary.

\subsection{External modules}
Apart from the \texttt{numpy} and \texttt{scipy} modules, we use the following external libraries:
\begin{itemize}
    \item \texttt{matplotlib} (\cite{matplotlib}) modules to provide plots of the macroscopic representations of the current state of a simulation.
    \item \texttt{networkx} (\cite{networkx}) to perform operations related to graph creation and manipulation.
    \item \texttt{cvxopt} (\cite{cvxopt}) to solve quadratic problems
    \item \texttt{nosetests} to build and run unit tests.
\end{itemize}
In addition, we created the following FORTRAN modules which are used in the framework:
\begin{itemize}
    \item \texttt{mde} to count particles violating a minimum distance and, if required, correct their positions.
    \item \texttt{micro\_macro} to convert positions and velocities to a density and velocity field using a kernel interpolant.
    \item \texttt{compute\_pressure} to solve a continuity equation augmented with Darcy's law.
    \item \texttt{pgs} to solve a linear complementary problem using a projected Gauss-Seidel method on sparse matrices (using algorithms from \cite{saad03}).
\end{itemize}

\section{Features}
\subsubsection{High performance for large particles}
\emph{Mercurial} is built to handle many particles simultaneously. To this end, the code uses vectorised operations wherever possible. 
In addition, it supports imposing particle interaction on a macroscopic scale, creating a hybrid simulation conforming to the concepts in Chapter~\ref{chap:analysis} and Chapter~\ref{sec:micro_macro}.
This way, the time-expensive microscopic particle interaction is avoided.

\subsubsection{Drawing tool}
A tool is included for drawing rectangular geometries. 
In these geometries, the user can specify particle inflow boundaries, outflow boundaries, and impermeable areas. 
An example of such a domain is illustrated in Figure~\ref{fig:example_domain}.

\begin{figure}[h!]
    \centering
    \includegraphics[width=0.7\textwidth]{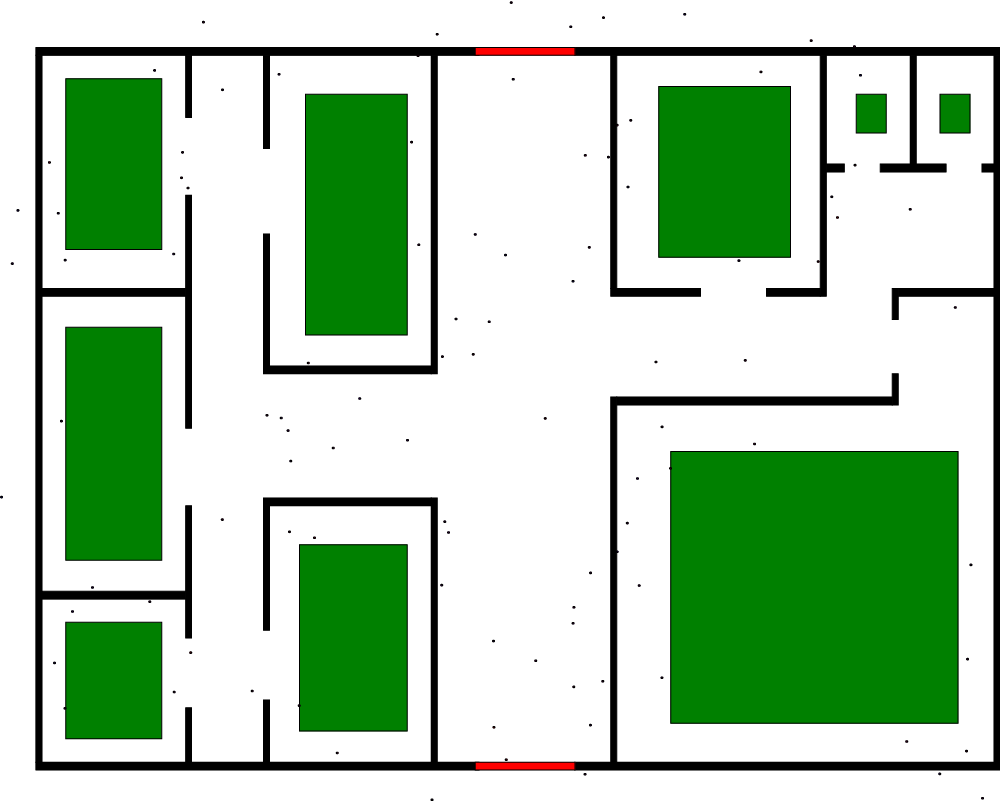}
    \caption{A domain used in \emph{Mercurial}. Entrances are indicated in green. Exits are red, generic obstacles are black.}
    \label{fig:example_domain}
\end{figure}

\subsubsection{Reproducibility and quantification}
The user specifies two files in a simulation: a configuration file containing all parameters and settings of the particle system, and a scene file containing the descriptions of the geometry, including obstacles, inflow, and outflow conditions.

If required, the user can extract and store output in a file with a data format both readable by Python and MATLAB.
The framework comes with a default result processing tool that outputs several graphs on simulated data (depicted in Section~\ref{sec:results} and Section~\ref{sec:results2}).
This way, it is easy to keep track of simulation runs, rerun the same simulation later in time, or extract different data from an earlier run simulation.

\subsubsection{Extensibility}
The framework is composed of several modules where interface and implementation are separated as much as possible. This modularity makes it easier to improve existing parts of the simulation without breaking the framework. 
In addition, it keeps the modifications to a minimum when one attempts to model a different system.

The object-oriented nature for the framework allows easy specification of individual particle traits.
We have separate visualisation modules for the microscopic and macroscopic representations. 
Both modules focus on providing a simple, lightweight graphical overview of the simulation. 
They can be suppressed to devote all computational resources to compute the evolution of the system.

The framework lacks a decent graphical user interface. All parameters and settings are set via the command line. This could be extended by a more user-friendly interface in future versions.

\section{Particle model implementation}
This is a continuous space, discrete time model. 
This means particles are able to assume any position $\vec{x}\in\Omega$, but their position in time is restricted to values in discrete time domain $\left\{ k \Delta t | k\in \mathbb{N} \right\}$ for some time step $\Delta t>0$.
We implement each particle as a separate object, which allows us to easily add and modify of individual properties. 
However, positions, velocities and maximum speeds of all particles are stored collectively in vectors to take advantage of vectorised operations.

\subsection{Geometry}
We assume $\Omega$ is rectangular. Parts of the domain can be made impermeable to particles. These parts are called \worddef{obstacles}. 
Should a particle collide with an obstacle, its motion is stopped at the edge.

Inflow is modelled with \emph{entrances}, while outflow is modelled with \emph{exits}.
\paragraph{Entrances}
We impose an inflow by modelling the number of new particles with a Poisson process depending on inflow parameter $\lambda$ which can be specified for each entrance.
For any given time $\Delta t$, the number of particles $N_I(\Delta t)$ entering the scene in that time has a probability distribution given by
\begin{equation}
    P\left( N_i=n \right) = \frac{(\lambda\Delta t)^n}{n}e^{-\lambda\Delta t}.
    \label{eq:inflow_prob}
\end{equation}
The total number of particles spawning from one obstacle can be limited to model finite-capacity sources, or to couple multiple simulation domains together.
\paragraph{Exits}
Exits are rectangular areas that regulate the outflow. Any particle that reaches the exit is removed from the scene.
By default, no conditions are imposed on the outflow. The particle motion and interaction should naturally incur outflow, and since the particle speed and the exit sizes are limited, so is the outflow.
We do provide the possibility to explicitly limit the outflow with a maximum value. This might be desirable when the exit represents a bottleneck.
In that case, the number of particles $N_{O,\Delta t}$ moving out of the scene in time $\Delta t$ has to satisfy
\begin{equation}
    N_{O,\Delta t} \leq c_O\Delta t,
    \label{eq:outflow_rate}
\end{equation}
where $c_O>0$ represents the outflow rate.
\subsection{Particles}
Particles are represented as moving discs with radius $r>0$, velocity $\vec{v}(t) \in \mathbb{R}^2$ and position $\vec{x}(t) \in\mathbb{R}^2$. 
Each time step, the particle positions are updated with time step $\Delta t$ using Euler integration:
\begin{equation}
    \vec{x}(t+\Delta t) = \vec{x}(t) + \vec{v}(t)\Delta t.
    \label{}
\end{equation}
The user is free in prescribing particle velocities, as it is essential to the nature of the simulation.
\emph{Mercurial} has two planner modules determining particle velocities on each time step, based on macroscopic interaction and the shape of the geometry. These planners are written with the goal of simulating crowds and are discussed in Chapter~\ref{chap:crowds}.

\subsection{Initial condition}
The default initial condition is a random distribution of particles, in accordance to \eqref{eq:langevin_}.
Particles can be spawned in all locations not occupied by obstacles, entrances or exits. 
The user is also able to create specific initial distributions. 
Distributions that are included by default are localised dense particle configurations (used in Section~\ref{sec:case_e}) and partly filled domains (used in Section~\ref{sec:case_f}).

\section{Implemented methods}
We provide a brief summary of some of the methods implemented in \emph{Mercurial}. This is by no means a complete list, but meant to provide some insight in the contents of the framework.
\subsection{Geometry-related methods}
\begin{itemize}
    \item Load geometry file
    \item Load configuration file
    \item Move pedestrians
    \item Apply boundary conditions; inflow/outflow, obstacles.
    \item Compute minimum distance violations.
\end{itemize}

\subsection{Static planner}
\begin{itemize}
    \item Create visibility graph
    \item Find shortest paths
    \item Find line segment intersections
\end{itemize}

\subsection{Domain/interaction potential}
\begin{itemize}
    \item Find scene/boundary interface
    \item Compute unit cost field
    \item Compute potential field
    \item Compute pressure
    \item Apply pressure
\end{itemize}

\subsection{Macro-micro interaction}
\begin{itemize}
    \item Interpolate particle values
    \item Compute bilinear interpolation functions
    \item Plot macroscopic values
\end{itemize}

\subsection{Results}
\begin{itemize}
    \item Create time spent histogram
    \item Create minimum path length histogram
    \item Create scatter plot
    \item Plot density heatmap
    \item Plot minimum distance violations
\end{itemize}

\end{appendices}
\end{document}